\documentclass[10pt]{article}
\usepackage{caption}

\usepackage[dvipsnames]{xcolor}
\usepackage{framed}
\definecolor{shadecolor}{rgb}{0.9,0.9,0.9}

\usepackage{mathtools}
\usepackage{amsmath}
\usepackage[shortlabels]{enumitem}
\usepackage{graphicx,epic,eepic,epsfig,amsmath,latexsym,amssymb,verbatim,color}
 
\usepackage{amsfonts}       % blackboard math symbols
\usepackage{nicefrac}       % compact symbols for 1/2, etc.

\usepackage{amsmath}
\usepackage{bbm}

\usepackage{float}
\usepackage{tikz}
\usetikzlibrary{chains}
\usetikzlibrary{fit}
\usetikzlibrary{arrows}
\usetikzlibrary{decorations}

\usepackage{epsfig}
\usetikzlibrary{shapes.symbols,patterns} % for source symbols
\usepackage{pgfplots}

\usepackage[strict]{changepage}
\usepackage{hyperref}
\hypersetup{colorlinks=true,citecolor=blue,linkcolor=blue,filecolor=blue,urlcolor=blue,breaklinks=true}

\usepackage[marginal]{footmisc}
\usepackage{url}
\usepackage{theorem}
\newtheorem{definition}{Definition}
\newtheorem{proposition}{Proposition}
\newtheorem{lemma}{Lemma}

\newtheorem{theorem}[proposition]{Theorem}
\newtheorem{corollary}[proposition]{Corollary}

\def\squareforqed{\hbox{\rlap{$\sqcap$}$\sqcup$}}
\def\qed{\ifmmode\squareforqed\else{\unskip\nobreak\hfil
\penalty50\hskip1em\null\nobreak\hfil\squareforqed
\parfillskip=0pt\finalhyphendemerits=0\endgraf}\fi}
\def\endenv{\ifmmode\;\else{\unskip\nobreak\hfil
\penalty50\hskip1em\null\nobreak\hfil\;
\parfillskip=0pt\finalhyphendemerits=0\endgraf}\fi}
\newenvironment{proof}{\noindent \textbf{{Proof~} }}{\hfill $\blacksquare$}

\newcounter{remark}

\newcounter{example}

\mathchardef\ordinarycolon\mathcode`\:
\mathcode`\:=\string"8000
\def\vcentcolon{\mathrel{\mathop\ordinarycolon}}
\begingroup \catcode`\:=\active
  \lowercase{\endgroup
  \let :\vcentcolon
  }

\usepackage{cleveref}
\usepackage{graphicx}
\usepackage{xcolor}

\RequirePackage[framemethod=default]{mdframed}
\newmdenv[skipabove=7pt,
skipbelow=7pt,
backgroundcolor=darkblue!15,
innerleftmargin=5pt,
innerrightmargin=5pt,
innertopmargin=5pt,
leftmargin=0cm,
rightmargin=0cm,
innerbottommargin=5pt,
linewidth=1pt]{tBox}

\newmdenv[skipabove=7pt,
skipbelow=7pt,
backgroundcolor=red!15,
innerleftmargin=5pt,
innerrightmargin=5pt,
innertopmargin=5pt,
leftmargin=0cm,
rightmargin=0cm,
innerbottommargin=5pt,
linewidth=1pt]{rBox}

\newmdenv[skipabove=7pt,
skipbelow=7pt,
backgroundcolor=blue2!25,
innerleftmargin=5pt,
innerrightmargin=5pt,
innertopmargin=5pt,
leftmargin=0cm,
rightmargin=0cm,
innerbottommargin=5pt,
linewidth=1pt]{dBox}
\newmdenv[skipabove=7pt,
skipbelow=7pt,
backgroundcolor=darkkblue!15,
innerleftmargin=5pt,
innerrightmargin=5pt,
innertopmargin=5pt,
leftmargin=0cm,
rightmargin=0cm,
innerbottommargin=5pt,
linewidth=1pt]{sBox}
\definecolor{darkblue}{RGB}{0,76,156}
\definecolor{darkkblue}{RGB}{0,0,153}
\definecolor{blue2}{RGB}{102,178,255}
\definecolor{darkred}{RGB}{195,0,0}
\newcommand{\nc}{\newcommand}
\nc{\rnc}{\renewcommand}
\nc{\lbar}[1]{\overline{#1}}
\nc{\bra}[1]{\langle#1|}
\nc{\ket}[1]{|#1\rangle}
\nc{\ketbra}[2]{|#1\rangle\!\langle#2|}
\nc{\braket}[2]{\langle#1|#2\rangle}

\nc{\proj}[1]{| #1\rangle\!\langle #1 |}
\nc{\avg}[1]{\langle#1\rangle}
\nc{\liebracket}[1]{\left\langle #1\right\rangle_{\Op{Lie}}}
\nc{\smfrac}[2]{\mbox{$\frac{#1}{#2}$}}
\nc{\tr}{\operatorname{Tr}}
\nc{\ox}{\otimes}
\nc{\dg}{\dagger}
\nc{\dn}{\downarrow}
\nc{\cA}{{\cal A}}
\nc{\cB}{{\cal B}}
\nc{\cC}{{\cal C}}
\nc{\cD}{{\cal D}}
\nc{\cE}{{\cal E}}
\nc{\cF}{{\cal F}}
\nc{\cG}{{\cal G}}
\nc{\cH}{{\cal H}}
\nc{\cI}{{\cal I}}
\nc{\cJ}{{\cal J}}
\nc{\cK}{{\cal K}}
\nc{\cL}{{\cal L}}
\nc{\cM}{{\cal M}}
\nc{\cN}{{\cal N}}
\nc{\cO}{{\cal O}}
\nc{\cP}{{\cal P}}
\nc{\cQ}{{\cal Q}}
\nc{\cR}{{\cal R}}
\nc{\cS}{{\cal S}}
\nc{\cT}{{\cal T}}
\nc{\cU}{{\cal U}}
\nc{\cV}{{\cal V}}
\nc{\cX}{{\cal X}}
\nc{\cY}{{\cal Y}}
\nc{\cZ}{{\cal Z}}
\nc{\cW}{{\cal W}}
\nc{\csupp}{{\operatorname{csupp}}}
\nc{\qsupp}{{\operatorname{qsupp}}}
\nc{\var}{{\operatorname{var}}}
\nc{\rar}{\rightarrow}
\nc{\lrar}{\longrightarrow}
\nc{\polylog}{{\operatorname{polylog}}}
\nc{\wt}{{\operatorname{wt}}}
\nc{\av}[1]{{\left\langle {#1} \right\rangle}}
\nc{\supp}{{\operatorname{supp}}}
\nc{\spn}{{\operatorname{span}}}

\nc{\argmin}{{\operatorname{argmin}}}

\nc{\Adj}[1]{\Op{Ad}_{#1}}
\nc{\ad}{{\operatorname{ad}}}
\nc{\Ad}{{\operatorname{Ad}}}

\nc{\liedqnn}{{\frak{g}_{\Op{QRENN}}}}

\nc{\polysub}{{\cC_{\Op{polySub}}}}
\nc{\BPbar}{{\overline{\Op{polySub}}}}

\def\a{\alpha}

\def\x{\xi}

\def\O{\Omega}

\nc{\RR}{{{\mathbb R}}}
\nc{\CC}{{{\mathbb C}}}
\nc{\FF}{{{\mathbb F}}}
\nc{\NN}{{{\mathbb N}}}
\nc{\ZZ}{{{\mathbb Z}}}
\nc{\PP}{{{\mathbb P}}}
\nc{\QQ}{{{\mathbb Q}}}
\nc{\UU}{{{\mathbb U}}}
\nc{\EE}{{{\mathbb E}}}
\nc{\id}{{\operatorname{id}}}

\nc{\CHSH}{{\operatorname{CHSH}}}

\newcommand{\Op}{\operatorname}
\newcommand{\Var}{\Op{Var}}

\nc{\be}{\begin{equation}}
\nc{\ee}{{\end{equation}}}
\nc{\bea}{\begin{eqnarray}}
\nc{\eea}{\end{eqnarray}}
\nc{\<}{\langle}
\rnc{\>}{\rangle}
\nc{\rU}{\mbox{U}}

\nc{\ob}[1]{#1}
\nc{\norm}[1]{\left\|{#1}\right\|}
\nc{\adj}[1]{\Op{ad}_{#1}}

\nc{\SEP}{{\text{\rm SEP}}}
\nc{\NS}{{\text{\rm NS}}}
\nc{\LOCC}{{\text{\rm LOCC}}}
\nc{\PPT}{{\text{\rm PPT}}}
\nc{\EXT}{{\text{\rm EXT}}}
\nc{\Sym}{{\operatorname{Sym}}}

\nc{\ERLO}{{E_{\text{r,LO}}}}
\nc{\ERLOCC}{{E_{\text{r,LOCC}}}}
\nc{\ERPPT}{{E_{\text{r,PPT}}}}
\nc{\ERLOCCinfty}{{E^{\infty}_{\text{r,LOCC}}}}
\nc{\Aram}{{\operatorname{\sf A}}}

\usepackage{tikz}
\usepackage{hyperref}
\hypersetup{colorlinks=true,citecolor=blue,linkcolor=blue,filecolor=blue,urlcolor=blue,breaklinks=true}

\makeatletter
\def\grd@save@target#1{%
  \def\grd@target{#1}}
\def\grd@save@start#1{%
  \def\grd@start{#1}}
\tikzset{
  grid with coordinates/.style={
    to path={%
      \pgfextra{%
        \edef\grd@@target{(\tikztotarget)}%
        \tikz@scan@one@point\grd@save@target\grd@@target\relax
        \edef\grd@@start{(\tikztostart)}%
        \tikz@scan@one@point\grd@save@start\grd@@start\relax
        \draw[minor help lines,magenta] (\tikztostart) grid (\tikztotarget);
        \draw[major help lines] (\tikztostart) grid (\tikztotarget);
        \grd@start
        \pgfmathsetmacro{\grd@xa}{\the\pgf@x/1cm}
        \pgfmathsetmacro{\grd@ya}{\the\pgf@y/1cm}
        \grd@target
        \pgfmathsetmacro{\grd@xb}{\the\pgf@x/1cm}
        \pgfmathsetmacro{\grd@yb}{\the\pgf@y/1cm}
        \pgfmathsetmacro{\grd@xc}{\grd@xa + \pgfkeysvalueof{/tikz/grid with coordinates/major step}}
        \pgfmathsetmacro{\grd@yc}{\grd@ya + \pgfkeysvalueof{/tikz/grid with coordinates/major step}}
        \foreach \x in {\grd@xa,\grd@xc,...,\grd@xb}
        \node[anchor=north] at (\x,\grd@ya) {\pgfmathprintnumber{\x}};
        \foreach \y in {\grd@ya,\grd@yc,...,\grd@yb}
        \node[anchor=east] at (\grd@xa,\y) {\pgfmathprintnumber{\y}};
      }
    }
  },
  minor help lines/.style={
    help lines,
    step=\pgfkeysvalueof{/tikz/grid with coordinates/minor step}
  },
  major help lines/.style={
    help lines,
    line width=\pgfkeysvalueof{/tikz/grid with coordinates/major line width},
    step=\pgfkeysvalueof{/tikz/grid with coordinates/major step}
  },
  grid with coordinates/.cd,
  minor step/.initial=.2,
  major step/.initial=1,
  major line width/.initial=2pt,
}
\makeatother

\usepackage{thmtools}
\usepackage{thm-restate}
\usepackage{etoolbox}
\makeatletter
\def\problem@s{}
\newcounter{problems@cnt}

\newcommand{\allproblems}{\problem@s}
\makeatother

\usepackage{xcolor}
\pgfplotsset{compat=1.18}
\usepackage{amsmath,amsfonts,amssymb,array,graphicx,mathtools,multirow,bm,times,tcolorbox,relsize,booktabs,subfigure}
\usepackage{ulem}
\usepackage[utf8]{inputenc}
\usepackage[T1]{fontenc}
\usepackage{authblk} % package for author list
\usepackage[paperwidth=199.8mm,paperheight=297mm,centering,hmargin=20mm,vmargin=2cm]{geometry}
\usepackage[ruled,vlined,linesnumbered]{algorithm2e}
\usepackage{quantikz}
\usepackage{diagbox}

\definecolor{colortwo}{rgb}{0.4,0.77,0.17}
\definecolor{colorthree}{rgb}{0.01,0.51,0.93}
\newcommand{\update}[1]{\textcolor{magenta}{#1}}

\newcommand{\hardware}{\textcolor{black}{\textit{AMD EPYC 7542 32-Core Processor (RAM 512GB)}}}
\allowdisplaybreaks

\begin{document}
\title{Quantum Recurrent Embedding Neural Network}

\author[1]{Mingrui Jing\thanks{mjing638@connect.hkust-gz.edu.cn}}
\author[1]{Erdong Huang\thanks{ehuang652@connect.hkust-gz.edu.cn}}
\author[1]{Xiao Shi\thanks{xiaoshi@hkust-gz.edu.cn}}
% \affiliation{Thrust of Artificial Intelligence, Information Hub,
% The Hong Kong University of Science and Technology (Guangzhou), Guangdong 511453, China}

\author[2]{Shengyu Zhang\thanks{shengyzhang@tencent.com}}
% \affiliation{Tencent Quantum Laboratory}

\author[1]{Xin Wang\thanks{felixxinwang@hkust-gz.edu.cn}}

\affil[1]{\small Thrust of Artificial Intelligence, Information Hub,\par The Hong Kong University of Science and Technology (Guangzhou)}

\affil[2]{\small Tencent Quantum Laboratory}

\date{\today}
\maketitle

\begin{abstract}
Quantum neural networks have emerged as promising quantum machine learning models, leveraging the properties of quantum systems and classical optimization to solve complex problems in physics and beyond. However, previous studies have demonstrated inevitable trainability issues that severely limit their capabilities in the large-scale regime. In this work, we propose a quantum recurrent embedding neural network (QRENN) inspired by fast-track information pathways in ResNet and general quantum circuit architectures in quantum information theory. By employing dynamical Lie algebras, we provide a rigorous proof of the trainability of QRENN circuits, demonstrating that this deep quantum neural network can avoid barren plateaus. Notably, the general QRENN architecture resists classical simulation as it encompasses powerful quantum circuits such as QSP, QSVT, and DQC1, which are widely believed to be classically intractable. Building on this theoretical foundation, we apply our QRENN to accurately classify quantum Hamiltonians and detect symmetry-protected topological phases, demonstrating its applicability in quantum supervised learning. Our results highlight the power of recurrent data embedding in quantum neural networks and the potential for scalable quantum supervised learning in predicting physical properties and solving complex problems.
\end{abstract}

\tableofcontents

\newpage
%%%%%%%%%%%%%%%%%%%%%%%%%%%%%%%%%%%%%%%%%%%%%%%%%%%%%%%%%%%%%%%%%%%%%%%%%%%

\section{Introduction}
Quantum technology offers the potential to tackle computational challenges once deemed intractable on classical devices, opening new avenues in optimization, simulation, and data analysis. In recent years, quantum processors have demonstrated clear quantum advantages in specialized tasks~\cite{riste2017demonstration,bravyi2018quantum, bravyi2020quantum, huang2022quantum,daley2022practical,acharya2024quantum}, underscoring the rapidly expanding horizon for advanced computations and fostering growing interest in merging quantum methods with artificial intelligence. This synergy has led to the emergence of \textit{quantum machine learning} (QML)~\cite{schuld2015introduction,biamonte2017quantum,zhang2020toward,li2022recent,qian2022dilemma,caro2022generalization,cerezo2022challenges,tian2023recent,jerbi2023quantum}, wherein classical machine learning techniques are augmented by quantum phenomena such as superposition and entanglement, advancing data processing~\cite{li2021vsql,Yu2022power} and pattern recognition~\cite{liu2022representation,senokosov2024quantum} in ways that may unlock unprecedented performance gains.

A foundational element of QML is the use of \textit{quantum neural networks} (QNNs)~\cite{benedetti2019parameterized,ostaszewski2021structure}, wherein quantum gates with tunable parameters are iteratively optimized to carry out specific computational tasks. Although modern noisy intermediate-scale quantum (NISQ) technologies~\cite{Preskill2018quantum} do not yet support fully universal quantum computation, they support \textit{variational quantum algorithms} (VQAs)~\cite{cerezo2021variational,bittel2021training}, effectively coupling quantum hardware with classical optimizers in a feedback loop. These VQAs have shown promise in domains such as quantum chemistry~\cite{bauer2020quantum,hempel2018quantum}, combinatorial optimization~\cite{farhi2014quantum,zhou2020quantum}, and machine learning~\cite{chen2020variational,takaki2021learning,zhang2024statistical}, and they continue to evolve alongside advances in circuit and hardware design.

Despite significant progress, training deeper or more expressive quantum neural networks (QNNs) continues to be hindered by the phenomenon of \textit{barren plateaus} (BP)—regions of parameter space where gradients vanish, making optimization inefficient~\cite{Larocca2025,Mcclean2018barren}. This gradient suppression poses a major bottleneck for scaling quantum machine learning (QML) models to larger sizes. The barren plateau problem in QNNs bears similarities to the vanishing gradient problem encountered in classical deep neural networks (DNNs)~\cite{Sze2017efficient}, where gradients can exponentially diminish as they backpropagate through many layers, leading to stagnated training. Classical BPs were largely overcome by innovations such as residual neural networks (ResNets)~\cite{He2016deep}, which introduce skip connections to facilitate gradient flow and enable the training of very deep architectures. In contrast, quantum BP arises due to the exponential concentration of measure in high-dimensional Hilbert spaces and the expressiveness of quantum circuits, making it both analogous to and fundamentally distinct from its classical counterpart.

The sources of BP in QNNs are diverse, including circuit expressiveness~\cite{ortiz2021entanglement,zhao2021analyzing,pesah2021absence,holmes2022connecting}, initialization and measurement choices~\cite{Cerezo2021cost,abbas2021power,holmes2021barren}, and quantum noise~\cite{wang2021noise,schumann2024emergence,garcia2024effects}. To address these challenges, recent efforts have introduced refined initialization strategies, alternative training heuristics, and QNN architectures specifically designed to suppress BPs~\cite{Zhang2022,wang2024trainability,shi2024avoiding}. Examples include hardware-efficient ansatzes tailored to device constraints, and problem-informed designs like the Hamiltonian variational ansatz~\cite{park2024hamiltonian,wiersema2020exploring} and QAOA~\cite{farhi2014quantum,zhou2020quantum}. A unifying theoretical foundation for these approaches has emerged through the lens of dynamical Lie algebras (DLA)~\cite{Fontana2024characterizing,Ragone2024lie}, where the gradient variance of trace-form loss functions is governed by the DLA structure and the projections of the initial state and measurement operator onto its components. This perspective connects various BP phenomena and aligns with recent Lie-algebraic trainability analyses in structured QML settings~\cite{Allcock2024dynamical,kazi2024analyzing,Mao2024towards}.

\begin{figure}[h!]
    \centering
    \includegraphics[width=0.9\linewidth]{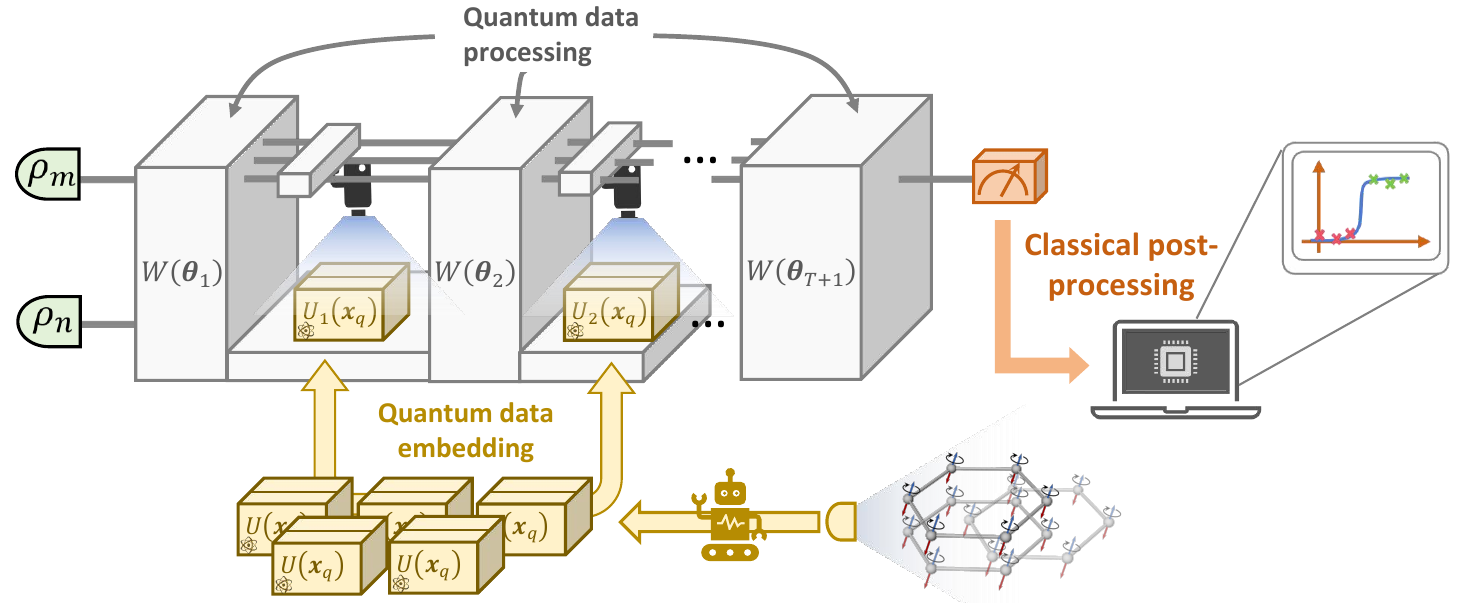}
    \caption{\textbf{The general framework of quantum supervised learning from a view of quantum circuit architecture.} The quantum data, represented as physical evolution $U_l(\bm{x}_q)$, can be repeatedly embedded into the processing blocks made of QNN layers $W(\bm{\theta}_t)$. Measurements acting on a few systems followed by classical post-processing are required in order to deliver accurate predictions.}
    \label{fig:main_fig}
\end{figure}

Nonetheless, many investigations into QNNs' design restrict their generators to individual or carefully chosen Pauli matrices within a periodic circuit framework, with little algorithmic and architectural concerns. Besides, existing design paradigms lack the proper integration and utilization of data, either classical or quantum, thereby overlooking its impact on the trainability of QNNs and their capacity to perform tasks, even though the data reuploading~\cite{perez2020data} has already demonstrated significant advantages in various computational tasks within the realm of quantum supervised learning~\cite{Havlivcek2019supervised,Schuld2021supervised}.

{In this work, we introduce the quantum recurrent embedding neural network (QRENN), a versatile framework grounded in quantum circuit architecture theory~\cite{Chiribella2008} that tightly integrates tunable QNN components with quantum data embedding. We propose a novel training strategy inspired by quantum hypothesis testing~\cite{Chen2024hypothesis,Regula2023postselected}, and analyze trainability through the DLA formalism. Our theoretical analysis proves that QRENN avoids barren plateaus under control embeddings, due to the gentle expansion of the Lie algebra via a direct sum structure with bounded dimensions. Importantly, QRENN resists dequantization arguments such as those proposed by Cerezo et. al.~\cite{Cerezo2023does}, as it naturally encompasses classically intractable primitives such as QSP and QSVT~\cite{Gilyen2019quantum,Low2017optimal,Rossi2022multivariable}. These connections underscore the model's robustness against classical surrogates, while maintaining trainability under realistic assumptions.}

{We validate our theoretical claims with numerical experiments showing polynomially decaying gradient variance under random initialization. We then apply QRENN to quantum supervised learning tasks, including Hamiltonian classification and symmetry-protected topological (SPT) phase detection. Our results demonstrate high accuracy and strong robustness against classical and quantum noise. Taken together, these findings highlight the power of data-aware QNN architectures and suggest that QRENN offers a scalable and trainable path forward for quantum machine learning and quantum many-body physics.}

\section{Preliminaries}
\paragraph{Quantum neural networks (QNNs).}
{Quantum neural networks (QNNs)~\cite{gupta2001quantum} have emerged as a promising subclass of variational quantum algorithms (VQAs)~\cite{cerezo2021variational}, standing at the intersection of quantum computing and artificial intelligence, with their potential to harness quantum advantages in learning and inference tasks~\cite{abbas2021power}. A general QNN model is composed of three essential components: a designed input quantum state $\rho$ and an observable $O$, the data embedding register through the set of unitary $\{U_l\}_l$ and the data processing register through parameterized circuit template~\cite{Schuld2020} $W(\bm{\theta})$. The total QNN can be represented as,
\begin{equation}
\bm{U}(\bm{x};\bm{\theta}) = \prod_l U_l(\bm{x}) W(\bm{\theta}_l),
\end{equation}
where $\bm{x}$ is the data embedded in the $l$-th layer and $\bm{\theta}_l$ represents the tunable parameters in the $l$-th processing layer. The QNN is then optimized for a particular task, analogous to the classical machine learning models~\cite{Dunjko2018} and the final output is extracted from the measurements on the state $\bm{U}(\bm{x};\bm{\theta})\rho \bm{U}(\bm{x};\bm{\theta})^{\dagger}$ with respect to Hermitian $O$. These measurements, as the predictions from the model, are often classically post-processed and passed to a loss function, and the goal is to minimize the loss via varying $\bm{\theta}$.}

\paragraph{Barren plateau.}
A major challenge in training QNNs is the phenomenon of \textit{barren plateaus}~\cite{Mcclean2018barren,holmes2021barren}, where the gradient, of high probability, can vanish exponentially with system size, 
which induces a tremendous growth in the complexity of accurately evaluating the gradient values. Besides, exponentially small gradient magnitudes make it difficult to accurately determine the correct optimization direction, leading to a challenge for assessing convergence during the practical optimization process. 
% \update{(to check and complement: exponentially small convergence rate by BP?)} 
Given a fixed pair of an input quantum state $\rho$ and an observable operator $O$, the trace form loss function $\ell(\bm{x};\bm{\theta}) = \tr(\bm{U}(\bm{x};\bm{\theta}) \rho \bm{U}(\bm{x};\bm{\theta})^{\dagger} O)$ is often used to compose the total loss in QML. The structure of $\bm{U}(\bm{x};\bm{\theta})$ directly affects the behaviors during training.

Formally, we say that the QNN experiences a \textit{Barren Plateaus} (BP) if the variance of the loss function gradient with respect to any parameter $\theta_\mu$ in the model decays exponentially with the system size~\cite{Mcclean2018barren}, i.e.,
\begin{equation*}
    \Var_{\bm{\theta}\sim\nu}[\partial_{\theta_{\mu}}\ell(\bm{x};\bm{\theta})]\in \cO\left(\frac{1}{b^n}\right),
\end{equation*}
where $n$ denotes the number of qubits in the system and $b>1$. Typically, $\nu$ is the uniform distribution over the range of the parameters. We view the QNN as \textit{untrainable} if it experiences a BP regarding a random initialization. Conversely, if the variance of its gradient remains polynomially bounded from below with respect to the number of qubits (i.e., in the order of $\Omega(1/\Op{poly}(n))$), then the model does not experience BP, one of the biggest challenges in training of parametrized quantum circuits. 

\paragraph{Dynamical Lie algebras.} \textit{Dynamical Lie algebra} (DLA)~\cite{Alessandro2008,larocca2022diagnosing} was originally introduced in the field of quantum control theory. In recent years, it has gained significant traction in the study of VQAs' trainability issues, particularly in addressing challenges related to BP~\cite{Fontana2024characterizing,Ragone2024lie}. {Consider a sufficiently deep, $L$-layer periodic parametrized quantum circuit in the form}  
\begin{equation*}
    \bm{U}(\bm{\theta}) =\prod_{l=1}^L \prod_{k=1}^K e^{-i\theta_{l,k}H_k}.
\end{equation*}
The ensemble of all unitaries $\bm{U}(\bm{\theta})$ regarding different choices of $\bm{\theta}$ is a connected subgroup of $\Op{SU}(2^n)$~\cite{Alessandro2008}. Such a subgroup is associated with a real span of the Lie closure (i.e., closure under taking nested commutators), defined as,
\begin{equation*}
    \frak{g}:=\spn_{\RR}\liebracket{iH_1, iH_2, \cdots, iH_K},
\end{equation*}
called the DLA of the circuit~\cite{Dalessandro2021Introduction,Alessandro2008}.

The DLA completely characterizes the expressiveness of the QNN by delineating all possible unitary operations it can generate. It is known that a random initialization yields an approximately 2-design of the corresponding dynamical group $e^{\frak{g}}$ for a polynomially-sized periodic ansatz~\cite{Ragone2024lie,Fontana2024characterizing}. This enables a replacement of the parameter gradient with an abstract gradient~\cite{Fontana2024characterizing} defined as follows.
\begin{definition}[Abstracted Gradient~\cite{Fontana2024characterizing}]
    Let $G$ be a compact, connected Lie group represented as unitary matrices in the unitary group $\cU(V)$ over the vector space $V$. In addition, let $H\in \frak{g}$, and $iO, iA\in \frak{u}(V)$, the Lie algebra of $\cU(V)$. We define the abstracted gradient as:
    \begin{equation*}
        \partial_H \ell(A,O) := \tr\left(U^{\dagger}_{g^-} A U_{g^-} [H, U_{g^+} O U^{\dagger}_{g^+}]\right),
    \end{equation*}
    where $U_{g^{\pm}}$ represents arbitrary $g^{\pm} \in G$. %In order to prevent misleading, we can omit writing subscript $H$ to the differential operator.
\end{definition}
The gradient can be replaced by this abstracted version for any Lie subalgebra of %This replacement has been demonstrated for circuits generated by 
$\frak{su}(d)$ or $\frak{so}(d)$~\cite{Haah2024efficient}. Recent results have proven that for \textit{Lie Algebra Supported Ansatz} (LASA) with polynomially-scaled DLAs, a rapid mixing can occur towards an $\epsilon$-approximate 2-design of $e^{\frak{g}}$ after $\cO(\Op{poly}(n)\log(1/\epsilon))$ layers~\cite{Fontana2024characterizing}. {This yields an efficient estimation of the variance of gradient. More precisely,} {suppose $\frak{g} = \frak{c} \oplus \bigoplus_{j} \frak{g}_j$, where $\frak{c}$ is the center of the $\frak{g}$, an each $\frak{g}_j$ is a simple Lie subalgebra of $\frak{g}$. The gradient variance can be analytically evaluated via the following formula ~\cite{Fontana2024characterizing}} 
\begin{equation}
    \Var[\partial_H\ell(\rho,O)] = \sum_{j} \frac{\|H_{\frak{g}_j}\|_K^2\norm{O_{\frak{g}_j}}_F^2\|\rho_{\frak{g}_j}\|_F^2}{d^2_{\frak{g}_j}},
\end{equation}
where $O_{\frak{g}_j}$ represents the orthogonal projection of $O$ onto the Lie algebra $\frak{g}_j$, $d_{\frak{g}_j}$ gives the dimension of $\frak{g}_j$, and $\|\cdot\|_K$ and $\|\cdot\|_F$ represent the Killing norm and the Frobenius norm, respectively.

% \update{For polynomially-sized periodic ansatz, random initialization can yield approximately 2-designs of the corresponding dynamical group in the middle regions of the circuit ~\cite{Ragone2024lie,Fontana2024characterizing}. This makes it valid to replace the parameter gradient with an abstract gradient~\cite{Fontana2024characterizing}, and the behavior has been shown for systems generated by $\frak{su}(d)$ or $\frak{so}(d)$ with Pauli operators~\cite{Haah2024efficient}. It extends to various DLAs based on numerical evidence, and recent results confirm it for \textit{Lie Algebra Supported Ansatz} (LASA) with polynomial DLA, demonstrating rapid mixing to 2-design after $\cO(\Op{poly}(n)\log(1/\epsilon))$ layers, ensuring uniform exploration of parameter space and justifying the abstract gradient approach~\cite{Fontana2024characterizing}.}

\section{Supervised Learning with Quantum Data Embedding}

\subsection{Quantum supervised learning}\label{subsec:quantum_supervised_learning}

Quantum supervised learning is a pivotal subfield of quantum machine learning, focusing on using quantum devices to address regression, classification, and other learning tasks by exploiting quantum resources, such as coherence and entanglement of interested systems~\cite{Havlivcek2019supervised,Wang2024quantum}. 
% It is believed to have the potential to achieve better expressivity and, in certain settings, more efficient training compared to the classical approaches. 
The choice of data embedding strategies and how frequently data is “reuploaded” into the networks has a profound influence on the model’s representational power and generalization performance~\cite{Du2021learnability,Havlivcek2019supervised,Wang2024quantum}.
% Data occupies a central role in this framework: either classical or quantum data, the first step is to encode them into quantum systems via static or dynamic methods, which are then subsequently processed through quantum neural networks (QNNs)~\cite{Du2021learnability,Havlivcek2019supervised}. 
Therefore, designing flexible and efficient QNNs capable of extracting essential features from the input data is a crucial step for achieving reliable and scalable quantum supervised learning.

Consider a total training set cut into batches each of the form $\cT:=\{(\bm{x}_q, y_q): q = 1,2,\cdots, Q\}$, where $\bm{x}_q\in\RR^d$ is the input data, represented in either the static way as an initial state $\rho(\bm{x}_q)$ or a dynamic way through a quantum evolution $U(\bm{x}_q)$. In this work, we particularly focus on the latter. The aim is to learn a function $f: \RR^d \rightarrow \RR$ (or a discrete label set) by optimizing a QNN, interleaved with recurrent embeddings of $\bm{x}_q$, such that $f(\bm{x}_q) \approx y_q$. Concretely, let the entire network be represented as,
\begin{equation*}
    \bm{U}(\bm{x}_q;\bm{\theta}) = \prod_{l} g(U_l(\bm{x}_q))W(\bm{\theta}_l), 
\end{equation*}
where $g$ is the map that embeds data into the model, and $W(\bm{\theta}_l)$ is some trainable template QNN~\cite{Havlivcek2019supervised} with $\bm{\theta} = (\bm{\theta}_l)$. One can use the measurement outcomes from executing the QNN to represent the predicted labels or make the prediction by post-processing the measurement results. This framework has natural advantages in learning quantum data, for example, the Hamiltonian classification. In this work, we concentrate on classifying different features of quantum Hamiltonians by considering $\bm{x}_q = X_q$ for some Hermitian operators $X_q$ labelled with their corresponding features $y_q$. {In a more general case, $U_l$ can be seen as some layerwise function that maps $X_q$ to a unitary operator for further embedding.} %\syz{We haven't defined ``the left and right attached unitaries''.} \MR{Want to keep it general here.}

% This framework has {natural advantages in handling the supervised learning tasks on quantum data, for example, the unitary classification by embedding $\{U_q\}_q$ individually into the model.}
% This automatically incorporates more advanced quantum transformations, including quantum singular value transformation~\cite{Gilyen2019quantum}, higher-order quantum operations~\cite{Taranto2025higher} and quantum phase processing~\cite{Wang2023quantum}, offering a unifying view of quantum supervised learning models, encompassing both explicit circuit-based strategies and extensions involving kernel-based methods~\cite{Schuld2021supervised}.

\subsection{Quantum recurrent embedding neural network}\label{subsec:qrenn}
The embedding methods can significantly influence learning performance. In particular, QSVT- and QPP-based structures offer powerful approaches to approximate functions~\cite{Gilyen2019quantum,Wang2023quantum}. Meanwhile, the linear combination of unitaries (LCU) technique systematically composes complex operators from simpler blocks. All of these highlight the power of control embedding in driving versatile quantum transformations~\cite{Childs2012hamiltonian}. {Besides, the efficient implementation of control operations has been extensively studied in the context of quantum algorithms~\cite{zhou2020quantum,Low2017optimal}, providing important theoretical underpinnings for control embedding.}

\begin{figure}[h!]
    \centering
    \includegraphics[width=0.9\linewidth]{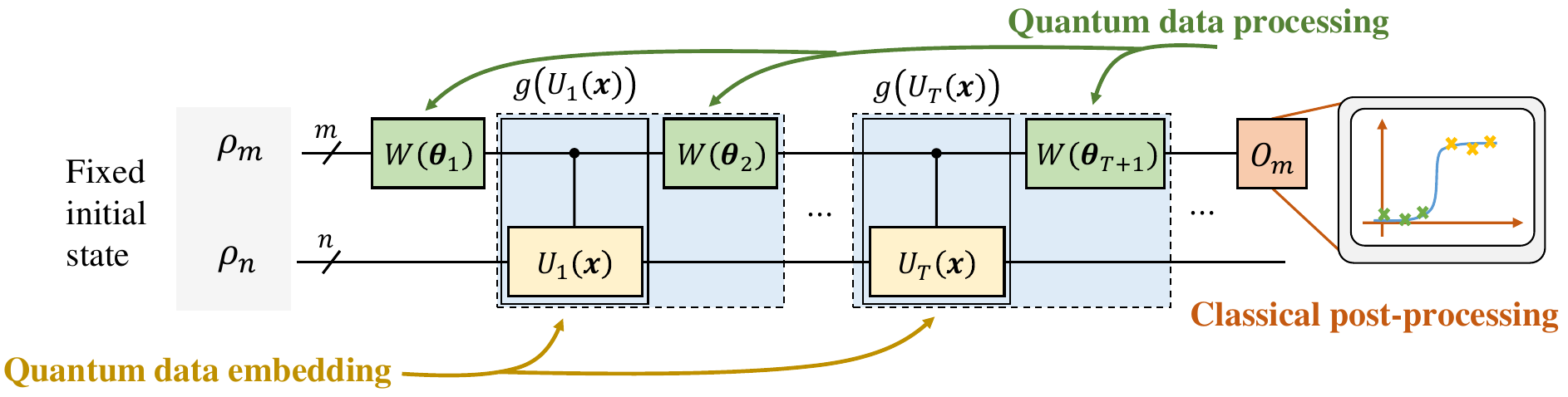}
    \caption{\textbf{The circuit diagram of a $T$-slot QRENN.} The input states and the measurement observable are fixed to $\rho_m\ox \rho_n$ and $O_m$, respectively. The upper $m$ qubits with green parameterized gates make up the data processing register. The lower $n$ qubits make up the data embedding register, where the data $\bm{x}$ is embedded through $g(U_t(\bm{x}))$.}
    \label{fig:general_frame_dqnn}
\end{figure}

{Inspired from the above, we design} the \textit{quantum recurrent embedding neural network} (QRENN) represented as a periodic sequence of unitary operators acting on an initially separated bipartite ($m+n$)-qubit quantum system as shown in Figure~\ref{fig:general_frame_dqnn}. We call the first $m$-qubit system the {\textit{data processing register}} and the next $n$-qubit system the {\textit{data embedding register}}, respectively. Let $\bm{\theta} = (\bm{\theta}_t)$ be the tunable parameters, where $\bm{\theta}_t = (\theta_{t,l})$ is a collection of all parameters in $W(\bm{\theta}_t)$, and 
$\bm{x}$ be the data vectors. A $T$-slot QRENN can be expressed as,
\begin{equation}\label{eq:dqnn_circuit}
    \bm{U}(\bm{x};\bm{\theta}) = (W(\bm{\theta}_{T+1})\ox I_{2^n})\left(\prod_{t=1}^{T} g(U_t({\bm{x}}))(W(\bm{\theta}_t)\ox I_{2^n})\right).
\end{equation}
For notational convenience, we will drop the last term $(W(\bm{\theta}_{T+1}) \ox I_{2^n})$ without loss of generality by incorporating it into the measurement operator.
% \update{Add a comment that we can drop the last term $(W(\bm{\theta}_{T+1})\ox I_{2^n})$ without loss of generality by adding a $U_0(0)$, and in the following we will only keep the rest $T$ shots for simplicity.} 
We specifically choose $g$ to be the control embedding in this work. Let $U$ be some $n$-qubit unitary matrix, the control embedding map $g$ has the following explicit form,
\begin{equation}\label{eq:encoding_map}
    g(U) := \ketbra{c}{c} \ox U + (I_{2^m} - \ketbra{c}{c})\ox I_{2^n},
\end{equation}
where $\ket{c}$, without loss of generality, is set to $\ket{1_m} = \ket{1}^{\ox m}$. {Let $U_t(\bm{x}) = e^{iH_t(\bm{x})}$ for some Hermitian generator $H_t:\RR^d \rightarrow i\frak{u}(2^n)$, embedded with data $\bm{x}$}. One can rewrite $g(U_t(\bm{x}))$ as $e^{i\ketbra{c}{c} \ox H_t(\bm{x})}$. Provided $W(\bm{\theta}_{t}) = \prod_{l=1}^L e^{i\theta_{t,l} \Omega_{l}}$ defined by some $m$-qubit $L$-layer periodic QNN templates~\cite{Ragone2024lie}. The expression~\eqref{eq:dqnn_circuit} can be written in the generator form,
\begin{equation}\label{eq:periodic_form_dqnn_circuit}
    \bm{U}(\bm{x};\bm{\theta}) = \prod_{t=1}^{T} \left(e^{i\ketbra{c}{c} \ox H_t(\bm{x})} \prod_{l=1}^L e^{i\theta_{t,l} \Omega_{l} \ox I_{2^n}}\right).
\end{equation}

This work focuses on quantum supervised learning tasks, specifically the classification of various types of quantum data using the QRENN framework. The proposed approach is inspired by quantum hypothesis testing and quantum discrimination tasks~\cite{Chen2024hypothesis,Regula2023postselected,Li2022sequential}. The following training strategy is devised: Let $\{M_{y_q}\}_q$ denote a set of positive operator-valued measurement (POVM) operators~\cite{Nielsen2010quantum}, referred to as labeling operators, which assign labels to each data point $\bm{x}_q$ in $\cT$ according to its corresponding feature class $y_q$. For example, when $y_q \in \{0, 1\}$, the eigen-projections $M_{y_q} = \ketbra{y_q}{y_q}$ of the observable $Z$ can be selected for training. Define $\ell(M_{y_q},\bm{x}_q; \bm{\theta}) = \tr(\bm{U}(\bm{x}_q; \bm{\theta}) \rho_0 \bm{U}(\bm{x}_q; \bm{\theta})^{\dagger} M_{y_q})$, where the initial state is fixed as $\rho_0 = \rho_m \otimes \rho_n$. The total loss over batch $\cT$ is given by
\begin{equation}\label{eq:total_loss}
    \cL(\bm{\theta}) = 1 - \frac{1}{Q}\sum_{q=1}^Q \ell(M_{y_q},\bm{x}_q; \bm{\theta}).
\end{equation}
Training of the QRENN is facilitated by classical-quantum hybrid algorithms, in which the network parameters are iteratively adjusted to minimize Eq.~\eqref{eq:total_loss} (e.g., via gradient descent) until convergence is achieved~\cite{Havlivcek2019supervised}. {After training, the same POVM is employed to test unseen data $\bm{x}'$ for prediction. Additionally, simple post-processing techniques may be applied to the measurement results~\cite{wu2024learning} to improve predictive performance.}

\subsection{Algebraic structure of QRENN}\label{subsec:dla_qrenn}
In the theoretical limit, the Dynamical Lie Algebra (DLA) can fully characterize the stochastic behavior of the QNN model~\cite{Fontana2024characterizing,Ragone2024lie}, as it encompasses all possible Hermitian generators that may interact with any tunable parameter~\cite{Dalessandro2021Introduction}. Given any Hermitian operator $H \in i\mathfrak{u}(2^n)$, the Baker–Campbell–Hausdorff (BCH) formula yields
\begin{equation*}
    e^{i\ketbra{c}{c} \ox H}e^{i\theta_{t,l} \Omega_{l} \ox I_{2^n}} = e^{i\ketbra{c}{c} \ox H + i\theta_{t,l} \Omega_{l} \ox I_{2^n} + i^2\frac{\theta_{t,l}}{2}[\ketbra{c}{c}, \Omega_l] \ox H + \cdots}.
\end{equation*}
Although the generator $i\ketbra{c}{c} \otimes H$ is not associated with any trainable parameter, the interaction between non-commutative generators can propagate the effect of parameter variations throughout the network.

To ensure consistency with previous studies~\cite{Fontana2024characterizing,Ragone2024lie} and to account for the effects of noncommutative algebra in the QRENN, we generalize our model by introducing a parametrization of the data embedding register. Specifically, let $\bm{\varphi} = (\varphi_{t})$ denote additional tunable parameters that are associated with the generators $H_t(\bm{x})$, enabling adjustment of the ‘weight’ of each data component at the $t$-th position in the network. This is analogous to the weight parameters in classical neural networks~\cite{Gurney2018introduction}. {This parametrization can also be interpreted in physical terms as a controllable evolution time.} With this extension, the network can be expressed as

% In order to make it consistent with previous studies~\cite{Fontana2024characterizing,Ragone2024lie} as well as cover the effect from noncommutative algebra in the QRENN, we slightly generalize our model by allowing a parametrization on the data embedding register. In particular, let $\bm{\varphi} = (\varphi_{t})$ be additional tunable parameters, which are attach to the generators $H_t(\bm{x})$ for adjusting the `weight' of each data component at the $t$-slot of the network. This can be equivalent to the weight parameters functioning in the classical neural networks~\cite{Gurney2018introduction}. \update{This can also be understood in physical terms as a controllable evolution time.} With this extension, the network becomes
\begin{equation}\label{eq:weighted_qrenn}
    \bm{U}(\bm{x};\bm{\theta}, \bm{\varphi}) = \prod_{t=1}^{T}\left(e^{i\varphi_{t}\ketbra{c}{c} \ox H_t(\bm{x})}\cdot\prod_{l=1}^L e^{i\theta_{t,l} \Omega_{l} \ox I_{2^n}}\right).
\end{equation}
We now define the DLA of both the standard and the parameterized QRENN model in the same form in order to prevent confusion. Formally, the DLA of QRENN is defined as,
\begin{equation}\label{eq:dla_qrenn_raw}
    \liedqnn := \spn_{\RR}\liebracket{\; i\ketbra{c}{c}\ox H_t(\bm{x}), i\Omega_l \ox I_{2^n}: \;l\in [L], t\in [T]\;}
\end{equation}
where we adopt the notation $[n] = \{1,2,\ldots,n\}$. If the range $[L]$ of an index $\ell$ is clear from the context, then we may also drop the range and simply write $\forall \ell$ for $\ell\in [L]$. 
% \MR{Checked}\update{Consistent with the notation: consider to use colon instead of comma, and explicitly write the range of each index.}

For convenience, in the following discussion, we use the shorthand notation $H_t$ to denote the data-embedded Hermitian operator $H_t(\bm{x})$ at each slot of the network, for some fixed data $\bm{x}$. Let $\cS = \{H_t\}_t$ be the set of possible Hermitian matrices (not necessarily commuting) acting on the data embedding register, and denote by $\cH_{\cS} \subset \CC^{2^n}$ the subspace on which the $H_t$ act non-trivially. According to the Wedderburn-Artin theorem, $\cH_{\cS}$ admits a direct sum decomposition as $\cH_{\cS} \simeq \bigoplus_{\bm{\lambda}} V_{\bm{\lambda}}$, where the restrictions of $\{H_t\}_t$ to each invariant subspace $V_{\bm{\lambda}}$ form the irreducible representation of the algebra spanned by $\{H_t\}_t$ (finite-dimensional $C^*$-algebra). The label $\bm{\lambda}$ is usually named the \textit{quantum number} shared by $\{H_t\}_t$. These invariant subspaces are important in quantum mechanics, which reflect coarse-grained properties, and each operator $H_t$ remains incompatible within $V_{\bm{\lambda}}$.

In this work, we focus on the scenario where all $H_t$ are commutative with each other. This happens when, for example, the data is embedded with diagonal real matrices. The set $\{H_t\}_t$ then generates a commutative algebra whose irreducible representations are one-dimensional. Select one eigenvalue $\lambda_t$ from each $H_t$ to form a tuple of eigenvalues $\bm{\lambda} = (\lambda_t)$, which then defines a unique \textit{joint eigenspace}
\begin{equation}
    V_{\bm{\lambda}}:=\{\ket{v} \in \CC^{2^n}: H_t\ket{v} = \lambda_t \ket{v}, \;\forall t\}.
\end{equation}
On $V_{\bm{\lambda}}$, every $H_t$ acts as a scalar multiple of the identity matrix. 
To be clear, consider the case of two commuting Hermitian operators \( H_1 \) and \( H_2 \). Each operator has its own set of eigenvalues, and for any pair of eigenvalues \( \lambda_1 \) and \( \lambda_2 \) associated with \( H_1 \) and \( H_2 \), respectively, the joint eigenspace is defined as:
\[
V_{(\lambda_1, \lambda_2)} := \Big\{ \ket{v} \in \mathbb{C}^{2^n} : H_1 \ket{v} = \lambda_1 \ket{v}, \; H_2 \ket{v} = \lambda_2 \ket{v} \Big\}.
\]

The joint eigenspace \( V_{(\lambda_1, \lambda_2)} \) is the intersection of the eigenspaces of \( H_1 \) and \( H_2 \) corresponding to \( \lambda_1 \) and \( \lambda_2 \). Each such eigenspace is uniquely labeled by the tuple \( (\lambda_1, \lambda_2) \).
This example illustrates that the joint eigenspace is the intersection of the individual eigenspaces of the commuting operators, with the eigenvalue tuple 
$\bm{\lambda}$ uniquely identifying each common subspace on which every 
$H_t$ acts as multiplication by the corresponding scalar 
$\lambda_t$. By characterizing the projection $\Pi_{\bm{\lambda}}$ onto each subspace $V_{\bm{\lambda}}$, one can induce direct sum algebraic structures in $\liedqnn$, as demonstrated below.
\begin{tcolorbox}
\begin{proposition}\label{prop:dla_main_result}
    Let $\{\Omega_j\}_j$ be the set of Hermitian generators of the data processing register so that $\liebracket{i\Omega_j}$ spans $\frak{su}(2^m)$. Suppose a set of commutative Hermitian data matrices $\{H_t\}_t$ is embedded into the QRENN model. Then the DLA of the QRENN model can be decomposed into,
    \begin{equation*}
        \liedqnn \simeq \frak{c} \oplus \frak{su}(2^m)^{\oplus r}
    \end{equation*}
    where $\frak{c}:=\spn_{\RR}\liebracket{iI_{2^m}\ox\Pi_{\bm{\lambda}}\;:\forall \bm{\lambda}}$ is the center, $r$ is the number of distinct joint eigenspaces from all of $H_t$, and $\Pi_{\bm{\lambda}}$ is the projection onto the corresponding space such that $\sum_{\bm{\lambda}} \Pi_{\bm{\lambda}} = I_{2^n}$.
\end{proposition}
\end{tcolorbox}
{In our later applications, we set $m = \cO( \log n)$ which makes $2^m = poly(n)$, %\syz{This is not right because $2^{O(\log n)} = poly(n)$. In later applications we actually set $m = \lfloor \log n\rfloor$?} 
to avoid an exponentially large dimension of processing register. Besides, being able to generate the whole $\frak{su}(2^m)$ enables us to have full control of $m$ qubits, which allows arbitrary processing of the data.}
%\syz{Consider to add some justifications on the assumption of $\liebracket{i\Omega_j} = \frak{su}(2^m)$. I can think of two: (1) in later applications $m=\log(n)$ which makes $2^m = n$, not an exponentially large dimension; (2) being able to generate the whole $\frak{su}(2^m)$ enables us to have full control of $m$ qubits, which allows arbitrary processing of the data.} 
This proposition showcases that with commutative data embedding Hermitian generators, the DLA of our QRENN can be decomposed into simple Lie algebras with bounded dimensions depending on the dimensionality of the data processing register.
The formal proof can be found in Appendix \ref{sec:appendix_dla_modified_periodic_models}. Here, we provide a proof sketch by showing that the original expression~\eqref{eq:dla_qrenn_raw} can be rewritten as
\begin{equation}\label{eq:modified_dqnn_dla}
    \liedqnn = \spn_{\RR}\langle \; iI_{2^m} \ox H_p, iP_s\ox I_{2^n}, iP_q\ox H_r:\;\forall p,s,q,r \; \rangle_{\Op{Lie}},
\end{equation}
where $\liebracket{iP_j: \; \forall j}$ spans an orthonormal basis for \(\mathfrak{su}(2^m)\); for instance, the basis may be chosen as the $m$-qubit normalized Pauli matrices formed from the Kronecker product of single-qubit Pauli matrices, with normalization taken with respect to the Hilbert–Schmidt norm. Notably, the terms \(iI_{2^m}\otimes H_p\) commute with all other generators, thereby spanning a real, abelian center of the algebra. Owing to the spectral decomposition of each \(H_p\), the central terms \(iI_{2^m}\otimes H_p\) can be rewritten in terms of the projections \(\Pi_{\bm{\lambda}}\), so that \(\mathfrak{c} = \operatorname{span}_{\mathbb{R}}\langle iI_{2^m}\otimes \Pi_{\bm{\lambda}} \,:\, \forall\, \bm{\lambda}\rangle_{\operatorname{Lie}}\). Exploiting the spectral theorem, each Hermitian data operator \(H_p\) is decomposed into a linear combination of projection operators. Analyzing the nested commutators reveals that the commutator ideal \([\liedqnn,\liedqnn]\) decomposes into a direct sum of simple subalgebras, each isomorphic to \(\mathfrak{su}(2^m)\) and corresponding to a particular eigenspace. Consequently, the overall structure of the DLA is given by
\(\liedqnn\simeq \mathfrak{c}\oplus \mathfrak{su}(2^m)^{\oplus r}\),
with \(r\) denoting the number of distinct eigenspaces.

It can be observed that the first two terms in the Lie bracket expansion of Eq.~\eqref{eq:modified_dqnn_dla} correspond to the same generators as those in the QRENN model without control embedding. In contrast, the third term introduces interactions between the data processing register and the data embedding register, which play a dominant role in determining both the trainability and expressive power of the QRENN. This decomposition of the DLA further implies enhanced training efficiency, as it allows the number of qubits in the processing register to remain limited without compromising performance.

{Importantly, the use of control operations is central to query-based quantum algorithms~\cite{Gomes2024multivariable}, as they enable coherent amplification of nontrivial spectral features. However, simulating the dynamics of general quantum Hamiltonians is classically intractable, and simulating their controlled evolutions is even more challenging. In our discussion, we assume that such control operations can be efficiently implemented on quantum hardware through Hamiltonian simulation techniques~\cite{Low2017optimal,Low2019hamiltonian}, leveraging the capabilities of quantum devices. Additionally, even with recent advancements such as $\mathfrak{g}$-sim~\cite{Goh2023lie}, a Lie-algebraic simulation method for PQCs, a complete characterization of $\liedqnn$ still requires computing the full spectrum of each Hamiltonian component $H_p$, which incurs exponential classical computational cost. Consequently, although the dimensions of the relevant Lie algebra subspaces in $\liedqnn$ are polynomially bounded, simulating the evolution under the dynamic group $e^{\liedqnn}$ remains infeasible for classical devices.}

% It can be observed that the first two terms in the Lie bracket expansion of Eq.~\eqref{eq:modified_dqnn_dla} correspond to the same generators as in the QRENN model without control embedding. The last term introduces interactions between the data processing register and the data embedding register, which subsequently dominate the trainability and expressive power of the QRENN model. This decomposition of the DLA also suggests high training efficiency for the QRENN, as the number of qubits in the processing register can be limited. \update{Notably, the uses of control operations are central to query-based algorithms~\cite{Gomes2024multivariable} due to their ability to coherently amplify nontrivial spectral features. However, simulating the dynamics of a general quantum Hamiltonian is known to be classically hard. Simulating the action of the control evolution is even harder. In our discussion, we have assumed that control operations can be efficiently implemented in the circuit via Hamiltonian simulation methods~\cite{Low2017optimal,Low2019hamiltonian}, leveraging the capabilities of quantum computers. Additionally, even with the recent development of $\frak{g}$-sim~\cite{Goh2023lie}, a Lie-algebraic classical simulation strategy for PQCs, a complete description of $\liedqnn$ requires computing the entire spectrum of each $H_p$, which costs exponential classical resources. As a result, despite the dimension of subspaces in $\liedqnn$ being polynomially bounded, the simulation over the dynamic group $e^{\liedqnn}$ can hardly be captured by classical devices.}

\section{Absence of Barren Plateaus in QRENN}\label{sec:trainability_dqnn}
The barren plateau (BP) phenomenon arises from multiple interrelated factors, including the expressiveness of the quantum circuit, the initialization of the quantum state, and the choice of measurement operators, as recently unified within a comprehensive theoretical framework~\cite{Fontana2024characterizing,Ragone2024lie,Barnum2003Generalizations,Barnum2004A}. Extensive research has been devoted to developing strategies that mitigate barren plateaus~\cite{Grant2019initialization,Skolik2021layerwise,Rad2022surviving,Friedrich2022avoiding,Kulshrestha2022beinit,Zhang2022,Liu2024mitigating,Jing2024quantum}; however, the fundamental challenge remains: designing architectures that are inherently resistant to barren plateaus and that ensure reliable trainability.

\begin{figure}[t!]
    \centering
    \includegraphics[width=0.85\linewidth]{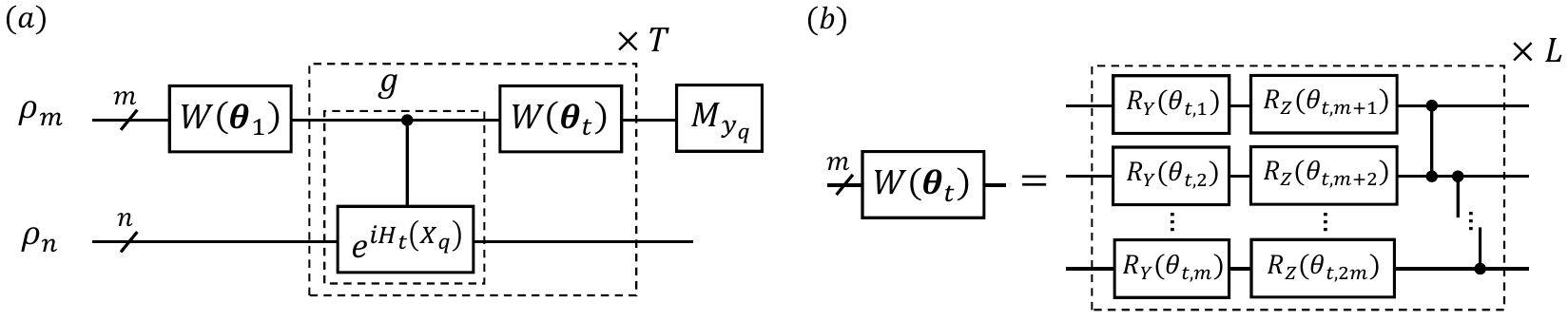}
    \caption{\textbf{The sketch of the QRENN circuit for quantum Hamiltonian supervised learning.} In (a), the model is initialized with a fixed initial state and embedded with the $q$-th data in $\cT$ via controlled-unitary evolutions $g(e^{i H_t(X_q)})$ followed by a measurement $M_{y_q}$; in (b), the PQC applied on the processing register is given by the template plotted on the right with $L$ repetitions. Within the block, $R_Y$ and $R_Z$ denote the single-qubit $Y$ and $Z$ rotations, respectively. {The entanglement between each pair of qubits in the processing register is generated via the linear-applied $CZ$ gates at the end of the block.}}
    \label{fig:QRENN_circuit}
\end{figure}

Relevant works have focused on designing symmetric architectures based on low-dimensional DLA structures~\cite{Schatzki2024theoretical,Nguyen2024theory}. In this work, we demonstrate that, through control embedding, our QRENN can avoid the BP phenomenon and represents a promising class of efficient solvers for quantum supervised learning problems. {We focus on the problem of classifying different types of quantum Hamiltonians. Consider a batch training set of Hamiltonians, each labeled with its corresponding feature}. In this setting, each Hamiltonian is treated directly as input data. A natural approach to embedding a Hamiltonian into the network is to utilize its associated unitary evolution. Adopting the notation introduced in previous sections, let $\cT = \{(X_q, y_q)\}$, where $X_q$ denotes the $q$-th data Hamiltonian. Consequently, the corresponding evolution $e^{iH_t(X_q)}$ is embedded into the QRENN via control embedding~\eqref{eq:encoding_map}.
% The embedding generator $H_t$ is taken to be identical throughout the network, 
%\syz{The general index for $H$ changes from $t$ in the above Definition 2 to $q$. Better to be consistent. If we like to avoid confusion with slot index $t\in [T]$, then maybe just consistently use $q$?}
The corresponding circuit diagram is presented in Figure~\ref{fig:QRENN_circuit}. {A natural example of $H_t$ is  $H_t(X_q) = X_q$ for all $t \in [T]$. In what follows, we will explicitly state the definition of $H_t$ before further discussion.} 
%\syz{In figure (a), should the loading layer be $e^{iX_q}$? In figure (b), should the parameters $\theta_t$ also appear in $R_Y$ and $R_Z$? } \MR{I have added a explanation for $H_t:\RR^d \rightarrow i\frak{u}(2^n)$. We want to keep it general until settled later in specific applications.}

To analyze the trainability of QRENN, we shall examine the stochastic behavior of the loss function partial derivatives $\partial_{t,\mu}\cL = \frac{\partial \cL}{\partial \theta_{t,\mu}}$ with respect to any $\mu$-th tunable parameter at the $t$-th slot of the model, where the loss function $\cL$ is defined in Eq.~\eqref{eq:total_loss}. Following established metrics~\cite{Mcclean2018barren,Cerezo2021cost}, the scaling of $\Var_{\bm{\theta},\bm{\varphi}}[\partial_{t,\mu} \cL]$ serves as the key indicator of the trainability. We say that the model has a BP in the $\theta_{t,\mu}$ direction if $\EE_{\bm{\theta},\bm{\varphi}}[\partial_{t,\mu} \cL] = 0$ and $\Var_{\bm{\theta},\bm{\varphi}}[\partial_{t,\mu} \cL]$ vanishes exponentially with system size.
\begin{tcolorbox}
\begin{proposition}\label{prop:zero_mean_gradient_lb_grad_var}
    Consider an ($m+n$)-qubit QRENN model with a set of Hermitian data generators $\cG = \{H_t(\bm{x}_q)\}_t$ embedded via Eq.~\eqref{eq:encoding_map} for each $q\in [Q]$, running on an input state $\rho$. %, and a tester set $\{M_{y_q}\}_q$. 
    Assume the model circuit is sufficiently deep to form an $\epsilon$-approximate $2$-design on the corresponding dynamic subgroup. Then the total loss function defined in Eq.~\eqref{eq:total_loss} has 
    \[
        \EE_{\bm{\theta},\bm{\varphi}}\left[\frac{\partial \cL}{\partial \theta_{t,\mu}}\right] = 0,
    \] 
    for all $t\in [T]$ and $\mu\in [L]$. Moreover, denoting the derivative generator as $\Omega$, then
    \begin{equation*}
        \Var_{\bm{\theta}, \bm{\varphi}}\left[\frac{\partial \cL}{\partial \theta_{t,\mu}}\right] \geq \frac{1}{Q^2}\sum_{q}\sum_{\bm{\lambda}_q}   \frac{\norm{(M_{y_p})_{\frak{g}_{\bm{\lambda}_q}}}_F^2\norm{\rho_{\frak{g}_{\bm{\lambda}_q}}}_F^2 \norm{\Omega_{\frak{g}_{\bm{\lambda}_q}}}_K^2}{d^2_{\frak{g}_{\bm{\lambda}_q}}},
    \end{equation*}
    where for each $\bm{x}_q\in\cT$, the summation is taking over all distinct joint eigenspace $\frak{g}_{\bm{\lambda}_q}$ with respect to the corresponding DLA, and the center $\frak{c}_q$ does not contribute to the variance.
\end{proposition}
\end{tcolorbox}
The proof of this proposition is provided in Appendix~\ref{appendix_subsec:total_loss_gradient}.%\syz{Don't quite see where the proof is in the appendix. Suggest to add an explicit one.} 
It should be noted that the proposition does not guarantee trainability for arbitrary training sets and input states. If either $M_{y_q}$ or the input state $\rho$ has an exponentially small projection onto {$\liedqnn$ with respect to any data in $\cT$}, the bound becomes trivially decaying. {On the other hand, if $\cT$ is carefully selected and appropriate input states and measurements are used, the gradient variance can be lower bounded by $1/\Op{poly}(n)$.}
%\syz{Let $V$ be the span of the ranges of all $H_t$, then $R^2_{\cS}(O) := \tr(\Pi_V O)^2$.}
\begin{definition}[Joint eigenspace overlap]\label{def:joint_eigenspace_radius}
    Consider a set of commuting Hermitian matrices $\cS = \{H_t\}_t \subset i\mathfrak{u}(d)$, with joint eigenspaces characterized by the projections $\{\Pi_{\bm{\lambda}}\}_{\bm{\lambda}}$. The joint eigenspace overlap of a Hermitian operator $O \in i\mathfrak{u}(d)$ with respect to the set $\cS$ is defined as $R^2_{\cS}(O) := \sum_{\bm{\lambda}\ne \bm{0}} \tr(\Pi_{\bm{\lambda}} O)^2$. If $\cS = \{H\}$ has one element $H$, we write $R^2_H(O)$ by replacing $\cS$ with $H$.
\end{definition}

\update{
\begin{tcolorbox}
\begin{proposition}\label{prop:variance_dqnn_On_In}
    Under the same assumptions as in Proposition~\ref{prop:zero_mean_gradient_lb_grad_var}, let $m$ scale as $\cO(\log(n))$, input state $\rho = \ketbra{\psi}{\psi}\ox\rho_n$ for some $m$-qubit pure state $\ket{\psi}$, and $n$-qubit state $\rho_n$. Suppose that each processing Hermitian $\Omega_{\mu} \ox I_{2^n}$ and labeling measurement $M_{y_q} \ox I_{2^n}$ acting locally on the processing register, and $\norm{\Omega_{\mu}}^2_F, \norm{M_{y_q}}^2_F$ scales as $\Omega(1)$ for $q\in[Q]$. If the batch size $Q$ scales as $\cO(\Op{poly}(n))$ and at least one data $X_q$ offers $\Omega(1/\Op{poly}(n))$ scaled $R_{X_{q}}(\rho_n)$. Then,
    $$\Var_{\bm{\theta}, \bm{\varphi}}[\partial_{t,\mu} \cL] \geq \Omega(1/\Op{poly}(n)).$$
\end{proposition}
\end{tcolorbox}}

{The detailed proof of the above results can be found in Appendix~\ref{appendix_subsec:proof_of_them_no_BP}. The fundamental reason underlying the proposition is that the control embedding is sufficiently ``weak'', such that it only polynomially expands the DLA when $m \in \cO(\log(n))$, compared to the case with decoupled data processing and embedding registers. As long as $M_{y_q}$ acts locally, an appropriate choice of $\cT$ can reduce or even eliminate the barren plateau effect, provided the overlap between $\rho_n$ and the eigenspaces of each data $X_q$ is sufficiently large.} %\syz{Please also write the proof in the appendix with explicit reference to this corollary.}

% The choice of the initial state $\rho_n$, which serves as a ``probe'' state~\cite{Giovannetti2011advances}, is also crucial for efficient training in supervised learning tasks. When $\rho_n = I_{2^n}/2^n$, our QRENN reduces to the well-known \textit{one-clean-qubit} model (DQC1)~\cite{Morimae2017power,Kim2024expressivity,Xuereb2023deterministic}, which has been proven to be classically hard to simulate~\cite{Fujii2018impossibility,Morimae2014hardness}. In this scenario, the interaction between $\rho_n$ and the data generators becomes exponentially small, potentially inducing barren plateaus.

{Recent studies, notably Cerezo et. al.~\cite{Cerezo2023does}, have raised notable concerns that variational quantum circuits which provably avoid BP may also admit efficient classical simulation, particularly if their effective dynamics are restricted to polynomially small subspaces. This perspective suggests a potential trade-off between trainability and quantum advantage. However, our proposed QRENN architecture resists such dequantization. Specifically, when the data processing register is restricted, the QRENN model naturally reduces to powerful algorithmic primitives such as QSP, QSVT, and their generalized variants~\cite{Gilyen2019quantum,Low2017optimal,Rossi2022multivariable}, which are believed to be classically intractable to simulate in general.  Importantly, the assumption of a sufficiently large initial overlap with relevant eigenspaces, which underpins our trainability guarantees, is not overly restrictive. Similar assumptions underlie the success of many QSP- and QSVT-based quantum algorithms~\cite{Lee2023evaluating,Dong2021efficient,Lin2022heisenberg,Gilyen2019quantum} where the input state is often prepared to match the desired spectral support. On the other hand, in the extreme case where this overlap is exponentially small, e.g., by choosing the maximally mixed state $\rho_n = I_{2^n}/2^n$, our model reduces to the DQC1 (one-clean-qubit) setting~\cite{Morimae2017power,Kim2024expressivity,Xuereb2023deterministic}, which is known to be classically hard to simulate despite its simplicity~\cite{Fujii2018impossibility,Morimae2014hardness}. Therefore, QRENN resides in a regime that is both trainable (under reasonable assumptions), making it a compelling candidate for scalable quantum machine learning.}

% \MR{argument for, analysis, connection to other algorithms, classical simulability.}

{To further investigate the trainability of QRENN in quantum supervised learning, we focus on a binary classification task. Specifically, we consider distinguishing three types of Hamiltonians from random Hermitian matrices generating Haar unitaries: $n$-qubit Pauli matrices, involutory matrices, and real diagonal matrices. To construct $\cT$ with fixed size $|\cT| = 100$, we randomly select $50$ elements from the feature set and another $50$ from Haar random unitaries. The $n$-qubit Pauli matrices are the tensor product of random single-qubit Pauli matrices, resulting in $4^n - 1$ elements (excluding the $n$-qubit identity). Involutory matrices are generated by constructing diagonal matrices with entries randomly chosen from $\{+1, -1\}$, then applying the adjoint action of a Haar unitary. For random diagonal Hamiltonians, each diagonal entry is independently and uniformly sampled from $[0, \pi)$.} For the binary classification, we define the two-outcome POVM $\{M_{0}, M_{1}\}$ as
\begin{equation}\label{eq:binary_labeling_op}
    M_0 = \frac{I_{2^{m}} - Z^{\ox m}}{2} \ox I_{2^n}; \quad M_1 = \frac{I_{2^{m}} + Z^{\ox m}}{2} \ox I_{2^n}.
\end{equation}
With the above setup, we now demonstrate the trainability of our QRENN in learning these three classes of Hamiltonians. %\syz{For the supervised learning part, consider to (1) verify that the conditions of corollary \ref{prop:variance_dqnn_On_In} all hold, (2) explicitly compute the value of the $\Omega(1/poly(n))$ bound, and (3) conduct numerical experiment for the variance in the left hand side by sampling the parameters and compute the variance of the statistical data. It'll be great if (1) shows that all the conditions indeed hold, and (2) and (3) combined show that the bound is quite tight.}
\begin{corollary}\label{coro:binary_class_pauli_inv_diag}
    Under the conditions established {in Proposition~\ref{prop:variance_dqnn_On_In} with $\{M_0, M_1\}$ defined in Eq.~\eqref{eq:binary_labeling_op}}, %\syz{Please clarify which conditions}, 
    the QRENN can distinguish $n$-qubit Pauli, involutory, and diagonal quantum Hamiltonians from random Hermitian operators without encountering BP.
\end{corollary}
The proof of the corollary directly follows from Proposition~\ref{prop:variance_dqnn_On_In}. By choosing the labeling operators $M_0, M_1$ defined in Eq.~\eqref{eq:binary_labeling_op}, we can first verify,
\begin{equation*}
    \norm{\frac{I_{2^m} \pm Z^{\ox m}}{2}}_F^2 = \frac{1}{2}\tr(I_{2^m} \pm Z^{\ox m}) = 2^{m-1} \in \Omega(1/\Op{poly}(n)),  
\end{equation*}
as $m$ scales as $\cO(\log(n))$. It then suffices to show that for each of these three features, there always exists $X_q$ such that $R_{X_q}(\rho_n)$ is sufficiently large. Indeed, for the Pauli and involutory feature sets, we can choose $\rho_n = \ketbra{+}{+}^{\otimes n}$. For any $X_q$ corresponding to these features, one can show that $R^2_{X_q}(\rho_n) > 1/2$, satisfying the assumption of Proposition~\ref{prop:variance_dqnn_On_In}. For the diagonal feature set, we set $\rho_n = 0.5\ketbra{+}{+}^{\otimes n} + 0.5 I_{2^n} / 2^n$ to examine the performance of our QRENN model with a noisy input probe. In this case, $R^2_{X_q}(\rho_n) > 1/4$. As a result, since $\cT$ contains half of its elements generated from each corresponding feature set, we conclude that the conditions of Proposition~\ref{prop:variance_dqnn_On_In} are satisfied, completing the proof. Detailed calculations of the joint eigenspace overlap can be found in Appendix~\ref{app:joint_eig_overlap}. 

\begin{figure}[t!]
    \centering
    \includegraphics[width=1.0\linewidth]{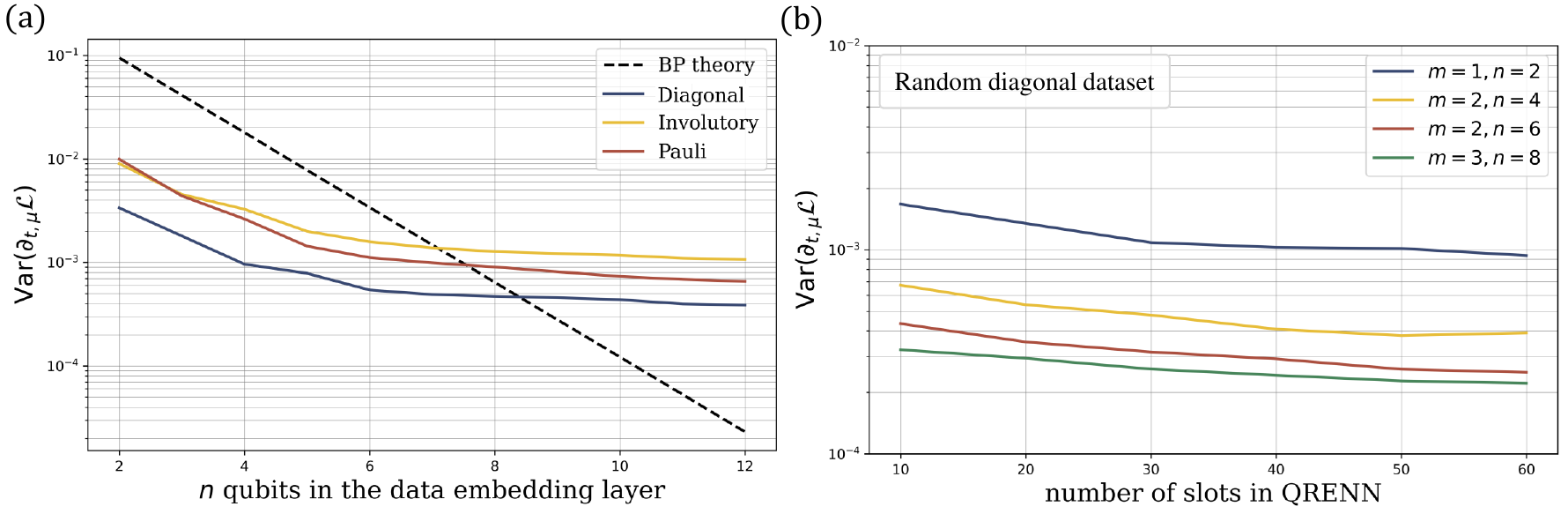}
    \caption{\textbf{Gradient statistics for QRENNs with various datasets embedded in with respect to (a) the number of qubits in the data embedding register and (b) the number of slots in the model.} In subfigure (a), the blue, yellow and red curves illustrate the results from embedding $\cT$ containing features of diagonal, involutory and Pauli Hamiltonians, respectively. The plot showcases our QRENN model experiences a polynomial decay in the gradient variance, diverging from BP. In subfigure (b), we concentrate on the diagonal Hamiltonian dataset by examining the gradient variance with respect to the circuit depth. The curves in different colors illustrate the results of $n = 2,4,6,8$. This setup illustrates the rapid mixing towards \(e^{\liedqnn}\) within the polynomial slots of the network.}
    \label{fig:gradient_test_dqnn}
\end{figure}

\subsection{Numerical demonstration on trainability}

In this section, we present numerical experiments on gradient statistics for learning the three datasets, demonstrating the trainability of QRENN. We also analyze the rate of concentration as the network transitions towards the fully mixing phase in the represented unitary group, as the number of slots in the model increases.

To construct the data processing register, we employ the $R_Y$-$R_Z$ circuit template~\cite{Sim2019expressibility} with multiple layers in each block $W(\bm{\theta}_t)$, as illustrated in Figure~\ref{fig:QRENN_circuit}(b). {To generate the entanglement between pairs of qubits in the processing register, we apply $CZ$ gates linearly after those single-qubit rotations.} %\MR{I have also added the explanation for the $CZ$ applied in the block in the caption of Figure~\ref{fig:QRENN_circuit}} \syz{Please give more details on the circuit block $W$: What are the gates after $R_Y$ and $R_Z$? Which pairs of qubits do we apply those gates?} 
For each sampling experiment, the initial state is fixed as $\ketbra{0}{0}^{\otimes m} \otimes \rho_n$, {with the corresponding probe states defined in Section~\ref{sec:trainability_dqnn} for each feature set.} All other conditions, including the measurement, circuit construction, size of $\cT$, and embedding method, are set as in Proposition~\ref{prop:variance_dqnn_On_In}.

The gradient sampling experiments are performed using the \textit{quairkit} Python toolkit~\cite{quairkit} on a physical workstation with~\hardware. For each experiment, every component of $\bm{\theta}$ is initialized uniformly over $[0,2\pi)$. The gradient statistics in Figure~\ref{fig:gradient_test_dqnn} are plotted for a fixed parameter $\theta_{1,1}$ with respect to (a) the number of qubits $n$ in the data embedding register and (b) the number of slots in QRENNs. Each data point is obtained by sampling $500$ randomly initialized networks.

In Figure~\ref{fig:gradient_test_dqnn}(a), the blue, yellow, and red curves correspond to the results for the diagonal, involutory, and Pauli Hermitian embedding datasets, respectively. In all three cases, the gradient variance decreases slowly with system size, representing a significant departure from the exponential decay (shown by the black dashed line) characteristic of barren plateaus.

Additionally, in Figure~\ref{fig:gradient_test_dqnn}(b), we focus on the random diagonal Hermitian dataset for different $n$. A rapid convergence in the variance values is observed as the number of slots in the QRENN increases. This demonstrates convergence to the second moment of the variance as a function of network depth and confirms the use of the abstracted gradient. Furthermore, the gradually converging curves coincide with the polynomial decay of the variance with $n$, as shown in Figure~\ref{fig:gradient_test_dqnn}(a).

\section{Application of QRENN in quantum supervised learning}
{Classifying quantum Hamiltonians is a fundamental task in quantum many-body physics and quantum information, as it enables the identification of distinct physical phases~\cite{Schuld2021supervised,Li2022}, symmetries~\cite{laborde2022quantum,Chen2025hypothesis}, and computational properties of quantum systems~\cite{huang2022quantum}. Such classification is vital for understanding emergent behaviors like topological order~\cite{wu2024learning}, for designing efficient simulation algorithms, and for verifying quantum hardware by distinguishing structured from random or noisy dynamics.} In this section, we show that recursive data embedding can enhance the classification power of quantum machine learning models, particularly for features of Hamiltonians previously examined in the gradient statistics. Furthermore, we numerically demonstrate that our model can be trained to predict symmetry-protected topological (SPT) phase ordering~\cite{wu2024learning}, highlighting the potential of QRENN for studying the underlying properties of condensed matter systems.

\subsection{Supervised learning on classes of Hamiltonian}
We first examine the performance of our QRENN in classifying different types of quantum Hamiltonians, as studied in the gradient statistics from Section~\ref{sec:trainability_dqnn}. To construct the batch training set $\cT$ and the test set for different numbers of qubits $n$ (ranging from $3$ to $8$), we first generate $600$ raw data points, with half generated from the feature sets and half from random Hermitian matrices that generate Haar unitaries. {Each data point is pre-labeled by $y_q = 1$ or $0$, indicating whether or not the data $X_q$ possesses the corresponding feature. We then uniformly select $100$ elements from the raw dataset to form the batch training set $\cT$, with the remaining data constituting the test set.} %\syz{Is there a reason why you used only 100 out of 600 for training?} \MR{No specific reason for the numbers, one may say that the generalization of our model is not bad? because the test set is relatively large than the training size.}%\syz{why shuffling? the test set shouldn't be seen by the model.} 
The model is first trained on the batch training set $\cT$ and then evaluated on the test set to assess its classification performance.
\begin{figure}[t!]
    \centering
    \includegraphics[width=0.85\linewidth]{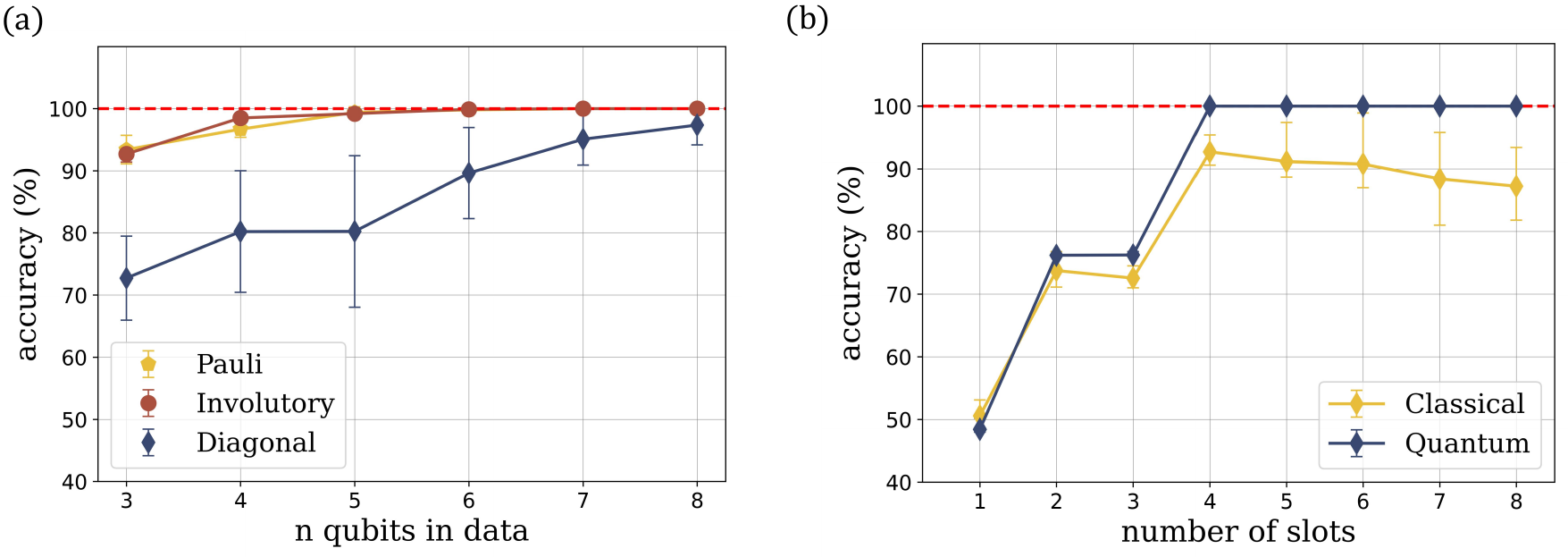}
    \caption{\textbf{The test accuracy of QRENN supervised learning on different feature sets of quantum Hamiltonian, with respect to (a) the number of qubits of the data embedded in, and (b) the number of slots for querying data.} In (a), we perform ideal simulations to examine the accuracy of QRENN in learning Pauli, involutory and diagonal features. The number of slots is fixed to $4$ for the Pauli and involutory feature sets, and $6$ for the diagonal feature set, respectively. In (b), we fix $n=6$ and examine the robustness of QRENN against both classical and quantum noise by increasing the number of slots.}
    \label{fig:application_ham_classification}
\end{figure}

We fix $m=2$ in the data processing register for all data sizes $n$. The embedding generator for Pauli and Involutory datasets is defined throughout the entire model as $H_t(X_q) = \frac{\pi}{2}(I_{2^n} - X_q)$ while for diagonal datasets, $H_t(X_q) = X_q$. Measurements on the data processing register are defined in Eq.~\eqref{eq:binary_labeling_op}, corresponding to the binary class labels. The total loss function is given in Eq.~\eqref{eq:total_loss}. During prediction, classification is performed by estimating the expectation value of the observable $Z^{\otimes 2} \otimes I_{2^n}$. If the expectation value is larger than zero. We treat the predicted label as $\Tilde{y}_q = 1$. Otherwise, $\Tilde{y}_q = 0$. Denote $\hat{f}(Z \otimes I_{2^8},X_q; \bm{\theta}^{\star}) = \tr(\bm{U}(X_q; \bm{\theta}^{\star}) \rho_0 \bm{U}(X_q; \bm{\theta}^{\star})^{\dagger} Z \otimes I_{2^8})$, we do the following processing rule,
\begin{equation}
\tilde{y}_q = 
\begin{cases} 
0, & \text{if } \hat{f}(Z \otimes I_{2^8},X_q; \bm{\theta}^{\star}) < 0, \\
1, & \text{otherwise}.
\end{cases}
\end{equation}
The accuracy is then quantitatively assessed as: $1- \frac{1}{Q} \sum_{q=1}^Q \left|y_q - \tilde{y}_q\right|$.

As shown in Figure~\ref{fig:application_ham_classification}(a), our model achieves nearly perfect accuracy in identifying the Pauli and Involutory Hamiltonians under ideal simulations. {The observed substantial variations in model accuracy suggest that the optimization process encounters numerous local minima when applied to the diagonal dataset.} Notably, we observe a clear improvement in classification accuracy for all three datasets as the number of data qubits increases. This phenomenon can be attributed to the geometric properties of high-dimensional spaces, particularly measure concentration. In high-dimensional systems, Haar random unitaries tend to exhibit similar statistical behaviors, whereas structured unitaries such as Pauli matrices remain distinctly identifiable. 

Furthermore, we extend our analysis by incorporating hardware imperfections into the quantum data embedding process, considering the effects of both classical and quantum noise. In classical supervised learning, robustness to label noise is often tested by introducing labeling errors~\cite{Song2022learning}. We first set the labeling error rate to $0.05$. As shown by the yellow curve in Figure~\ref{fig:application_ham_classification}(b), embedding more data improves test accuracy. We also introduce quantum noise arising from data embedding. Practical quantum hardware inevitably suffers from control errors and qubit crosstalk, which cause deviations from ideal evolution~\cite{Ash2020experimental,Zhou2023quantum}. To model these deviations, we introduce an additional noise operation $\cE$ in the form of a small global perturbation $\delta H_0$ {with $\norm{H_0}_{\infty} \leq 1$}, so that the control embedding gate becomes
\begin{equation}
\cE\left(e^{i\ketbra{1_m}{1_m}\otimes X_q}\right) = e^{i(\ketbra{1_m}{1_m}\otimes X_q + \delta H_0)}.
\end{equation}
We set the noise amplitude to $\delta = 0.01$, ensuring operation within the weak noise limit, and evaluate our model on the $n=6$ involuntary feature set.

As shown by the blue curve in Figure~\ref{fig:application_ham_classification}(b), introducing the perturbation causes the average performance to remain close to the ideal in the four-slot scenario, though measurable fluctuations emerge. These discrepancies highlight the stochastic effects originating from control imperfections and crosstalk. In both noisy scenarios, we observe a significant improvement in accuracy by querying more data points, highlighting the ability of quantum data to mitigate the effects of both classical and quantum noise in QRENN learning of Hamiltonian feature sets. {Furthermore, our results indicate that the QRENN model exhibits distinct robustness against quantum crosstalk noise compared to classical labeling errors. A possible explanation is that the model exhibits a strong dependence on data quality. With an increasing number of slots, classical labeling noise may exert a greater impact on the model's ability to correctly interpret the data, resulting in reduced classification performance.}

\subsection{Symmetry-protected topological phase (SPT)  detection}

We employ our QRENN approach to classify the symmetry-protected topological (SPT) phase in one-dimensional many-body systems, as illustrated in Figure~\ref{fig:detecting_SPT_phase}(a). By monitoring whether the measurement result stays near $0$, we can effectively distinguish the SPT phase. Based on this criterion, we define a set of binary labels $y_q$ as classification outputs (i.e., $y_q=0$ if the ground state $\ket{\psi}$ of the corresponding Hamiltonian falls in the SPT phase and otherwise, $y_q = 1$.) We consider a one-dimensional cluster-Ising model with periodic boundary conditions~\cite{smacchia2011statistical,li2022quantum}, described by
\begin{equation}~\label{cluster-Ising model with periodic boundary}
    H(\lambda) = -\sum_{j=1}^N X_{j-1} Z_{j} X_{j+1} + \lambda \sum_{j=1}^N Y_j Y_{j+1}.
\end{equation}
For $\lambda < 1$, the ground state of this system is in the cluster phase protected by a $\ZZ_2 \times \ZZ_2$ symmetry and exhibits a nonvanishing string order parameter. When $\lambda > 1$, it falls in the antiferromagnetic phase, and a continuous quantum phase transition occurs at $\lambda = 1$~\cite{smacchia2011statistical}.
\begin{figure}[h!]
    \centering
    \includegraphics[width=1.0\linewidth]{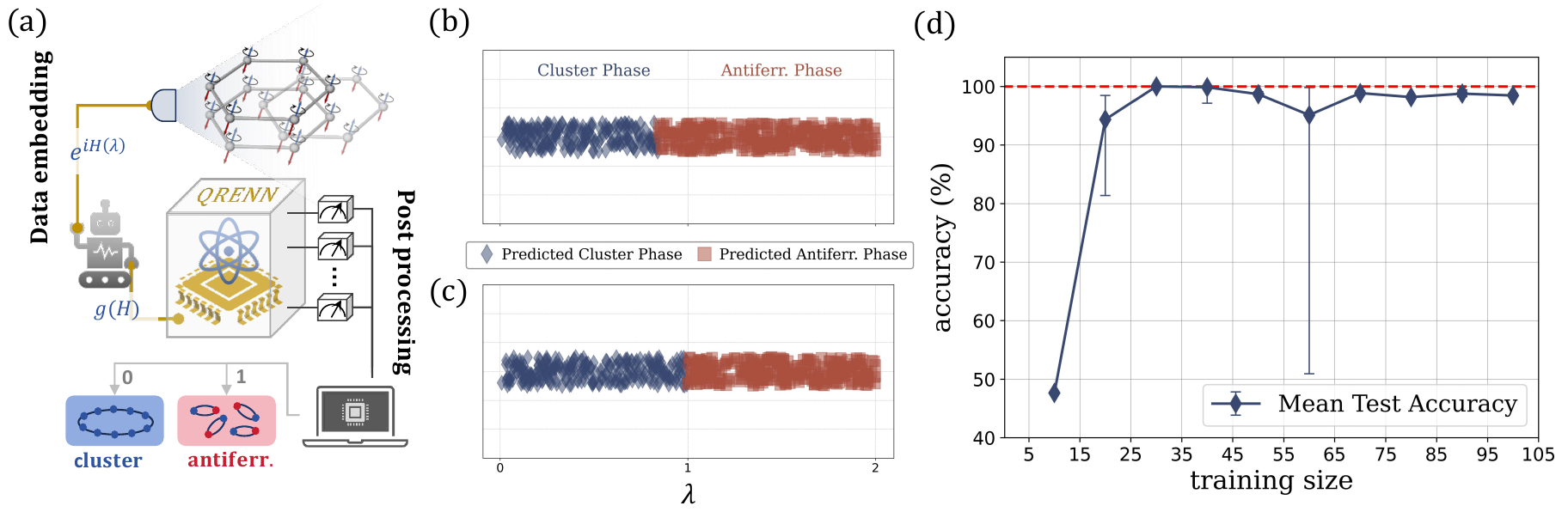}
    \caption{\textbf{Classification of SPT phases using supervised learning with QRENN.} In (a), the schematic framework illustrates the QRENN for identifying SPT phases. The data $H(\lambda)$ is embedded into the network as in Figure~\ref{fig:QRENN_circuit}.  In (b) - (c), classification results are obtained by QRENN for $500$ test data with initial states $\ket{0}^{\otimes 9}$ and $\ket{0} \otimes \ket{+}^{\otimes 8}$, respectively. {In (d), the mean test accuracy improves as the number of training samples which  are randomly selected from the same total dataset of $600$ increases, using initial states $\ket{0}\otimes\ket{+}^{\otimes 8}$.}}
    \label{fig:detecting_SPT_phase}
\end{figure}

To evaluate the performance of the QRENN model, we consider a quantum system comprising $n = 8$ qubits, preparing $40$ training samples and $560$ test samples with the Hamiltonian parameter $\lambda$ randomly sampled from the interval $[0,2]$. 
Within this configuration, the data processing register operates on $m = 1$ qubit,  the number of slots is $10$, and the initial quantum state is set as $\ket{0}^{\otimes 9}$. 
In the training set, the $q$-th data $H(\lambda_q)$ is assigned a label from $\{0,1\}$ according to the previous rules, and embedded into the model with the logic in Figure~\ref{fig:QRENN_circuit}. {Measurements performed on the data processing register are defined by Eq.~\eqref{eq:binary_labeling_op}, corresponding directly to the binary class labels. The total loss function is then formulated as shown in Eq.~\eqref{eq:total_loss}.}

Due to the quantum phase transition at $\lambda = 1$, measurement values near $1(\pm \; 1\times10^{-6})$ are empirically excluded from the analysis. Figure~\ref{fig:detecting_SPT_phase}(b) presents the results of detecting the SPT phase under the conditions described above. {While certain states in the antiferromagnetic phase, corresponding approximately to $\lambda \in [0.8, 1]$, are misclassified into the cluster phase, the overall classification accuracy remains high at $92.32\%$.}. Notably, changing the initial state from $\ket{0}^{\otimes 9}$ to $\ket{0} \otimes \ket{+}^{\otimes 8}$ markedly enhances classification performance, as shown in Figure~\ref{fig:detecting_SPT_phase}(c). {This modification raises the classification accuracy to $100\%$, completely distinguishing between states in the cluster and antiferromagnetic phases.} Both initial states have been numerically analyzed for their eigenspace overlaps with respect to the cluster-Ising model across different values of $\lambda$, consistently showing relatively large overlap magnitudes. Further details are provided in Appendix~\ref{app:joint_eig_overlap}.

Additionally, we investigate the influence of training size on QRENN model performance. To this end, we conduct $20$ independent experiments for each training size ranging from $10$ to $100$, {randomly selected from the same total dataset of $600$ samples}, and summarize the results in Figure~\ref{fig:detecting_SPT_phase}(d). The results clearly show an increasing trend in mean accuracy with larger training datasets, with performance stabilizing once the training set size exceeds {approximately $30$ to  $40$ samples}. Collectively, our findings indicate that the QRENN model reliably classifies the SPT phase in many-body quantum systems, and that its accuracy can be significantly improved by optimizing the initial state and expanding the training dataset.

\section{Concluding Remarks and Discussion}
We propose the Quantum Recurrent Embedding Neural Network (QRENN), a novel quantum neural network architecture for quantum supervised learning that incorporates both data processing and data embedding registers. This recurrent data embedding framework leverages the expressivity of general quantum circuit architectures and enables complex quantum information processing through repeated interactions between data and trainable parameters. Using DLA analysis, we show that the QRENN’s algebra decomposes as $\liedqnn \simeq \mathfrak{c} \oplus \mathfrak{su}(2^m)^{\oplus r}$, where $\mathfrak{c}$ represents the center and $r$ counts the number of distinct joint eigenspaces arising from the embedded data generators. We rigorously prove that this structure ensures a polynomial lower bound on the gradient variance under appropriate state overlap, thereby confirming the absence of barren plateaus in QRENN training. These theoretical results are validated numerically, demonstrating that QRENN maintains trainability even in moderately deep circuits. We further showcase the model’s practical utility by applying it to two challenging quantum supervised learning tasks: classifying quantum Hamiltonians and detecting SPT phases. In both cases, QRENN achieves high accuracy and robustness, with performance improving as more embedding slots which highlights its scalability and resilience to classical and quantum noise. 

Beyond the immediate performance advantages, QRENN offers a compelling case study in the broader discussion of the relationship between trainability and classical simulability, as first issued in~\cite{Cerezo2023does}. Although, in the complexity-theoretic landscape, whether or not our QRENN model can resist classification into $\Op{CSIM}$ or $\Op{CSIM_{QE}}$ under the proper initial state condition, still requires further discussion. Our observations suggest that QRENN occupies a unique position in the trainability-simulability landscape due to its connections to QSVT-based algorithms. The result and method open new avenues to explore quantum machine learning models that have scalability and potential quantum advantages.

There are several directions for further exploring the capabilities of QRENN. One important challenge is how to enhance the expressivity of the model while preserving its favorable trainability properties, considering significant applications in physics, machine learning, and beyond. This calls for the development of more flexible architectural designs or adaptive embedding strategies that maintain the Lie algebraic structure under control. Another promising direction is to explore a systematic characterization of QRENN from the perspective of quantum function transformation, akin to well-established frameworks such as multivariate quantum signal processing and quantum singular value transformation~\cite{Rossi2022multivariable}. Such an understanding would help clarify the functional landscape of QRENN and potentially guide the construction of more interpretable and powerful models. Beyond theoretical analysis, QRENN also opens up opportunities for broader applications in both classical and quantum supervised learning. Its architecture can be naturally extended to tasks in quantum many-body physics, quantum sensing, and general machine learning scenarios where recurrent embeddings or structured dynamics play a critical role. We anticipate that continued exploration of QRENN along these lines will contribute to scalable and principled quantum learning tools for applications with quantum advantages across disciplines.

{\paragraph{Acknowledgements --.} The authors would like to thank Benchi Zhao, Guangxi Li, Hongshun Yao, and the AQIS 2025 reviewers for thoughtful and insightful comments. 
{M.J., E.H., X.S. and X.W. were} partially supported by the National Key R\&D Program of China (Grant No.~2024YFB4504004), the National Natural Science Foundation of China (Grant. No.~12447107), the Guangdong Provincial Quantum Science Strategic Initiative (Grant No.~GDZX2403008, GDZX2403001), the Guangdong Provincial Key Lab of Integrated Communication, Sensing and Computation for Ubiquitous Internet of Things (Grant No. 2023B1212010007), the Quantum Science Center of Guangdong-Hong Kong-Macao Greater Bay Area, and the Education Bureau of Guangzhou Municipality.}

\bibliographystyle{unsrt}
\bibliography{main}

%%%%%%%%% SUPPLEMENTAL MATERIAL %%%%%%%%%

\newpage
\appendix
\setcounter{subsection}{0}
\setcounter{table}{0}
\setcounter{figure}{0}

\vspace{3cm}
% \onecolumngrid
% \vspace{2cm}

\begin{center}
\Large{\textbf{Appendix for ``Quantum Recurrent Embedding Neural Network''} \\ \textbf{
}}
\end{center}

\renewcommand{\theequation}{S\arabic{equation}}
% \numberwithin{equation}{section}
% \renewcommand{\thesubsection}{\normalsize{Supplementary Note \arabic{subsection}}}
\renewcommand{\theproposition}{S\arabic{proposition}}
\renewcommand{\thedefinition}{S\arabic{definition}}
\renewcommand{\thelemma}{S\arabic{lemma}}
\renewcommand{\thefigure}{S\arabic{figure}}
\setcounter{equation}{0}
\setcounter{table}{0}
\setcounter{section}{0}
\setcounter{proposition}{0}
\setcounter{definition}{0}
\setcounter{figure}{0}

% \tableofcontents
% %%%%%%%%%%%%%%%%%%%%%%%%%%%%%%%%%%%%%%%%%%%%%%%%%%%%%%%%%%%%%%%%%%%%%%%%%%%
\section{Notation and Background}

Let $\mathcal{H}$ be a finite-dimensional Hilbert space. Consider two parties, Alice and Bob, with associated Hilbert spaces $\mathcal{H}_A$ and $\mathcal{H}_B$, respectively, where the dimensions of $\mathcal{H}_A$ and $\mathcal{H}_B$ are denoted as $d_A$ and $d_B$. We use $\mathcal{L}^\dagger(\mathcal{H}_A)$ to represent the set of Hermitian operators acting on $\mathcal{H}_A$, $\mathcal{P}(\mathcal{H}_A)$ to represent the set of positive semidefinite operators acting on $\mathcal{H}_A$, and $\mathcal{D}(\mathcal{H}_A)$ to represent the set of all density operators on $\mathcal{H}_A$. The Lie algebra $\mathfrak{g}$ represents the set of generators for a dynamical Lie algebra, and $\spn_{\CC} \langle H_1, \dots, H_K \rangle_{\text{Lie}}$ denotes the complex span of the set of generators $\{H_1, \dots, H_K\}$ under the Lie bracket operation. $\mathcal{L}$ represents the loss function, while $\text{Var}_{\bm{\theta} \sim \nu}[\partial_{l,k}\mathcal{L}]$ and $\mathbb{E}_{\bm{\theta}, x}[\partial_{t, \mu} \mathcal{L}]$ represent the variance and expectation of the gradient of the loss function, respectively. The symbol $G$ represents a group with an element $g$, while $\text{GL}(d)$ is the general linear group of dimension $d$. The set $\mathcal{U}(d)$ represents the $d$-dimensional unitary group, and $U_g, V_g$ denote unitary representations of the group element $g$. The Lie algebras $\mathfrak{u}(d)$ and $\mathfrak{v}(d)$ represent Lie algebras of dimension $d$, and $\Omega_h, O_h$ represent Hermitian representations of Lie algebra elements $h$.

\section{Introduction to Lie groups, Lie algebra and Representation theory}\label{appendix:intro_lie_group_representation}

\subsection{Basics of Lie algebra}
In this section, we review the fundamental concepts of Lie algebras. A Lie algebra $\mathfrak{g}$ is a vector space over a field $\mathbb{F}$ equipped with a bilinear operation, called the Lie bracket, defined as $[\cdot, \cdot] : \mathfrak{g} \times \mathfrak{g} \rightarrow \mathfrak{g}.$ The Lie bracket satisfies the following properties for all elements $A, B, C \in \mathfrak{g}$. (1) \textbf{Antisymmetry}:$[A, B] = -[B, A]$ for any $A,B\in\mathfrak{g}$. Specifically, when $ A = B $, we have $ [A, A] = 0 $. (2) \textbf{Jacobi Identity}:
   \[
   [A, [B, C]] + [B, [C, A]] + [C, [A, B]] = 0.
   \]
for all $A,B,C \in\mathfrak{g}$. For matrix Lie algebras, the Lie bracket is defined by the commutator, $[A, B] = AB - BA$. These algebraic structures are fundamental in various domains, including quantum mechanics and the study of continuous transformation groups. 

The Baker-Campbell-Hausdorff (BCH) formula is a fundamental result in Lie algebra theory that expresses the logarithm of the product of two exponentials of Lie algebra elements in terms of their Lie brackets. Specifically, for elements $A, B \in \frak{g}$, the BCH formula provides an expression for $\log(e^A e^B)$ as a series involving commutators of $X$ and $Y$:
\begin{equation}
    \log(e^A e^B) = A + B + \frac{1}{2}[A,B] + \frac{1}{12}[A,[A,B]] - \frac{1}{12}[B,[A,B]] + \cdots.
\end{equation}
This series, while generally infinite, converges under certain conditions and plays a crucial role in connecting the algebraic structure of Lie algebras with the analytic structure of Lie groups. The BCH formula is instrumental in various applications, including the study of Lie group representations, deformation theory, and quantum mechanics, where it facilitates the combination of infinitesimal generators.

We say $\frak{g}$ is \textit{commutative} if for any pairs of $A,B$, $[A,B] = 0$. A \textit{Lie subalgebra} $\frak{h}\subseteq \frak{g}$ is a linear subspace of $\frak{g}$ which is closed under the Lie bracket. A Lie algebra is said to be \textit{compact} if it is a Lie algebra of a compact Lie group.

For matrix Lie algebras, the usual Lie bracket is given by the commutator $[A,B] = AB - BA$. For an $N$-dimensional real matrix Lie algebra with basis $\{E_1,\cdots, E_N\}$ that are independent over $\CC$, the \textit{complexification} of $\frak{g}$ ($\frak{g}_{\CC}$) is defined as the Lie algebra of the same basis, spanning in the complex field,
\begin{equation}
    \frak{g}_{\CC}:=\spn_{\CC}\{E_1, \cdots, E_N\}.
\end{equation}
A Lie algebra $\frak{g}$ can be decomposed into the direct sum of subalgebras $\frak{g}_1, \frak{g}_2,\cdots,\frak{g}_n$, as $\frak{g} = \oplus_{j=1}^n \frak{g}_j$, where `$\oplus$' denotes the direct sum of vector spaces. An \textit{ideal} of a Lie algebra $\frak{g}$ is a Lie subalgebra $\frak{i}\subseteq \frak{g}$ so that for all $A\in \frak{g}$ and $B\in \frak{i}$, $[A,B]\in\frak{i}$. Every Lie algebra has two \textit{trivial} ideals, i.e., $\frak{i}=\spn\{\bm{0}\}$ or $\frak{i} = \frak{g}$. Two special ideals are commonly discussed: (1) center of $\frak{g}$; (2) commutator ideal,
\begin{equation}
    \frak{c}(\frak{g}) := \{A\in\frak{g}:[A,B]=0, \; \forall B\in\frak{g}\},
\end{equation}
\begin{equation}
    [\frak{g},\frak{g}] := \spn\{[A,B]:A,B\in \frak{g}\}.
\end{equation}
A \textit{simple Lie algebra} is a Lie algebra that is non-abelian, containing no non-trivial ideals. A \textit{semisimple} Lie algebra can be decomposed as a direct sum of simple Lie algebras.

\begin{theorem}[\cite{Knapp1988LieGB}]
    Any subalgebra $\frak{g}\subseteq \frak{sl}(n,\FF)$, for $\FF = \RR$ or $\CC$ has the following decomposition, 
    \begin{equation}
        \frak{g} = \frak{c} \oplus \frak{g}_1 \oplus \cdots \oplus \frak{g}_n,
    \end{equation}
    where $\frak{c}$ is the center of $\frak{g}$, and each $\frak{g}_j$ is a simple Lie algebra over $\FF$.
\end{theorem}
If we have $\frak{g}$ been written in the above form, then we have $[\frak{g},\frak{g}] = \frak{g}_1 \oplus \cdots \oplus \frak{g}_n$ is semisimple.

\subsection{Representations and Norms}
Given $V$ a finite-dimensional inner product space over the field $\FF \in \{\RR,\CC\}$, and $\cU(V)$ the isometry group on $V$. Denote $\frak{u}(V)$ the algebra of skew-Hermitian operators on $V$, which is isomorphic to $\CC^n$ for some $n$. A unitary (Lie) group representation of $G$ is a smooth homomorphism $\phi:G\rightarrow \cU(V)$, which is also faithful if it is an injection. Such a representation makes the linear space $V$ a $G$-module, and we have then related the representation to the linear space.

A \textit{subrepresentation} of $\phi$ defines an invariant subspace of $V$, and we say $\phi$ a \textit{irreducible} if it has no non-trivial subrepresentations. Based on the Peter-Weyl theorem, every compact Lie group has a faithful, finite-dimensional unitary representation. Furthermore, any representation of a compact Lie group can be decomposed into a direct sum of irreducible representations. The representations of simple Lie algebras are always faithful or trivial.

Any map $f:V\rightarrow W$ connecting two representations $\phi:G\rightarrow \Op{GL}(V)$ and $\psi:G\rightarrow \Op{GL}(W)$ is called \textit{equivariant} if $f(\phi(g)v) = \sigma(g)f(v)$ for any $g\in G$ and $v\in V$, and if $f$ is bijective, then the representations are \textit{isomorphic as} $G$-modules. 

\begin{lemma}[Schur's Lemma~\cite{hall2000elementary}]
    Let $V$ and $W$ be vector spaces, and let $\rho_V$ and $\rho_W$ be irreducible representations of $G$ on $V$ and $W$, respectively, over an algebraically closed field. Let $f:V\rightarrow W$ be a $G$-equivariant linear map, then 
    \begin{enumerate}
        \item $f$ is either zero or an isomorphism.
        \item If $V=W$ and $\rho_V=\rho_W$, then the only non-trivial $G$-equivariant linear map is homothety, i.e., scalar multiples of the identity.
    \end{enumerate}
\end{lemma}

The group representation can induce a \textit{Lie algebra representation} $d\phi:\frak{g}\rightarrow \frak{u}(V)$, which, from the language of differential geometry, is the differential of the smooth map $\phi$. Lie algebra representations with respect to addition (as for the group operations), and the Lie bracket satisfying
\begin{equation}
    [d\phi(x), d\phi(y)] = d\phi([x,y]).
\end{equation}

Any compact Lie group $G$ containing unitary matrices has its \textit{standard representation}, which is the natural action on $V = \CC^n$. The standard representation of Lie algebra is defined in the same way by replacing unitary matrices with skew-Hermitian matrices. The \textit{adjoint representation} $\phi$ consists of $G$ acting on its Lie algebra is defined as,
\begin{equation}
    \phi(g)h = ghg^{-1}, \forall h\in \frak{g},
\end{equation}
and $V_\phi = \frak{g}$. The associated \textit{Lie algebra adjoint representation} is then given by $d\phi(h)k = [h,k]\; \forall k\in \frak{g}$. A special representation used in QML is so-called the \textit{tensor power} of a representation $\phi$ denoted as $\phi \ox \phi$ and,
\begin{equation}
    (\phi\ox\phi)(g) := \phi(g)\ox\phi(g) \; \forall g\in G,
\end{equation}
where the corresponding Lie algebra representation is given by $d\phi\ox d\phi$, defined as,
\begin{equation}
    (d\phi\ox d\phi)(h) := d\phi(h) \ox I + I\ox d\phi(h).
\end{equation}
In the following, we denote $\Ad$ and $\ad$ be the adjoint representations of Lie group and Lie algebra. 

Before defining the norms, we first introduce the trace form for the Lie algebra $\frak{g}$. Given the orthonormal basis for $\frak{g}$ denoted as $\{e_i\}$ the \textit{standard trace form} is derived as,
\begin{equation}
    -\tr(e_i e_j) = \delta_{ij}
\end{equation}
where the negative sign is to ensure the positive definiteness for the skew-Hermitian operators. More generally, the \textit{Killing form} is defined as a symmetric bilinear form,
\begin{equation}
    \kappa(x,y)_{\frak{g}} := \tr(\ad_x \circ \ad_y),
\end{equation}
for any $x,y\in \frak{g}$. We say the Killing form is \textit{non-degenerate} if for every non-zero $x\in\frak{g}$, there exists a $y\in\frak{g}$ s.t., $\kappa(x,y)_{\frak{g}} \neq 0$.
\begin{lemma}[Cartan's Criterion]
    The Killing form is non-degenerate for a real (complex) Lie algebra $\frak{g}$ if and only if $\frak{g}$ is semisimple.
\end{lemma}
For any representation $\phi$ of the compact simple Lie group $e^{\frak{g}}$, the \textit{index of the representation} $I_{\phi}$ is defined in the equation,
\begin{equation}
    -\tr(d\phi(e_i)d\phi(e_j)) = I_{\phi} \delta_{ij}.
\end{equation}
The \textit{standard norm} is then induced by the standard trace form $\norm{h}_{\frak{g}}^2 = -\tr(h^2)$. In particular, for the matrix Lie algebras, this reduces to the Frobenius norm.

\section{Dynamical Lie Algebra of Modified Periodic Models}\label{sec:appendix_dla_modified_periodic_models}
Refer to~\cite{Ragone2024lie,Fontana2024characterizing}, one can analyse the BP phenomenon via  DLA. The $n$-qubit DLA is defined as the Lie closure of the circuit's generators,
\begin{equation}
    \frak{g}=\spn_{\RR}\langle i\Omega_1, i\Omega_2, \cdots, i\Omega_N \rangle_{\Op{Lie}},
\end{equation}
which is closed under commutation and forms a subspace of $\frak{u}(2^n)$ of all skew-Hermitian operators. As $\frak{g}$ is a subalgebra of $\frak{u}(2^n)$, it is automatically a reductive Lie algebra~\cite{Knapp1988LieGB}, and we can always decompose it into the direct sum of ideals,
\begin{equation}
    \frak{g} = \frak{c}\oplus \bigoplus_{j=1}^{k-1} \frak{g}_j,
\end{equation}
where each $\frak{g}_j$ is simple subalgebra and $\frak{c}$ is abelian, i.e., $\frak{c}$ is the  center of $\frak{g}$ and $[\frak{g},\frak{g}] = \bigoplus_{j=1}^{k-1} \frak{g}_j$ is semisimple, called commutator ideal of $\frak{g}$. We suppose the parameterized circuit is deep enough to form an approximate $2$-design on each of the components $e^{\frak{g}_j}$, which will allow us to compute the variance via the Haar integration over the Lie group.

Let's delve into a rigorous mathematical argument to elucidate why fixed generator are inherently involved in the DLA, even when they are not associated with any trainable parameters.

\subsection{Lie Algebraic representation of QRENN}
We first consider the specific structure of QRENN with $m$-qubit data processing register and $n$-qubit data embedding register. Suppose $g_c$ is defined as stated in Sec.~\ref{sec:trainability_dqnn} for some control state $\ket{c}\in\cH_A$, i.e,
\begin{equation}\label{appendix:def_gc}
    g_c(V) = \ketbra{c}{c} \ox V + (I_{2^m} - \ketbra{c}{c}) \ox I_{2^n}.
\end{equation}
More generally, one can extend the  controlling space and define the map,
\begin{equation}\label{appendix:def_g_projection}
    g_{\Pi}(V) = \Pi \ox V + (I_{2^m} - \Pi) \ox I_{2^n}
\end{equation}
where $\Pi$ is some projection of rank less than $2^m$. Let $V = e^{iH}$ for some Hermitian generator $\Omega\in i\frak{u}(2^n)$, then we could express,
\begin{equation}
\begin{aligned}
    \Op{CV} &= \Pi\ox e^{iH} + (I_{2^m} - \Pi)\ox I_{2^n}\\
    &= \Pi \ox \left(\sum^{\infty}_{k=0}\frac{i^k}{k!}H^k\right) + (I_{2^m} - \Pi)\ox I_{2^n}\\
    &=\Pi\ox\left(\sum^{\infty}_{k=1}\frac{i^k}{k!}H^k\right) + (\Pi + I_{2^m} - \Pi)\ox I_{2^n} \\
    &= \sum^{\infty}_{k=0}\frac{i^k}{k!}(\Pi\ox H)^k = e^{i\Pi \ox H}.
\end{aligned}
\end{equation}
To compute the total DLA of our QRENN model, let us start with the $t$-th processing layer $W(\bm{\theta}_t) = \prod_{l=1}^L e^{i\theta_{t,l}\Omega_{l}}$, where $\bm{\theta}_t=(\theta_{t,l})_l$ and $i\Omega_1, \cdots, i\Omega_L\in \frak{su}(2^m)$. Since the control projection $\Pi$ can be evolved under any local unitary from $\Op{SU}(2^n)$ acting on the data processing register. Without loss of generality, we can set $\Pi = \sum_{j=0}^K \ketbra{j}{j}$ to be some diagonal matrices with $K<2^m$. 
% in the following discussion. One can expand the projection $\ketbra{c}{c}$ as
% \begin{equation}
%     \ketbra{c}{c} = \frac{1}{2^m}(I-Z)^{\ox m} = \frac{1}{2^m}\sum_{\bm{b}\in\{0,1\}^m} (-1)^{\sum_{j=1}^{m} b_i} \bigotimes_{j=1}^m Z^{b_j},
% \end{equation}
% where $\bm{b} = (b_1, \dots, b_m)$, with $b_i \in \{0, 1\}$ for $i \in \{1, \dots, m\}$, and for the term $Z^{\bm{b}} = Z^{b_1} \otimes \cdots \otimes Z^{b_m}$. Additionally, when $\bm{b} = 0$, this implies that $b_i = 0$ for all $i \in \{1, \dots, m\}$. By extracting the term of $\bm{b} = \bm{0}$ out and expressing,
% \begin{equation}
%     \ketbra{c}{c} = \frac{1}{2^m}\left(I_{2^m} + \sum_{\bm{b}\neq \bm{0}} (-1)^{\sum_{i=1}^m b_i} Z^{\bm{b}}\right).
% \end{equation}
% We, therefore, can derive the $g_c(V)$ gate in terms of its generator,
% \begin{equation}\label{eq:decompose_cv}
% \begin{aligned}
%     \exp\left(i\ketbra{c}{c}\ox H\right) &= \exp\left(\frac{i}{2^m}\left(I_{2^m} + \sum_{\bm{b}\neq\bm{0}} (-1)^{\sum_{i=1}^m b_i} Z^{\bm{b}}\right)\ox H\right)\\
%     &= \exp\left(\frac{i}{2^m}(I_{2^m} \ox H) + \frac{i}{2^m}\left(\sum_{\bm{b}\neq\bm{0}} (-1)^{\sum_{i=1}^m b_i} Z^{\bm{b}}\ox H\right)\right)\\
%     &=\exp\left(\frac{i}{2^m}(I_{2^m} \ox H)\right) \exp\left( \frac{i}{2^m}\left(\sum_{\bm{b}\neq\bm{0}} (-1)^{\sum_{i=1}^m b_i} Z^{\bm{b}}\ox H\right)\right),
% \end{aligned}
% \end{equation}
% where the last equality holds due to the commutative relation between the generators.

Before further discussion, let us modify our QRENN model a bit in order to match the usual setup of QNN. Suppose now we are allowing embedding multiple data into a $T$-slot QRENN where $\bm{\theta} = (\bm{\theta}_1, \cdots, \bm{\theta}_T)$ are the tunable parameters and $\bm{\varphi} = (\varphi_1, \cdots, \varphi_{T})$ attached to the data vectors. The total QRENN circuit can be represented as a parameterized circuit $\bm{U}(\bm{\theta}, \bm{\varphi})$,
\begin{equation}\label{eq:dqnn_circuit_appendix}
    \bm{U}(\bm{x};\bm{\theta}, \bm{\varphi}) = (W(\bm{\theta}_{T+1})\ox I_{2^n})(\prod_{t=1}^{T} g_{\Pi}(U_t(\bm{x};{\varphi_t}))(W(\bm{\theta}_t)\ox I_{2^n})).
\end{equation}
The embedding map $g_{\Pi}$ is defined the same as in Eq.~\eqref{appendix:def_g_projection}. The cost function was chosen to be the expectation value of some observable $O\in i\frak{u}(2^{m+n})$ with respect to the resulting state $\ket{\psi(\bm{x};\bm{\theta}, \bm{\varphi})} = \bm{U}(\bm{x};\bm{\theta}, \bm{\varphi})\ket{0}^{\ox m}\ket{0}^{\ox n}$. In particular, let $O = O_m \otimes I_{2^n}$ where $O_m$ is a Hermitian operator in $i\frak{u}(\cH_{2^m})$. We assume an input state $\rho = \rho_m\ox \rho_n$ acting on a Hilbert space $\cH \simeq (\CC^2)^{\ox n+m}$, and the cost function can be expressed as,
\begin{equation}
% \label{eq:local_loss_dqnn}
    \ell(O_m \otimes I_{2^n}, \bm{x};\bm{\theta}, \bm{\varphi}) = \tr\left(\bm{U}(\bm{x};\bm{\theta}, \bm{\varphi})(\rho_A\ox \rho_B)\bm{U}^{\dagger}(\bm{x};\bm{\theta}, \bm{\varphi})(O_m \otimes I_{2^n})\right).
\end{equation}
Without loss of generality, we set both $W(\bm{\theta}_t)$ and $U_t(\bm{x};\bm{\varphi}_t)$ to be the periodic QNN template, and $W(\bm{\theta}_t) = \prod_{l=1}^L e^{i\theta_{t,l} \Omega_{l}}$, $U_t(\bm{x};\varphi_t) =  e^{i\varphi_{t} H_q}$ for some Hermitian generators $\Omega_l \in i\frak{u}(2^m)$ and $H_t = H_t(\bm{x}) \in i\frak{u}(2^n)$, encoded with data $\bm{x}$. The explicit form of $\bm{U}(\bm{x};\bm{\theta}, \bm{\varphi})$ can be then represented as,
\begin{equation}
    \bm{U}(\bm{x};\bm{\theta}, \bm{\varphi}) = \prod_{t=1}^{T}\left(\prod_{k=1}^Ke^{i\varphi_{t,k}\Pi \ox H_t}\cdot\prod_{l=1}^L e^{i\theta_{t,l} \Omega_{l} \ox I_{2^n}}\right).
\end{equation}
With our previous expansion form of the embedding map, one can express the total DLA of the modified QRENN model as,
\begin{equation}
    \liedqnn := \spn_{\RR}\langle \; i\Pi\ox H_t, i\Omega_l \ox I_{2^n}\;:\forall t,l \; \rangle_{\Op{Lie}},
\end{equation}
\begin{lemma}\label{lem:g_dqnn_hp}
    Let $\{\Omega_l\}_l$ and $\{H_k\}_k$ be the sets of Hermitian generators where $\spn_{\RR}\liebracket{i\Omega_l} = \frak{su}(d)$ and $H_k$ are some Hermitian matrices of dimension $d'$. Let $\Pi$ be a projection of rank $r<d$. Then,
    $$\spn_{\RR}\langle \; i\Pi\ox H_k, i\Omega_l \ox I_{d'}\;:\forall k,l \; \rangle_{\Op{Lie}} = \spn_{\RR}\langle \; iI_d \ox H_p, i\Omega_q\ox H_r, i\Omega_s\ox I_{d'}\;:\forall p,q,r,s \; \rangle_{\Op{Lie}}.$$
\end{lemma}
\begin{proof}
    Denoting the Lie algebra on the left as $\frak{h}$ and the one on the right as $\liedqnn$, we show $\liedqnn\subseteq \frak{h}$ and $\frak{h} \subseteq \liedqnn$.
    Let $\{i\hat{\Omega}_l\}_l$ be a basis of $\frak{su}(d)$ comprising skew-Hermitian matrices from $\{i\Omega\}_l$ and their nested commutators. Since $\Pi$ is a non-trivial projection ($\Pi\neq I_d$), there exists real coefficients $\{\alpha_j\}_{j=0}^{d^2-1}$ where $\alpha_0\neq 0$ and at least one $\alpha_j \neq 0$ for $j \geq 1$, such that:
    \begin{equation*}
    \begin{aligned}
        \Pi = \alpha_0 I_d + \sum_{j=1}^{d^2-1} \alpha_j \hat{\Omega}_j = \sum_{j=0}^{d^2-1} \alpha_j \hat{\Omega}_j,
    \end{aligned}
    \end{equation*}
    where we denote $\hat{\Omega}_0 = I_d$. To show $\liedqnn\subseteq \frak{h}$, 
    we need to show $i\Pi\ox H_k$ and $i\Omega_l\ox I_{d'}$ can be represented by elements in $\frak{h}$. Note that $i\Omega_l \ox I_{d'} \in \frak{h}$ by definition. Similarly, $i\hat{\Omega}_l \ox H_k \in \frak{h}$ since $\frak{h}$ contains all linear combinations of $iI_d \ox H_p$, $i\Omega_q \ox H_r$, $i\Omega_s\ox I_{d'}$ and their nested commutators. Therefore,
    \begin{equation*}
        i\Pi\ox H_k = \alpha_0 (iI_d \ox H_k) + \sum_{j=1}^{d^2-1} \alpha_j (i\hat{\Omega}_j \ox H_k) \in \frak{h}.
    \end{equation*}
    Due to the linearity of commutator, we have all nested commutators in $\liebracket{i\Pi\ox H_k, i\Omega_l \ox I_{d'}, \; \forall k,l}$ contained in $\frak{h}$
    and we have proven $\liedqnn\subseteq \frak{h}$.
    
    The other direction is nontrivial. By construction, the elements $i\O_s \ox I_{d'}$, so does $i\hat{\O}_s \ox I_{d'}$, automatically in  $\liedqnn$ as before. Then, we aim to show that the other two generators can also be represented in $\liedqnn$. Notice that, $[A\ox B, C\ox D] = [A,C] \ox BD + CA \ox [B,D]$, and,
    \begin{equation*}
        [i\Pi\ox H_k, i\hat{\Omega}_l\ox I_{d'}] = - [\Pi, \hat{\Omega}_l] \ox H_k \in \liedqnn.
    \end{equation*}
    Defining $\Delta_{l} = -[\Pi, \hat{\Omega}_l]$, we have $\Delta_{l}\ox H_k \in\liedqnn$. Since $\Pi = I_{r} \oplus 0_{d-r}$ is a projection, we can use the Cartan decomposition of $\frak{su}(d) = \frak{l} \oplus \frak{p}$, where $\frak{l}$ contains block-diagonal and $\frak{p}$ contains block-off-diagonal matrices. With this decomposition:
    % One can denote $\Delta_{l} = -[\Pi, \hat{\Omega}_l]$, and for every $l,k$, $\Delta_{l}\ox H_k \in\liedqnn$. Notice that $\Delta_l \in \frak{su}(d)$ as any commutators are skew-Hermitian and traceless. Since $\Pi\prec I_{d}$ is a projection, one can write $\Pi = I_{r} \oplus 0_{d-r}$ where $I_{r}$ is an identity on the $r$-dimensional subspace (with $1\leq r < d$). Relative to this decomposition, for every element $X \in \frak{su}(r+d)$ written in block form:
    % \begin{equation*}
    % X = \begin{pmatrix}
    %         X_{11} & X_{12}\\
    %         X_{21} & X_{22}
    % \end{pmatrix} = 
    % \begin{pmatrix}
    %         X_{11} & 0\\
    %         0 & X_{22}
    % \end{pmatrix} +
    % \begin{pmatrix}
    %         0 & X_{12}\\
    %         X_{21} & 0
    % \end{pmatrix},
    % \end{equation*}
    % where the first term lies in the block-diagonal (Lie) subalgebra $\frak{l}$ and the second term lies in the block-off-diagonal subalgebra $\frak{p}$ so that, $\frak{su}(r+d) = \frak{l} \oplus \frak{p}$. This is called a \textit{Cartan decomposition} of $\frak{su}(r+d)$, satisfying 
    % \begin{equation*}
    %     [\frak{l},\frak{l}]\subset \frak{l}; \; [\frak{l},\frak{p}]\subset \frak{p}; \; [\frak{p},\frak{p}]\subset \frak{l}.
    % \end{equation*}
    % The projection $\Pi$ commutes with block-diagonal matrices (i.e. those with $X_{12} = X_{21} = 0$) and does not commute with an off-diagonal matrix. In fact,
    \begin{equation*}
        [\Pi, X] = 
        \begin{pmatrix}
            0 & X_{12}\\
            -X_{21} & 0
        \end{pmatrix}.
    \end{equation*}
    Thus, $\Op{ad}_{\Pi}$ annihilates block-diagonal parts while acting injectively on off-diagonal parts. Since $\{i\hat{\Omega}_l\}_l$ spans $\frak{su}(d)$, the set $\{\Delta_l\}_l$ spans the subspace $\frak{p}$.
    % Thus, the adjoint representation $\Op{ad}_{\Pi}$ annihilates the block-diagonal part of $X$ and acts (injectively, up to constants) on the off-diagonal part. Because the collection $\{i\hat{\Omega}_l\}_l$ forms a basis of $\frak{su}(d)$, it spans both the block-diagonal and block-off-diagonal matrices. In particular, some of the $\hat{\Omega}_l$ will have nonzero off-diagonal blocks. For any such $\hat{\Omega}_l$, the commutator $\Delta_l$ produces an operator which is supported only in the block-off-diagonal blocks. By varying over the entire basis, it follows that the set $\{\Delta_l\}_l$ spans the entire subspace $\frak{p}$. 
    
    Now, to span the subalgebra of $\frak{l}$, we we analyze further commutators,
    \begin{equation*}
        [\Delta_{l}\ox H_k, i\hat{\Omega}_p \ox I_{d'}] = i[\Delta_{l}, \hat{\Omega}_p] \ox H_k.
    \end{equation*}
    % Denoting the set,
    % \begin{equation*}
    %     \cC:= \{[\Delta_l, i\hat{\Omega}_p], \;\forall l,p\},
    % \end{equation*}
    Letting $\cC = \{[\Delta_l, i\hat{\Omega}_p]\}_{l,p}$, we show that $\cC$ contains a basis for all block-diagonal skew-Hermitian traceless matrices. For block-off-diagonal matrices $A,B \in \CC^{r\times (d-r)}$,
    % where we can show that from above that the $\liedqnn$ contains all the generators $Y\ox H_k$ for all $k$ and some $Y\in\cC$. We now construct the basis for all block-diagonal skew-Hermitian, traceless matrices in the following sense: For any complex matrix $A,B\in \CC^{r\times d - r}$, the commutator,
    \begin{equation*}
        \left[
        \begin{pmatrix}
            0 & A\\
            -A^{\dagger} & 0
        \end{pmatrix},
        \begin{pmatrix}
            0 & B\\
            -B^{\dagger} & 0
        \end{pmatrix}
        \right] = 
        \begin{pmatrix}
            -AB^{\dagger} + BA^{\dagger} & 0\\
            0 & -A^{\dagger} B + B^{\dagger} A
        \end{pmatrix},
    \end{equation*}
    can be represented as a linear combination of $[\Delta_l, i\hat{\Omega}_p]$. To see this, one can represent the first matrix using $\Delta_l$'s and the second one using $i\hat{\Omega}_p$'s. The linearity of the Lie bracket then guarantees the statement.
    
    We can specifically choose $(A,B) = (E_{jk}, iE_{jk})$ where 
    $1\leq j\leq r, r< k \leq d$, where $E_{jk}$ is the matrix unit with a $1$ at the $(j,k)$-entry and $0$ elsewhere, so that,
    \begin{equation*}
        \begin{pmatrix}
            iE_{jk}E_{kj} + iE_{jk}E_{kj} & 0\\
            0 & -iE_{kj} E_{jk} - iE_{kj} E_{jk}
        \end{pmatrix} = 
        2\begin{pmatrix}
            iE_{jj} & 0\\
            0 & -iE_{kk}
        \end{pmatrix}. 
    \end{equation*}
    Similarly, let 
    $(A,B) = (E_{jk}, E_{lk})$ for $j\neq l$,
    \begin{equation*}
        \begin{pmatrix}
            -E_{jk}E_{kl} + E_{lk}E_{kj} & 0\\
            0 & -E_{kj}E_{lk} + E_{kl}E_{jk}
        \end{pmatrix} = 
        \begin{pmatrix}
            E_{lj} - E_{jl} & 0\\
            0 & 0
        \end{pmatrix}, 
    \end{equation*}
    and $(A,B) = (E_{jk}, iE_{lk})$ for $j\neq l$,
    \begin{equation*}
        \begin{pmatrix}
            iE_{jk}E_{kl} + iE_{lk}E_{kj} & 0\\
            0 & -iE_{kj}E_{lk} -iE_{kl}E_{jk}
        \end{pmatrix} = 
        \begin{pmatrix}
            i(E_{jl} + E_{lj}) & 0\\
            0 & 0
        \end{pmatrix}. 
    \end{equation*}
    Notice that by ranging over all $j,k,l$ values, the set of $\{iE_{jj}, E_{lj}-E_{jl}, i(E_{jl} + E_{lj})\}$, having $r(r-1) + r = r^2$ independent elements, forms an explicit basis for $r\times r$ skew-Hermitian matrices forming a real vector space of dimenion $r^2$. Similar strategy applies to the bottom diagonal block by swapping the indices of $j,k$ and $l,k$. As a result, for any skew-Hermitian $C\in\frak{su}(r), D\in \frak{su}(d - r)$ such that $\tr(C+D) = 0$, one can find them as a real linear combination of these basis elementary matrices where the traceless condition applies by noticing that $\tr(E_{jj} - E_{kk}) = 0$. We then have constructed the basis elements for $\frak{l}$ as linear combinations of elements in $\cC$.

    % \textbf{Up-to-now}, we have shown that any element in $\frak{su}(d)$ can be written as a real linear combination of  $\Delta_l$ (block-off-diagonal) and the $[\Delta_l, i\hat{\Omega}_p]$ (block-diagonal) components. Therefore, the generator $i\Omega_q\ox H_r \in \frak{h}$ for any fixed $q,r$ can be formed by the (real) linear combination of $\Delta_l\ox H_r$ and $[\Delta_j, i\hat{\Omega}_k]$ in terms of index $l,j,k$, which implies $i\Omega_q\ox H_r \in \liedqnn$. At last, for any $p$, $I_{d}\ox H_p$ has the following decomposition,
    Therefore, for any fixed $q,r$, the generator $i\Omega_q\ox H_r \in \frak{h}$ can be formed by linear combinations of elements in $\liedqnn$. Finally, for any $p$,
    \begin{equation*}
        iI_{d}\ox H_p = \frac{1}{\a_0}i\Pi\ox H_p - \sum_{j=1} \frac{\a_j}{\a_0}i\hat{\Omega}_j \ox H_p.
    \end{equation*}
    Based on the above argument, for any fixed $p$, it is clear to say the terms $i\Pi\ox H_p$ and $i\hat{\Omega}_j\ox H_p$ for any $p,j$ are in $\liedqnn$ so that the linear combination of them must also lie in $\liedqnn$.
    Therefore, we have $iI_d \ox H_p$ in $\liedqnn$. \textbf{Above all}, all generators of $\frak{h}$ can be represented as the (real) linear combination of the generators and the nested commutators in $\liedqnn$, we conclude $\frak{h}\subseteq \liedqnn$ and hence proves the equality.
\end{proof}

Notably, we can make a basis transformation by turning $\Omega_l$ to the $2^m \times 2^m$ generalized $m$-qubit Pauli matrices, denoted as $P_l$, forming an orthogonal basis of $i\frak{su}(2^m)$. Therefore, $\liedqnn = \spn_{\RR}\langle iI_{2^m} \ox H_p, iP_q\ox H_r, iP_s\ox I_{2^n}\;:\forall p,q,r,s \rangle_{\Op{Lie}}$. Clearly, one can use these new generators to identify more interesting structures of $\liedqnn$. One can observe that the structure of $\liedqnn$ depends on the commutators between any pairs of the generators in the above. For example, 
\begin{itemize}
    \item $[iI_{2^m} \ox H_p, iI_{2^m} \ox H_q] = (i)^2 I_{2^m}\ox [H_p, H_q] = -I_{2^m} \ox [H_p, H_q]$;
    \item $[iI_{2^m} \ox H_p, iP_q \ox H_r] = (i)^2 P_q\ox [H_p, H_r] = - P_q\ox [H_p, H_r]$;
    \item $[iP_s \ox I_{2^n}, iP_q \ox H_r] = (i)^2 [P_s, P_q]\ox H_r = - [P_s, P_q]\ox H_r$;
    \item $[iP_s \ox I_{2^n}, iP_q \ox I_{2^n}] = (i)^2 [P_s, P_q]\ox I_{2^n} = - [P_s, P_q]\ox I_{2^n}$;
    \item $[iP_p \ox H_q, iP_r \ox H_s] = - (P_pP_r \ox H_qH_s - P_rP_p \ox H_sH_q)$.
\end{itemize}
These five commutators can further produce much more complicated terms with nested Lie brackets. The last commutator makes it rather difficult to determine the structure of $\liedqnn$, in general. 

\subsection{QRENN with fixed data}
In quantum many-body physics, controlling a fixed quantum evolution is crucial for predicting fundamental properties of the quantum system. Suppose a fixed evolution carrying the information of the quantum system $V = e^{iH \varphi_t}$ through the embedding map, where $H$ is some fixed Hermitian operator acting on $i\frak{u}(\cH_{2^n})$. The following discussion showcases that $\liedqnn$ can be decomposed into blocks with limited dimension.

\begin{definition}
    Given any Hermitian operator $H$, the power set of $H$, denoted as $\cP_H$, is defined as, $\cP_H := \{M: M = H^k, k = 0, 1, \cdots, \infty\}$. 
\end{definition}

\begin{lemma}
    $\spn_{\FF}\liebracket{\cP_H}$ is an \textit{abelian} Lie algebra, and $\spn_{\FF}\liebracket{\cP_H} = \spn_{\FF}\cP_H$ for $\FF \in \{\CC, \RR\}$.
\end{lemma}
\begin{proof}
    Since $H^k$ can be simultaneously diagonalizable, which automatically proves that $\spn_{\CC}\liebracket{\cP_H}$ is abelian, making it isomorphic to $\RR^r$ for some integer $r$. Moreover, $[H^m, H^n] = 0$ for any integers $m,n>1$, the Lie closure does not expand the space, which leads to $\liebracket{\cP_H} = \cP_H$.
\end{proof}

\begin{lemma}\label{lem:dim_PO}
    Given $H\in i\frak{u}(d)$ with $r$ distinct real eigenvalues, then $\dim[\spn_{\RR}\cP_H] = r$.
\end{lemma}
\begin{proof}
    The proof directly follows from the property of minimal polynomial $m_H(x) = 0$ where,
    \begin{equation}
        m_H(x) := \prod_{j=1}^r(x - \lambda_j),
    \end{equation}
    and the degree of $m_H(x)$ is $r$; $\lambda_j$ is a distinct eigenvalue of $H$. Then, any polynomial $P(x)$ of degree greater or equal to $r$ can be reduced modulo $m_H(x)$ to a polynomial of degree less than $r$. On the other hand, let us assume there exists $c_0, \cdots, c_r \in \RR$ such that $\sum_{j=0}^{r} c_j H^j = 0$. Then, by applying this operator to $H$'s eigenvectors to have,
    \begin{equation}
        \left(\sum_{j=0}^{r} c_j H^j\right) \ket{\lambda_s} = \left(\sum_{j=0}^{r} c_j \lambda_s^j\right) \ket{\lambda_s} = 0 \Rightarrow \left(\sum_{j=0}^{r} c_j \lambda_s^j\right) = 0,
    \end{equation}
    for any $\ket{\lambda_s} \neq 0$. We have proven the polynomial $P(x) = \sum_{j=0}^{r} c_j x^j$ has $r+1 > r$ distinct roots $\lambda_s$ which causes contradiction (Cayley-Hamilton theorem). Therefore, $\{I, H^1, \cdots, H^{r-1}\}$ is linearly independent, and $\dim[\spn_{\RR}\cP_H] = r$.
\end{proof}

\begin{definition}
    A Lie algebra $\frak{g}$ is called \textit{perfect} if every element in $\frak{g}$ can be expressed as a commutator of two elements within $\frak{g}$. 
\end{definition}

\begin{lemma}
    Given a real Lie algebra $\frak{g}$ and a fixed Hermitian operator $H$ with the number of distinct eigenvalues $r$, denote the  $\frak{v}_H(\frak{g}):=\spn_{\RR}\liebracket{O\ox H, O\ox I, \; \forall O\in \frak{g}}$, where $I$ is the identity operator lies in the same space of $H$. Then $\dim[\frak{v}_H(\frak{g})] \leq r\dim[g]$. Moreover, the equality holds if $\frak{g}$ is perfect.
\end{lemma}
\begin{proof}
    For any $O_1, O_2\in\frak{g}$, we have, 
    \begin{equation}
        [O_1\ox H, O_2 \ox H] = O_1O_2 \ox H^2 - O_2 O_1 \ox H^2 = [O_1, O_2] \ox H^2 = O' \ox H^2,
    \end{equation}
    where $O' = [O_1,O_2]\in \frak{g}$ by definition of Lie algebra. Now consider another vector space,
    \begin{equation}
        \frak{g}\ox \spn_{\RR}\cP_H := \spn_{\RR}\{O\ox I, O\ox H^1, \cdots, O\ox H^k, \cdots, \forall O\in \frak{g}\}.
    \end{equation}
    The tensor product space here can be well-defined as $\spn_{\RR}\cP_H$ is clearly abelian. By Lemma~\ref{lem:dim_PO}, we have $\dim[\frak{g}\ox \spn \cP_H] = r \dim[\frak{g}]$ where $r$ is the number of distinct eigenvalues of $H$. Then, from the definition of $\frak{v}_H(\frak{g})$, we have $\frak{v}_H(\frak{g})\subseteq \frak{g}\ox \spn\cP_H \Rightarrow \dim[\frak{v}_H(\frak{g})] \leq r\dim[g]$. Further, if $\frak{g}$ is perfect, taking any element $O'\in\frak{g}$, there exists $O_1, O_2\in \frak{g}$ such that $[O_1, O_2]\ox H^m = [O_1\ox H^{m-1}, O_2\ox H]$. For $m=1$, the definition of $\frak{v}_H(\frak{g})$ ensuring that it contains $O\ox H$ for all $O$; For $m>1$, we can prove the statement by induction for which $O_1\ox H^{m-1} \in \frak{v}_H(\frak{g})$. Therefore, we can conclude $\frak{v}_H(\frak{g})= \frak{g}\ox \spn_{\RR}\cP_H$.
\end{proof}
\begin{lemma}\label{lem:vh_su}
    For any fixed Hermitian matrix $H$, we have $\frak{v}_H(\frak{s}\frak{u}(d)) = \frak{s}\frak{u}(d)\ox \spn_{\RR}\cP_H$
\end{lemma}
\begin{proof}
    Since $\frak{s}\frak{u}(d)$ is a simple Lie algebra and perfect, i.e., $[\frak{s}\frak{u}(d), \frak{s}\frak{u}(d)] = \frak{s}\frak{u}(d)$, taking an orthonormal basis $\{E_j\}_{j=1}^{d^2 - 1}$ of $\frak{s}\frak{u}(d)$, one can show that,
    \begin{equation}
        \frak{v}_H(\frak{s}\frak{u}(d)) = \spn_{\RR}\{E_j\ox H^k, \forall 1\leq j \leq d^2 - 1,\; k=0,1\cdots \},
    \end{equation}
    where $\{E_j \ox H^k\}_{j,k}$ forms a basis for $\frak{s}\frak{u}(d)\ox \spn_{\RR}\cP_H$. Therefore, $\dim[\frak{v}_H(\frak{s}\frak{u}(d))] = \dim[\frak{s}\frak{u}(d)\ox \spn_{\RR}\cP_H]$, and we have $\frak{v}_H(\frak{s}\frak{u}(d)) = \frak{s}\frak{u}(d)\ox \spn_{\RR}\cP_H$.
\end{proof}

From the above discussion, we have demonstrated that the attachment of an extra Hermitian operator to a perfect real matrix Lie algebra via tensor product can non-trivially create interesting structure in the result Lie algebra $\frak{v}_H$. For instance, let $H$ be an involutory matrix, i.e., $H^2 = I$. The following lemma can be derived.
\begin{lemma}
    Given $H\in i\frak{u}(\cH_{d'})$ an involutory traceless Hermitian, then $\frak{v}_{H}(\frak{su}(d))$ is semisimple, and $\frak{v}_{H}(\frak{su}(d))\simeq \frak{su}(d)\oplus\frak{su}(d)$.
\end{lemma}
\begin{proof}
    Start with $H\in i\frak{u}(\cH_{d'})$ being an involutory Hermitian unitary, s.t., $\tr(H) = 0, H^2 = I$. Then,
\begin{equation*}
    \frak{v}_H(\frak{su}(d)) = \spn_{\RR}\liebracket{E_k \ox H, \; \forall k},
\end{equation*}
where $\{E_j\}$ forms a orthonormal basis of $\frak{su}(d)$ and $I$ denote the $d'$-dimensional identity matrix. One can define a new orthonormal basis of $\cP_H$ as $\{\frac{I+H}{\sqrt{2d'}}, \frac{I-H}{\sqrt{2d'}}\}$ regarding the Hilbert-Schmidt inner product. By using this basis, we could compute the commutators as 
\begin{equation}
\begin{aligned}
    [E_i\ox (I+H), E_j \ox (I+H)] &= [E_i, E_j] \ox (I+H)\\
    [E_i\ox (I+H), E_j \ox (I-H)] &= 0\\
    [E_i\ox (I-H), E_j \ox (I-H)] &= [E_i, E_j] \ox (I-H).
\end{aligned}
\end{equation}
Since $\frak{su}(d)$ is perfect, $[E_i, E_j]$ is also an element of $\frak{su}(d)$. Further, based on the commutators above, defining,
\begin{equation}
\begin{aligned}
    \frak{g}_+ = \spn_{\RR}\liebracket{\frac{1}{\sqrt{2d'}}E_j\ox (I+H)\;: \forall j};\; \frak{g}_- = \spn_{\RR}\liebracket{\frac{1}{\sqrt{2d'}}E_j\ox (I-H) \;: \forall j}.
\end{aligned}
\end{equation}
It can be verified that both $\frak{g}_{\pm}$ form subalgebras of $\frak{v}_H(\frak{su}(d))$ which are closed under the Lie bracket, respectively. Clearly, $\frak{g}_{\pm}\simeq \frak{su}(d)$, being simple, by definition and their ability to generate the entire Lie algebra $\frak{v}_H(\frak{su}(d))$. Besides, due to the middle commutator of the above, the two subalgebras commute with each other. Therefore, $\frak{g}_{\pm}$ forming direct sum decomposition of $\frak{v}_H(\frak{su}(d))$ which then derives $\frak{v}_H(\frak{su}(d)) = \frak{g}_{+} \oplus \frak{g}_{-} \simeq \frak{su}(d)\oplus\frak{su}(d)$. 
\end{proof} 

Now comes for the general Hermitian. As for finite dimension, $H$ can be always decomposed into block sectors acting on the orthogonal subspaces of the entire Hilbert space. With this structure, one can then decompose $\frak{v}_H(\frak{su}(d))$ into direct sum of simple algebras.
\begin{lemma}\label{lem:decompose_dla_dqnn}
    Let $H$ be any finite dimensional Hermitian matrix with $r$ distinct  eigenvalues, and let $\{\Pi_j\}_{j=1}^r$ denote the set of projections onto each distinct eigenspace of $H$. Then, then $\frak{v}_H(\frak{su}(d))$ has the following form,
    \begin{equation*}
        \frak{v}_H(\frak{su}(d)) = \spn_{\RR}\liebracket{E_j \ox \Pi_k \;: \forall j,k},
    \end{equation*}
    where $\{E_j\}_j$ forms an orthonormal basis of $\frak{su}(d)$. Furthermore, $\frak{v}_H(\frak{su}(d))$ is semisimple. 
\end{lemma}
\begin{proof}
    Suppose $H$ is a Hermitian operator acting on a $d'$-dimensional Hilbert space, and let $\Pi_j$'s be the projections onto the $r$ distinct orthogonal eigenspaces of $H$, so that $\Pi_i\Pi_j = \delta_{ij} \Pi_j$, and $\sum_j \Pi_j = I_{d'}$. Notices that $\cP_H$ is simultaneously diagonalizable as the integer power of any Hermitian operator that commutes with itself. Recall the orthonormal basis $\{E_j\}_j$ for $\frak{su}(d)$, by definition, we can then express $\frak{v}_H(\frak{su}(d))$ as,
    \begin{equation*}
        \frak{v}_H(\frak{su}(d)) = \spn_{\RR}\liebracket{E_j\ox H^k\;: \forall j,k}.
    \end{equation*}
    We can also define another Lie algebra,
    $$\frak{g}_{\Pi}:=\spn_{\RR}\liebracket{E_j\ox \Pi_k,\;: \forall j,k},$$ and our goal is to  demonstrate that $\frak{v}_H(\frak{su}(d)) = \frak{g}_{\Pi}$. We start with the ($\subseteq$) direction. Notice that $H = \sum_{k} \lambda_k \Pi_k$ with $\lambda_k \neq 0$. For every $j$, we have,
    \begin{equation*}
        E_j \ox H = \sum_k \lambda_k (E_j \ox \Pi_k).
    \end{equation*}
    Because the RHS is a real linear combination of the elementary generators of $\frak{g}_{\Pi}$, it is clear that $E_j\ox H$ lies in $\frak{g}_{\Pi}$ for all $j$. Taking Lie brackets and real linear combinations (i.e. forming the Lie closure) then shows $\subseteq$ relation.

    The reverse inclusion ($\supseteq$) is a little more involved. The key point is that the $\Pi_k$'s are polynomials in $H$. More precisely, if $\{\lambda_k\}_k$ are the distinct eigenvalues, then the \textit{Lagrange interpolation theorem} tells that there exist real polynomials $p_k(x)$ satisfying, $p_l(\lambda_m) = \delta_{lm}$. One has the identity,
    \begin{equation}\label{eq:LIT}
        \Pi_k = p_k(H) = \prod_{j\neq k} \frac{H - \lambda_j I_{d'}}{\lambda_k - \lambda_j}.
    \end{equation}
    Now we can observe the following, as from Lemma~\ref{lem:vh_su}, any element of form $i\Omega \ox H^k$ for $k = 1,2,\cdots r-1$ lies in $\frak{v}_H(\frak{su}(d))$ due to the repeated commutators from $\liebracket{E_j\;:\forall j}$. Besides, from Eq.~\eqref{eq:LIT}, any spectral projection $\Pi_k$ can be written as,
    \begin{equation*}
        \Pi_k = p_k(H) = c_{k,0} I_{d'} +  \sum_{j=1}^{r-1} c_{k,j} H^j.
    \end{equation*}
    Therefore, for any $l$,
    \begin{equation*}
        E_l \ox \Pi_k = c_{k,0} (E_l \ox I_{d'}) + \sum_{j=1}^{r-1} c_{k,j} (E_l\ox H^j),
    \end{equation*}
    which lies in $\frak{v}_H(\frak{su}(d))$. \textbf{Above all}, both the inclusion relation then proves that $\frak{v}_H(\frak{su}(d)) = \frak{g}_{\Pi}$.
    
    We can further construct an orthonormal basis corresponding to the Hilbert-Schmidt inner product for $\frak{v}_H(\frak{su}(d))$ as $\{\frac{1}{\sqrt{\chi_k}} E_j \ox \Pi_k\}_{j,k}$ where $\chi_k$ denotes the dimension of each projective space regarding the index $k$. By construction, one can verify that,
    \begin{equation}
    \begin{aligned}
        [E_i \ox \Pi_k, E_j \ox \Pi_l] = \delta_{kl}[E_i,E_j]\ox \Pi_l,
    \end{aligned}
    \end{equation}
    which can not generate a new element for $k\neq l$, hence closed under the matrix Lie bracket. Therefore, we can define the subalgebras $\frak{g}_j:=\spn_{\RR}\liebracket{E_k \ox \Pi_j, \; \forall k}$ for each fixed $j$-th orthogonal sector of $H$ so that, $\frak{v}_H(\frak{su}(d)) = \bigoplus_j \frak{g}_j$, and we now complete the proof
\end{proof}

Recalling the DLA of our QRENN model. For fixed data, i.e., the generator of the data unitary is invariant, Lemma~\ref{lem:g_dqnn_hp} gives,
\begin{equation*}
    \liedqnn = \spn_{\RR}\langle \; iI_{2^m} \ox H, iP_q\ox H, iP_s\ox I_{2^n}\;:\forall q,s \; \rangle_{\Op{Lie}}.
\end{equation*}
In this case, the first term in the bracket commute with the other two terms, which then create the center of $\liedqnn$. Moreover, the other two terms can be demonstrated to generate the entire $\frak{v}_H(\frak{su}(2^m))$, by definition. As a result, we can derive the following proposition to characterize the structure of $\liedqnn$ with fixed data.
\begin{proposition}\label{prop:dla_fixed_H}
    Let $\{\Omega_j\}_j$ be the set of Hermitian generators of the data processing register so that $\liebracket{i\Omega_j}$ spans $\frak{su}(2^m)$. Suppose a fixed quantum evolution $V = e^{i\varphi_t H}$ is embedded via the map~\eqref{eq:encoding_map}. Then, the DLA of the QRENN model can be decomposed into,
    \begin{equation*}
        \liedqnn \simeq \frak{c} \oplus \frak{su}(2^m)^{\oplus r}
    \end{equation*}
    where $\frak{c}=\spn_{\RR}\liebracket{iI_{2^m}\ox H}$ is the center, and $r$ denotes the number of distinct eigenspaces of $H$.
\end{proposition}
\begin{proof}
    Given $H$ a $2^n \times 2^n$ Hermitian matrix with $r$ distinct orthogonal subspaces. One can write $H = \sum_{j=1}^r \lambda_j \Pi_j$ for $\lambda_j\in\RR$ where $\Pi_j\Pi_k = \delta_{jk} \Pi_k$ and $\sum_j \Pi_j = I_{2^n}$. Recalling the DLA of our QRENN model via Lemma~\ref{lem:g_dqnn_hp},
\begin{equation*}
    \liedqnn = \spn_{\RR}\langle \; iI_{2^m} \ox H, iP_q\ox H, iP_s\ox I_{2^n}\;:\forall q,s \; \rangle_{\Op{Lie}}.
\end{equation*}
Notice that 
$$[iI_{2^m} \ox H, iP_q \ox H^k] = - P_q \ox [H, H^k] = 0$$ 
$$[iI_{2^m} \ox H, iP_s \ox I_{2^n}] = - P_s \ox [H, I_{2^n}] = 0$$
showing that the subalgebra generated by $iI_{2^m}\ox H$ is abelian. Denote $\frak{c} = \spn_{\RR}\liebracket{i I_{2^m}\ox H}$ forming a one-dimensional center of $\liedqnn = \frak{c} \oplus [\liedqnn, \liedqnn]$. The commutator ideal is derived as,
\begin{equation*}
    [\liedqnn, \liedqnn] = \spn_{\RR}\liebracket{iP_q\ox H, iP_s\ox I_{2^n}\;:\forall q,s}.
\end{equation*}
Now since $\liebracket{iP_q, \;:\forall q}$ spans the whole $\frak{su}(2^m)$ by assumption. We can follow a similar logic as before and define 
\begin{equation}\label{eq:basis_fixed_H_dla}
    \cB:=\left\{\frac{i}{\sqrt{\chi_s}}P_q\ox \Pi_s\right\}_{q,s},
\end{equation}
as an orthonormal basis for the ideal where $\chi_s = \Op{rank}(\Pi_s)$. We can then apply Lemma~\ref{lem:decompose_dla_dqnn} and define $\frak{g}_j:=\spn_{\RR}\liebracket{\frac{i}{\sqrt{\chi_j}}P_q \ox \Pi_j \;: \forall q}$ for each fixed $j$-th orthogonal sector of $H$ so that these subalgebras are isomorphic to $\frak{su}(2^m)$, and commutative with each other. Therefore, 
\begin{equation*}
    [\liedqnn, \liedqnn] = \frak{v}_H(\frak{su}(2^m) = \bigoplus_{j=1}^{r} \frak{su}(2^m) = \frak{su}(2^m)^{\oplus r},
\end{equation*}
as required.
\end{proof}

\subsection{QRENN with simultaneous diagonalizable data framework}
A natural extension of the above theory is to consider a set of Hermitian generators from the data embedding register. Suppose now we are embedding different data unitary $e^{i\varphi_t H_t(\bm{x})}$ through the controlling embedding, where $\cS =\{H_t\}_t$ forms a set of data generators acting on $i\frak{u}(\cH_{2^n})$, and $[H_\alpha, H_\beta] = 0$. Recalling the DLA as
\begin{equation*}
    \liedqnn=\spn_{\RR}\langle \; iI_{2^m} \ox H_p, iP_q\ox H_r, iP_s\ox I_{2^n}\;:\forall p,q,r,s \; \rangle_{\Op{Lie}}.
\end{equation*}
Notice that
\begin{equation*}
\begin{aligned}
    [iP_s \ox I_{2^n}, iP_q \ox H_r] &= -[P_s, P_q] \ox H_r\\
    [iP_s \ox I_{2^n}, iP_q \ox I_{2^n}] &= -[P_s, P_q] \ox I_{2^n}
\end{aligned}
\end{equation*}
can not generate new elements through the Lie bracket since $[P_s, P_q]$ is also an element in $\frak{su}(2^m)$. For example, by taking $\{P_q\}_q$ as a set of $m$-qubit Cartan-Weyl basis of $\frak{su}(2^m)$, one can always find $P_r$ such that $[P_s, P_q] = iP_r$. Therefore, the only terms that can generate new elements via the Lie bracket are,
\begin{equation*}
    [iP_p \ox H_q, iP_r\ox H_s] = -[P_p, P_r] \ox H_sH_q.
\end{equation*}
This is similar to the fixed data scenario. The DLA, in this case, will be spanned by generators given in the above expression and the extra terms of $iP_q \ox H_{s_1}H_{s_2}\cdots$. Notice that the terms with cumulative matrix product of $H_p$ will end up as they can finally become linearly dependent for finite dimension. In fact, since all $H_p$'s commute. They can be simultaneous diagonalizable. The product, therefore, has its image spaces joined together.

Consequently, the underlying Hilbert space (typically $\cH = \CC^{d}$ for finite-dimensional matrices) decomposes into an orthogonal direct sum of joint eigenspaces associated with tuples of eigenvalues from the operators. Formally, each $H_t$ has its own spectrum and the joint eigenspace for a tuple of eigenvalues $\bm{\lambda} = (\lambda_t)$, where $\lambda_t \in \RR$ for Hermitian operators. We define,
\begin{equation}\label{eq:joint_eigenspaces}
    V_{\bm{\lambda}} := \bigcap_t \Op{Ker}[H_t - \lambda_t I_{2^n}].
\end{equation}
Then the space $\cH$ that $H_t$'s acting on can be decomposed as,
\begin{equation*}
    \cH_{\cS} \simeq \bigoplus_{\bm{\lambda}} V_{\bm{\lambda}},
\end{equation*}
where the direct sum runs over all distinct eigenvalue tuples. Notice that the space is not merely a collection of individual eigenspaces of each $H_t$. The original degeneracy appears in each $H_t$ broken and replaced by the joint eigenspaces uniquely labelled by $\bm{\lambda}$. In quantum mechanics, this corresponds to \textit{compatible observables} sharing a common eigenbasis. Degenerate blocks imply that additional observables are needed to resolve the degeneracy.
\begin{proposition}\label{prop:dla_commute_H}
    Let $\{\Omega_j\}_j$ be the set of Hermitian generators of the data processing register so that $\liebracket{i\Omega_j}$ spans $\frak{su}(2^m)$. Suppose a set of commutative Hermitian data matrices $\{H_t\}_t$ is embedded into the QRENN model via the controlling map~\eqref{eq:encoding_map}. Then the DLA of the QRENN model can be decomposed into,
    \begin{equation*}
        \liedqnn \simeq \frak{c} \oplus \frak{su}(2^m)^{\oplus r}
    \end{equation*}
    where $\frak{c}:=\spn_{\RR}\liebracket{iI_{2^m}\ox\Pi_{\bm{\lambda}}\;:\forall \bm{\lambda}}$ is the center, $r$ is the number of distinct joint eigenspaces from all of $H_t$, $\Pi_{\bm{\lambda}}$ is the projection onto the corresponding space such that $\sum_{\bm{\lambda}} \Pi_{\bm{\lambda}} = I_{2^n}$.
\end{proposition}
\begin{proof}
    For a set of commutative Hermitian matrices $\{H_t\}_t$, the spectral theorem guarantees that they can be simultaneously diagonalized, which is a special case of block diagonalization where all blocks are $1 \times 1$ (i.e., scalar entries). 
    
    For each tuple of eigenvalues $\bm{\lambda} = (\lambda_t)$ (one eigenvalue $\lambda_t$ from each $H_t$), choose an orthonormal basis for each $V_{\bm{\lambda}}$ as in Eq.~\ref{eq:joint_eigenspaces} and use them to construct the orthogonal projections $\{\Pi_{\bm{\lambda}}\}$. Every $H_t$ will be block diagonal, with each block corresponding to a joint eigenspace $V_{\bm{\lambda}}$ and acts as $\lambda_t I_{\bm{\lambda}}$ where $I_{\bm{\lambda}}$ is the identity operator acting on $V_{\bm{\lambda}}$. Therefore, one can follow the similar proof logic in Proposition~\ref{prop:dla_fixed_H} for the fixed data scenario, since $H_t$'s are pairwise commuting, using the language of joint eigenspace, defining the Lie algebra,
    \begin{equation*}
        \frak{g}_{\Pi}:=\spn_{\RR}\liebracket{iP_l \ox \Pi_{\bm{\lambda}} \;:\forall l, \bm{\lambda}},
    \end{equation*}
    one has,
    \begin{equation*}
        H_t = \sum_{\bm{\lambda}} \lambda_t \Pi_{\bm{\lambda}}.
    \end{equation*}
    Then for each $t$,
    \begin{equation*}
        iP_l \ox H_t = iP_l \ox \sum_{\bm{\lambda}} \lambda_t \Pi_{\bm{\lambda}} = \sum_{\bm{\lambda}} \lambda_t (iP_l \ox \Pi_{\bm{\lambda}}),
    \end{equation*}
    which proves that all generators with expression $iP_l\ox H_t$ lie in $\frak{g}_{\Pi}$, and hence $\liedqnn\subseteq \frak{g}_{\Pi}$. Now for the reversing inclusion, for several commuting $H_t$, one uses a similar multivariate Lagrange interpolation argument, i.e., there exist real polynomials $p_{\bm{\lambda}}(H_1, \cdots, H_T)$ s.t., 
    \begin{equation*}
        p_{\bm{\lambda}}(H_1, \cdots, H_T) = \Pi_{\bm{\lambda}}.
    \end{equation*}
    As from the minimal polynomial for a fixed data Hermitian, we can then express $\Pi_{\bm{\lambda}}$ as,
    \begin{equation*}
        \Pi_{\bm{\lambda}} = \sum_{k_1, \cdots, k_T} c^{(\bm{\lambda})}_{k_1, \cdots, k_T} H_1^{k_1} H_2^{k_2} \cdots H_T^{k_T}.
    \end{equation*}
    We then immediately have,
    \begin{equation*}
        iP_l\ox \Pi_{\bm{\lambda}} = \sum_{k_1, \cdots, k_T} c^{(\bm{\lambda})}_{k_1, \cdots, k_T} (iP_l \ox H_1^{k_1} H_2^{k_2} \cdots H_T^{k_T}),
    \end{equation*}
    where the constant terms (with $k_1=k_2=\cdots=k_T = 0$) indicates the zero-degree terms $iP_l \ox I_{2^n}$. Notice that every monomial $H_1^{k_1}H_2^{k_2}\cdots H_T^{k_t}$ is an element of the generator set of $\liedqnn$ via Lie bracket expansion. As a result, for any $l,\bm{\lambda}$, the term $iP_l\ox \Pi_{\bm{\lambda}} \in \liedqnn$, which then implies $\liedqnn\supseteq\frak{g}_{\Pi}$, and therefore, $\liedqnn=\frak{g}_{\Pi}$.
    
    Now since $\{\Pi_{\bm{\lambda}}\}$ indicate orthogonal projections, i.e., $\Pi_{\bm{\lambda}}\Pi_{\bm{\lambda}'}=\delta_{\bm{\lambda}\bm{\lambda}'}\Pi_{\bm{\lambda}}$. One can then analyses and decomposes $\liedqnn$ in the following. Firstly,
    \begin{equation*}
        \liedqnn = \frak{c}\oplus\spn_{\RR}\liebracket{iP_q\ox \Pi_{\bm{\lambda}}, iP_s\ox I_{2^n} \;: \forall q, s, \bm{\lambda}},
    \end{equation*}
    where $\frak{c}:=\spn_{\RR}\liebracket{iI_{2^m}\ox\Pi_{\bm{\lambda}}\;:\forall \bm{\lambda}}$ forms a center as the generators within $\frak{c}$ commute with all element in $\liedqnn$. Then, we construct the orthonormal basis $\cB = \{\frac{i}{\sqrt{\chi_{\bm{\lambda}}}} P_q\ox \Pi_{\bm{\lambda}}\}_{q, \bm{\lambda}}$. We define $\frak{g}_{\bm{\lambda}}:=\spn_{\RR}\liebracket{\cB}$ so that,
    \begin{equation*}
        \liedqnn = \frak{c}\oplus \bigoplus_{\bm{\lambda}} \frak{g}_{\bm{\lambda}} \simeq  \frak{c}\oplus \bigoplus_{\bm{\lambda}} (\frak{su}(2^m) \ox \RR\cdot \Pi_{\bm{\lambda
        }}) \simeq \frak{c} \oplus \frak{su}(2^m)^{\oplus r},
    \end{equation*}
    and we complete the proof.
\end{proof}

\section{Absence of Barren Plateaus in the QRENN model}\label{appendix:trainability_of_dqnn_On_In}

One of the most essential parts of devising a QML model is proving its trainability. Previous research has shown that sufficient randomness of the template parametrised quantum circuit (QNN) can induce barren plateaus regarding a random initialization strategy for the tunable variables in the model. In this appendix, we mainly discuss the trainability of our QRENN model and theoretically demonstrate that with certain assumptions, the control operation can restrict the expansion of the DLA and, hence, prevent the gradient from being trapped in the plateaus.

In the context of VQAs, one should start with the analysis of the gradient of loss function defined in Sec. ~\ref{subsec:qrenn} with respect to the parameter space. In the abstract setting, one can always assume the sufficient deepness of the QNN, such that the dynamical group represented via the circuit can be fully mixed up to the second moment~\cite{Fontana2024characterizing,Ragone2024lie,Allcock2024dynamical}. The original QNN gradient is then replaced by the \textit{abstracted gradient}.
\begin{definition}[Abstracted Gradient~\cite{Fontana2024characterizing}]
    Let $G$ be a compact, connected Lie group represented as unitary matrices in the unitary group $\cU(V)$ over the vector space $V$. In addition, let $H\in \frak{g}$, and $iO, iA\in \frak{u}(V)$, the Lie algebra of $\cU(V)$. We define the abstracted gradient as:
    \begin{equation*}
        \partial_H \ell(A,O) := \tr\left(U^{\dagger}_{g^-} A U_{g^-} [H, U_{g^+} O U^{\dagger}_{g^+}]\right),
    \end{equation*}
    where $U_{g^{\pm}}$ represents arbitrary $g^{\pm} \in G$. 
\end{definition}

In our case, the $\liedqnn$ forms a subalgebra of $\frak{u}(2^{m+n})$, which is known as a compact Lie algebra of the unitary group $\cU(2^{m+n})$. The dynamical group $e^{\liedqnn}$ is concretely a (connected) Lie subgroup of $\cU(2^{m+n})$. From the theory of Lie groups. Every connected Lie subgroup of $\cU$ is a closed subgroup. With the Hausdorff topology, $e^{\liedqnn}$ is, therefore,  automatically compact (with the subspace topology), and so does $\liedqnn$.
\begin{lemma}[Vanishing mean~\cite{Fontana2024characterizing}]\label{lem:vanishing_mean}
    Let $G$ be a compact Lie group, and $\phi:G\rightarrow \cU(V)$ a representation. Then for any $O, A \in \frak{gl}(V)$:
    \begin{equation}
        \EE_{g^+,g^-\sim\mu^{\ox 2}}[\partial_H\ell(A,O)] = 0
    \end{equation}
\end{lemma}
Since $e^{\liedqnn}$ is compact, with the natural representation of $e^{\liedqnn}$ into the unitary matrices. Suppose the abstract setting is matched for the QRENN model. The above Lemma implies that,
\begin{equation*}
    \EE_{\bm{\theta}, \bm{\varphi}}[\partial_{t,\mu}\ell(\rho, O;\bm{\theta},\bm{\varphi})] = \EE_{g^+,g^-\sim\mu^{\ox 2}}[\partial_H\ell(\rho,O)] = 0
\end{equation*}

\begin{theorem}[Compact group variance~\cite{Fontana2024characterizing}]\label{them:compact_group_var}
    Let $G$ be a compact and connected Lie group with Lie algebra $\frak{g} = \frak{c} \oplus \bigoplus_{j} \frak{g}_j$. Suppose $\phi$ is a finite-dimensional unitary representation, $iO,iH\in\frak{g}$, and $\rho$ is a density matrix. Then the following holds:
    \begin{equation}
        \Var[\partial_H\ell(\rho,O)] = \sum_{j} \frac{\|H_{\frak{g}_j}\|_K^2\norm{O_{\frak{g}_j}}_F^2\|\rho_{\frak{g}_j}\|_F^2}{d^2_{\frak{g}_j}},
    \end{equation}
    where the center $\frak{c}$ does not contribute to the variance.
\end{theorem}
Specifically, $\norm{O_{\frak{g}_j}}_F$ and $\|\rho_{\frak{g}_j}\|_F$ denote the Frobenius norms of the orthogonal projections of the operators $O$ and $\rho$, respectively, onto the simple ideal $\frak{g}_j$. Similarly, $\|H_{\frak{g}_j}\|_K$ is the corresponding Killing norm of $H$.

To demonstrate our main results, we first compute the projections of the initial state and the measurement operator onto each decomposition ideal $\frak{g}_j$, Taking an input state  $\rho = \rho_m\ox \rho_n$ as an example, where $\rho_m$ and $\rho_n$ acting on $\cH_{2^m}$ and $\cH_{2^n}$, respectively. The (normalized) observable $O = O_m\ox O_n$ for some $O_m$ and $O_n$ acting on $\cH_{2^m}$ and $\cH_{2^n}$ locally with bounded Frobenius norm. 

\subsection{Operator projections onto the DLA}\label{subsec:proj_onto_dla_appendix}
In the first place, for either fixed generator or commutative generators. The set of embedded data Hermitian operators from a multi-slot QRENN can span a commutative algebra that breaks into a direct sum of subalgebra. Therefore, the total DLA can be decomposed into simple Lie algebras where each component is in the form $\frak{g}=\spn_{\RR}\liebracket{iP_q\ox \Pi \;:\forall q}$ for $\liebracket{iP_q \;: \forall q}$ spans $\frak{su}(2^m)$ and $\Pi$ is a projection. In the following, we aim to derive the projections of $\rho$ and $O$ onto the Lie algebra in the form. Suppose $O_m\in i\frak{su}(2^m)$, then the projection (w.r.t. Hilbert-Schmidt inner product) of it onto $\frak{g}$ is,
\begin{equation}
\begin{aligned}
    O_{\frak{g}} &= \frac{1}{\chi}\sum_{q} \tr(((iP_q)^{\dagger} \ox \Pi)(O_m\ox O_n)) iP_q \ox \Pi\\
    &= \frac{1}{\chi}\left(\sum_{q} \tr(P_q O_m) P_q\right) \ox \tr(\Pi O_n)\Pi\\
    &= \frac{\tr(\Pi O_n)}{\chi} O_m \ox \Pi
\end{aligned}
\end{equation}
Similarly, for an input state $\rho_m\ox\rho_n$, we can compute the projection in the same way,
\begin{equation}
\begin{aligned}
   \rho_{\frak{g}} &= \frac{1}{\chi}\sum_{q} \tr(((iP_q)^{\dagger} \ox \Pi)(\rho_m\ox \rho_n)) iP_q \ox \Pi\\
   &= \frac{1}{\chi}\left(\sum_{q} \tr(P_q \rho_m) P_q\right) \ox \tr(\Pi \rho_n)\Pi\\
   &= \frac{\tr(\Pi \rho_n)}{\chi} (\rho_m)_{i\frak{su}(2^m)}\ox \Pi.
\end{aligned}
\end{equation}
Here $\chi$ denotes the rank of $\Pi$. We can further compute the norms of these projections of $O$ and $\rho$,
\begin{equation}\label{eq:fnorm_o_rho}
\begin{aligned}
    \norm{O_{\frak{g}}}_F &= \frac{\tr(\Pi O_n)}{\chi}\norm{O_m \ox \Pi}_F = \frac{\tr(\Pi O_n)}{\chi}\norm{O_m}_F \norm{\Pi}_F = \sqrt{\chi}^{-1}\tr(\Pi O_n)\norm{O_m}_F,\\
    \norm{\rho_{\frak{g}}}_F &= \frac{\tr(\Pi \rho_n)}{\chi} \norm{(\rho_m)_{i\frak{su}(2^m)}\ox \Pi}_F = \sqrt{\chi}^{-1}\tr(\Pi \rho_n) \norm{(\rho_m)_{i\frak{su}(2^m)}}_F,
\end{aligned}
\end{equation}
where we have used the multiplicity of the Frobenius norm under Kronecker product and the fact that $\norm{\Pi}_F = \sqrt{\tr(\Pi)} = \sqrt{\chi}$. Next, we target to derive the Killing norm of the differentiated generator $H$ from the model. Given $\phi:e^{\frak{g}}\rightarrow \cU(V)$ a Lie group representation. We have the corresponding Lie algebra representation $d\phi:\frak{g} \rightarrow \frak{u}(V)$. The natural way to relate Killing norm and the Frobenius norm is to use the associated index of the representation~\cite{Fontana2024characterizing} and have,
\begin{equation}
    \norm{d\phi(h)}_K^2 = \norm{h}_K^2 = \frac{I_{\Op{Ad}}}{I_{\phi}}\norm{d\phi(h)}_F^2
\end{equation}
for any $h\in \frak{g}$, $I_\phi$ is called the index of the representation having the value twice of the Dynkin index~\cite{Fuchs1995affine}. Notice that for compact Lie algebras, $I_{\Op{Ad}}$ scales as $\Theta(\sqrt{d_{\frak{g}}})$ for all non-exceptional classes.

In the formal definition from Lie algebra representation, for any $X,Y\in \frak{g}$, a Lie algebra, the adjoint action is defined as $\adj{X}(Y) = [X,Y]$ and the Killing form $\kappa(X,Y)$ is $\kappa(X,Y) = \tr(\adj{X} \circ \adj{Y})$. While for matrix Lie algebras, the Killing form is proportional to the trace form. In particular, let $\{E_k\}_{k=1}^{d_{\frak{g}}}$ be an orthonormal basis for $\frak{g}$, then the Killing norm of $H\in\frak{g}$ induced by the form regarding $\frak{g}$ can be computed as,
\begin{equation*}
    \norm{H}_K^2 = \sum_{k,l=1}^{d_{\frak{g}}} \tr([H, E_k]E_l)^2 = \sum_{k=1}^{d_{\frak{g}}}\norm{[H, E_k]}_F^2.
\end{equation*}
Now, we are ready to compute the Killing norm term in the gradient variance. 
\begin{lemma}\label{lem:killing_norm_iomega_ox_H}
    Let $i\Omega\in\frak{u}(2^m)$ and $iH\in\frak{u}(2^n)$, the Killing norm of the projection of $i\Omega\ox H$ onto $\frak{g}$ is, 
    \begin{equation*}
        \norm{(i\Omega \ox H)_{\frak{g}}}_K^2 = \frac{\tr(\Pi H)^2}{\chi^2}\norm{i\Omega}_K^2
    \end{equation*}
    where $\norm{i\Omega}_K$ is the Killing norm of $i\Omega$ within $\frak{su}(2^m)$.
\end{lemma}
\begin{proof}
    Let $i\Omega \in \frak{u}(2^m)$ and $H\in i\frak{u}(2^n)$. The projection of $i\Omega \otimes H$ onto $\frak{g}$ (defined in Eq.~\eqref{eq:basis_fixed_H_dla}) can be derived as,
    \begin{equation*}
    \begin{aligned}
        (i\Omega \otimes H)_{\frak{g}} &= \frac{1}{\chi}\sum_{q} \tr((-iP_q \ox \Pi)(i\Omega \ox H))(iP_q \ox \Pi) = \frac{\tr(\Pi H)}{\chi} (i\Omega)_{\frak{su}(2^m)} \ox \Pi.
    \end{aligned}
    \end{equation*}
    The Killing norm of the projection regarding the Lie algebra $\frak{g}$ can be computed,
    \begin{equation*}
    \begin{aligned}
        \norm{(i\Omega \otimes H)_{\frak{g}}}_K^2 &= \frac{\tr(\Pi H)^2}{\chi_j^3}\sum_{q=1}^{d_{\frak{g}}} \norm{[(i\Omega)_{\frak{su}(2^m)} \ox \Pi, iP_q\ox \Pi]}_F^2\\ &= \frac{\tr(\Pi H)^2}{\chi^2}\sum_{q=1}^{d_{\frak{g}}} \norm{[(i\Omega)_{\frak{su}(2^m)}, iP_q]}_F^2 = \frac{\tr(\Pi H)^2}{\chi^2}\norm{(i\Omega)_{\frak{su}(2^m)}}_K^2
    \end{aligned}
    \end{equation*}
\end{proof}

\subsection{Trainability of QRENN concerning the trace form loss}
With all calculations and discussions in the previous sections, we can first demonstrate that the trainability of QRENN with respect to a fixed observable $O$ by applying Theorem~\ref{them:compact_group_var} to the single term trace form loss function $\ell(\rho, O)$. In the following two sections, we discuss the scenario in which all the slots of QRENN are embedded with a fixed generator. The generalization towards commutative generators follows the same language, for which we will omit the repeating proofs.

\subsubsection{Training the data processing register}
Consider taking the derivative with respect to the $\mu$-th parameter at the $t$-th slot of the data processing register, which brings down the generator $i\Omega_{\mu} \ox I_{2^n}$ from the exponent of the model. Assuming the circuit is sufficiently deep to approximate the unitary $2$-design of $e^{\liedqnn}$ so that one can replace the parameter gradient with the abstract gradient. {In the case of a fixed data generator, the joint eigenspace overlap reduces to the eigenspace overlap, i.e., given $H\in i\frak{u}(d)$, the eigenspace overlap of $\rho\in\cD(d)$ regarding $H$ is given by $R^2_H(\rho_n) := \sum_{k = 1}^{r} \tr(\Pi_k \rho)^2$ where $r$ is the number of distinct eigenspaces of $H$.}

\begin{lemma}\label{lem:fnorm_0_state_Om_ox_In}
    Given the initial state $\rho = \ketbra{0_m}{0_m} \ox \rho_n$ and the $O = O_m \otimes I_{2^n}$. The Frobenius norms of Eqs.~\eqref{eq:fnorm_o_rho} is derived as,
    \begin{equation*}
    \begin{aligned}
        \norm{O_{\frak{g}}}_F =\sqrt{\chi}\norm{O_m}_F, \quad
        \norm{\rho_{\frak{g}}}_F = \sqrt{(1 - 2^{-m}) / \chi}\tr(\Pi \rho_n).
    \end{aligned}
    \end{equation*}
\end{lemma}
\begin{proof}
    Start with $O = O_m \otimes I_{2^n}$. From Eqs.~\eqref{eq:fnorm_o_rho}, we have,
    \begin{equation*}
        \norm{O_{\frak{g}}}_F = \sqrt{\chi}^{-1}\tr(\Pi O_n) \norm{O_m}_F = \sqrt{\chi} \norm{O_m}_F.
    \end{equation*}
    Then, let $\rho_m = \ketbra{0_m}{0_m}$, we first compute its projection onto $i\frak{su}(2^m)$. Notice that, $\ketbra{0}{0} = \frac{I+Z}{2}$ where $I$ and $Z$ are the single-qubit identity and Pauli-Z matrices. Therefore,
    \begin{equation*}
        2^m\ketbra{0_m}{0_m} = \left(I+Z\right)^{\ox m} = I_{2^m} + \sum_k Z_k + \sum_{k\neq l}Z_kZ_l + \cdots, 
    \end{equation*}
    where $Z_k$ denotes the operator that a Pauli-Z acts solely on the $k$-th qubit. We can see the only term that does not lie in $i\frak{su}(2^m)$ is the identity $I_{2^m}$ having a zero overlap with $\frak{g}$. Hence,
    \begin{equation*}
    \begin{aligned}
        \norm{(\rho_m)_{i\frak{su}(2^m)}}^2_F = 2^{-2m}\norm{2^m\ketbra{0_m}{0_m} - I_{2^m}}^2_F = \tr(\ketbra{0_m}{0_m} - \frac{2}{2^m}\ketbra{0_m}{0_m} + I_{2^m} / 2^{2m}) = 1 - 2^{-m}.
    \end{aligned}
    \end{equation*}
    The total value of $\norm{\rho_{\frak{g}}}$ can be derived $\norm{\rho_{\frak{g}}}_F = \sqrt{(1 - 2^{-m})/\chi}\tr(\Pi \rho_n)$.
\end{proof}

\begin{proposition}\label{prop:variance_data_processing_layer}
    For a sufficiently deep ($m+n$)-qubit QRENN model with a fixed Hermitian data generator $H$ of $r$ distinct eigenspaces, embedded via Eq.~\eqref{eq:encoding_map}, let $m$ scales $\cO(\log(n))$, input state $\rho = \rho_m\ox\rho_n$, and observable $O = O_m\ox O_n$. Then $$\EE_{\bm{\theta},\bm{\varphi}}[\partial_{t,\mu}\ell] = 0,$$
    where the derivative is taken w.r.t. the ($t,\mu$)-th parameter  within the data processing register. Moreover, let $\rho_m$ be any pure $m$-qubit state, $O_n = I_{2^n}$, and $\norm{\Omega_{\mu}}^2_F, \norm{O_m}^2_F$ scales constantly. If $R_H(\rho_n)$ regarding $H$ can scale $\Omega(1/\Op{poly}(n))$. we have,
    $$\Var_{\bm{\theta}, \bm{\varphi}}[\partial_{t,\mu} \ell] \geq \Omega(1/\Op{poly}(n)).$$
\end{proposition}    
\begin{proof}
    Let the tangent generator from the differentiation be $i\Omega_{\mu}\ox I_{2^n}$, where $i\Omega_{\mu} \in \frak{su}(2^m)$. From the sufficient-deep assumption, one can replace the parameter derivative with the abstract derivative where the differentiation operator is $i\Omega_{\mu} \ox I_{2^n}$. Lemma~\ref{lem:killing_norm_iomega_ox_H} gives the explicit value of the Killing norm, i.e., 
    \begin{equation*}
        \norm{(i\Omega_{\mu}\ox I_{2^n})_{\frak{g}_j}}_K^2 =  \frac{\tr(\Pi_j)^2}{\chi_j^2}\norm{i\Omega_{\mu}}_K^2 = 2^{m+1}\norm{i\Omega_{\mu}}_F^2 = 2^{m+1}\norm{\Omega_{\mu}}_F^2,
    \end{equation*}
    where $I_{\Op{Ad}} = 2d$ for $\frak{su}(d)$ and $I_{\phi} = 1$. Combine the fact with Eqs.~\ref{eq:fnorm_o_rho}, and apply Theorem~\ref{them:compact_group_var} to derive,
    \begin{equation}\label{eq:exact_var_trace_form_loss_grad}
    \begin{aligned}
        \Var[\partial_H\ell(\rho,O)] &=  \sum_{j=1}^{r} \frac{2^{m+1}\norm{\Omega_{\mu}}_F^2\norm{O_{\frak{g}_j}}_F^2\|\rho_{\frak{g}_j}\|_F^2}{d^2_{\frak{g}_j}}\\
        &= \frac{2^{m+1} \norm{\Omega_{\mu}}_F^2\norm{O_m}_F^2}{(2^{2m} - 1)^2} \sum_{j=1}^{r} \tr(\Pi_j O_n)^2\tr(\Pi_j\rho_n)^2 \norm{(\rho_m)_{i\frak{su}(2^m)}}_F^2
    \end{aligned}
    \end{equation}
    where $d_{\frak{g}_j} = 2^{2m} - 1$ from the isomorphism. Notice that for any pure state $\rho_m = \ketbra{\psi}{\psi} = V_{\psi}\ketbra{0_m}{0_m}V^{\dagger}_{\psi}$, the state generation unitary $V_{\psi}$ can be absorbed in the first data processing register due to the Haar invariance, and we can set $\rho_m$ to be the zero state w.l.o.g. Now we apply Lemma~\ref{lem:fnorm_0_state_Om_ox_In} to derive,
    \begin{equation*}
    \begin{aligned}
        \Var[\partial_H\ell(\rho,O)] &=  \frac{2^{m+1}(1 - 2^{-m}) \norm{\Omega_{\mu}}_F^2\norm{O_m}_F^2}{(2^{2m} - 1)^2} \sum_{j=1}^{r} \chi_j^2\tr(\Pi_j\rho_n)^2 \geq 
        \frac{2^{m+1} \norm{\Omega_{\mu}}_F^2\norm{O_m}_F^2}{(2^{2m} - 1)^2} \sum_{j=1}^{r} \tr(\Pi_j\rho_n)^2.
    \end{aligned}
    \end{equation*} 
    Now we take the assumption that $m$ scales as $\cO(\log(n))$ and set the generator $\Omega_{\mu}$ and $O_m$ to be some Hermitian traceless operator such that  $\norm{\Omega_{\mu}}^2_F, \norm{O_m}^2_F \sim \Omega(1)$. The polynomially lower-bounded assumption from $R^2_H(\rho_n)$ then gives,
    \begin{equation*}
        \Var[\partial_H\ell(\rho,O)] \geq   \Omega\left(\frac{2nR^2_{H}(\rho_n)}{(n^2-1)^2}\right) \geq \Omega(1/\Op{poly}(n)),
    \end{equation*}
    as required.
\end{proof}

The above proposition has demonstrated the trainability of all the parameters $\bm{\theta}$ within the data processing register regarding a fixed data generator. One can find that the initial overlap of the state $\rho_n$ with the data generator $H$  dominates the occurrence of BP. In particular, we have the following corollary.
\begin{corollary}\label{coro:bp_exponential_overlap_rho_n}
    Let $\{i\Omega_k\}_{k}$, $H$ and the other setups be the same as defined in Proposition~\ref{prop:variance_data_processing_layer}. Suppose $\norm{\Omega_{\mu}}_F^2$ and $\norm{O_m}_F^2$ scales as $\cO(2^{m})$, for instance the $m$-qubit Pauli operators. If $R_H(\rho_n)$ scales as $\cO(2^{-n})$ and $\chi_{\max}$ scales as $\cO(\Op{poly}(n))$ where $\chi_{\max}:=\max_j \chi_j$. Then, $\Var_{\bm{\theta}, \bm{\varphi}}[\partial_{t,\mu} \ell] \leq \cO(2^{-2n})$.
\end{corollary}
\begin{proof}
    Recalling the steps from the proof of Proposition~\ref{prop:variance_data_processing_layer}, and we can derive the explicit form of the variance as,
    \begin{equation*}
    \begin{aligned}
        \Var[\partial_H\ell(\rho,O)] &=  \frac{2n(1 - n^{-1}) \norm{\Omega_{\mu}}_F^2\norm{O_m}_F^2}{(n^2 - 1)^2} \sum_{j=1}^{r} \chi_j^2\tr(\Pi_j\rho_n)^2 \\
        &\leq \frac{2n(1 - n^{-1}) \norm{\Omega_{\mu}}_F^2\norm{O_m}_F^2\chi^2_{\max}}{(n^2 - 1)^2} R_H^2(\rho_n) \\
        &\leq \cO(\Op{poly}(n) / 2^{2n})\leq \cO(1/2^{2n}).
    \end{aligned}
    \end{equation*}
\end{proof}

Corollary~\ref{coro:bp_exponential_overlap_rho_n} showcases the fact that if the initial state of the data embedding register has an exponentially small overlap onto the eigenspace of $H$, in other words, the interaction between the evolution $e^{iHt}$ and the $\rho_n$ is negligible, if the maximum eigenspace of $H$ has a dimension scales as $\cO(\Op{poly}(n))$, the parameters in the data processing register will experience BP. For example, taking $\rho_n = I_{2^n}/2^n$ for which our QRENN model can reduce to the famous \textit{one-clean-qubit} model (or DQC1)~\cite{Morimae2017power,Kim2024expressivity,Xuereb2023deterministic} which has been theoretically proven as classically hard-to-simulate~\cite{Fujii2018impossibility,Morimae2014hardness}. In this case,
\begin{equation*}
    R_H^2(I_{2^n}/2^n) = \frac{1}{2^{2n}}\sum_{j} \tr(\Pi_j)^2 \leq \frac{r\chi_{\max}^2}{2^{2n}} \leq \frac{\chi^2_{\max}}{2^n}.
\end{equation*}
If $\chi_{\max}$ scales as $\cO(\Op{poly}(n))$, the above corollary demonstrates the trainability issue of the model, making it non-scalable in practice. Despite that, the maximum dimension among all eigenspaces of $H$ can non-trivially affect the trainability as well. For instance, under the same assumption in the DQC1 scenario, if $H = I_{2^n}$, then $R_H^2(\rho_n) = \tr(I_{2^n} \rho_n) = 1$. The bound becomes trivial. In this extreme case, the data processing register and the data embedding register are fully decoupled, which our theory reduces to the original result in~\cite{Fontana2024characterizing,Ragone2024lie}.

\subsubsection{Training the data embedding register}
In this section, we will discuss the trainability of our QRENN model (fixed data) for the parameters $\bm{\varphi}$ within the data embedding register, even though the parametrised QRENN is not explicitly used in the learning tasks. Compared to the previous section, we now take the derivative with respect to the $\mu$-th parameter at the $t$-slot of the data embedding register, which brings down the generator $i\ketbra{1_m}{1_m}\ox H$. The same assumption of the circuit depth is applied, and we aim to prove that there is no BP in $\bm{\varphi}$ as well.

\begin{proposition}\label{prop:variance_data_reuploading_layer}
    For a sufficiently deep ($m+n$)-qubit QRENN model with a fixed Hermitian data generator $H$ of $r$ distinct eigenspaces, embedded via Eq.~\eqref{eq:encoding_map}, let $m$ scales as $\cO(\log(n))$, input state $\rho = \rho_m\ox\rho_n$, and observable $O = O_m\ox O_n$. Then $$\EE_{\bm{\theta},\bm{\varphi}}[\partial_{t,\mu}\ell] = 0,$$
    where the derivative is taken w.r.t. the ($t,\mu$)-th parameter within the data embedding register. Moreover, let $\rho_m$ be any pure $m$-qubit state, $O_n = I_{2^n}$, and $\norm{\Omega_{\mu}}^2_F, \norm{O_m}^2_F$ scales constantly. For $H$ of $r$ distinct eigenvalues, if $\tr(H\rho_n)^2$ scales as $\Omega(r/\Op{poly}(n))$. we have,
    $$\Var_{\bm{\theta}, \bm{\varphi}}[\partial_{t,\mu} \ell] \geq \Omega(1/\Op{poly}(n)).$$
\end{proposition}    
\begin{proof}
Let the tangent generator from the differentiation be $i\ketbra{1_m}{1_m}\ox H$, where $H \in i\frak{u}(2^n)$ is the data generator. Again, we can replace the parameter derivative with the abstract derivative. Then we apply Lemma~\ref{lem:killing_norm_iomega_ox_H} to derive,
\begin{equation*}
    \norm{(i\ketbra{1_m}{1_m} \ox H)_{\frak{g}_j}}_K^2 = \frac{\tr(\Pi_j H)^2}{\chi_j^2}\norm{(i\ketbra{1_m}{1_m})_{\frak{su}(2^m)}}_K^2 = 2(2^m-1)\lambda_j^2
\end{equation*}
where the projection of $i\ketbra{1_m}{1_m}$ onto $\frak{su}(2^m)$ can be derived from the steps of Lemma~\ref{lem:fnorm_0_state_Om_ox_In} as,
\begin{equation*}
    (i\ketbra{1_m}{1_m})_{\frak{su}(2^m)} = i\ketbra{1_m}{1_m} - i2^{-m}I_{2^m},
\end{equation*}
and the Killing norm,
\begin{equation*}
    \norm{(i\ketbra{1_m}{1_m})_{\frak{su}(2^m)}}_K^2 = 2^{m+1} \norm{i\ketbra{1_m}{1_m} - i2^{-m}I_{2^m}}_F^2 = 2(2^m-1).
\end{equation*}
We then apply Theorem~\ref{them:compact_group_var} to derive,
\begin{equation}
\begin{aligned}
    \Var[\partial_H\ell(\rho,O)] &=  \frac{2(2^m-1)}{(2^{2m} - 1)^2}\sum_{j=1}^{r} \lambda_j^2\norm{O_{\frak{g}_j}}_F^2\|\rho_{\frak{g}_j}\|_F^2\\
    &= \frac{2(2^m-1)\norm{O_m}_F^2}{(2^{2m} - 1)^2}\sum_{j=1}^{r} \lambda_j^2 \tr(\Pi_j O_n)^2\tr(\Pi_j\rho_n)^2 \norm{(\rho_m)_{i\frak{su}(2^m)}}_F^2 \\
    &= \frac{2(2^m-1)(1-2^{-m})\norm{O_m}_F^2}{(2^{2m} - 1)^2} \sum_{j=1}^{r} \lambda_j^2\chi_j^2 \tr(\Pi_j\rho_n)^2,
\end{aligned}
\end{equation}
where we have taken $\rho_m = \ketbra{0_m}{0_m}$ as from previous sections. Now suppose  $\norm{O_m}_F^2\sim \Theta(1)$ and $m\sim\cO(\log(n))$. Since $\chi_j\geq 1$, by Cauchy-Schwarz inequality, one has,
\begin{equation*}
    \sum_{j=1}^{r} \lambda_j^2\chi_j^2 \tr(\Pi_j\rho_n)^2 \geq \sum_{j=1}^r \lambda_j^2 \tr(\Pi_j\rho_n)^2 \geq \frac{\tr(H\rho_n)^2}{r}.
\end{equation*}
Based on the assumption where $\tr(H\rho_n)^2$ scales as $\Omega(r/\Op{poly}(n))$, we can then complete the proof,
\begin{equation*}
    \Var[\partial_H\ell(\rho,O)] \geq \Omega\left(\frac{1}{n^2} \frac{r}{r\Op{poly}(n)}\right) \geq \Omega\left(1/\Op{poly}(n) \right).
\end{equation*}
\end{proof}

As we can see, the trainability of $\bm{\varphi}$ in the data embedding register is also dominated by the interactions between the data generator and the initial state $\rho_n$. Interestingly, the trainability of $\bm{\varphi}$, slightly different from the results of $\bm{\theta}$, is closely related to the expectation value of $\tr(H\rho_n)$. This has indicated the fact that a designed initial state $\rho_n$ from the QRENN model can magnify the `efficiency' of embedding information to the data processing register.  

We also perform numerical gradient statistics for the QRENN with different datasets embedded, as well as the rate of concentration for the networks regarding the transition towards the fully mixing phase in the represented unitary group by the growth of the number of slots in the model. To construct the data processing register, we choose the $R_Y$-$R_Z$ circuit template~\cite{Sim2019expressibility} with multiple layers to construct the processing unitary $W(\bm{\theta}_t)$, shown in Figure~\ref{fig:QRENN_circuit}(b). For each sampling experiment, the initial state is fixed to $\ket{0}^{\ox m+n}$ and the observable is chosen to be $O_m = Z^{\ox m}$. The objective function is defined in the trace form $\ell(\rho, O)$. The tunable parameters $\bm{\varphi}_t$ are attached within each data embedding register, making the embedding block a parameterized oracle.

\begin{figure}[t!]
    \centering
    \includegraphics[width=1.0\linewidth]{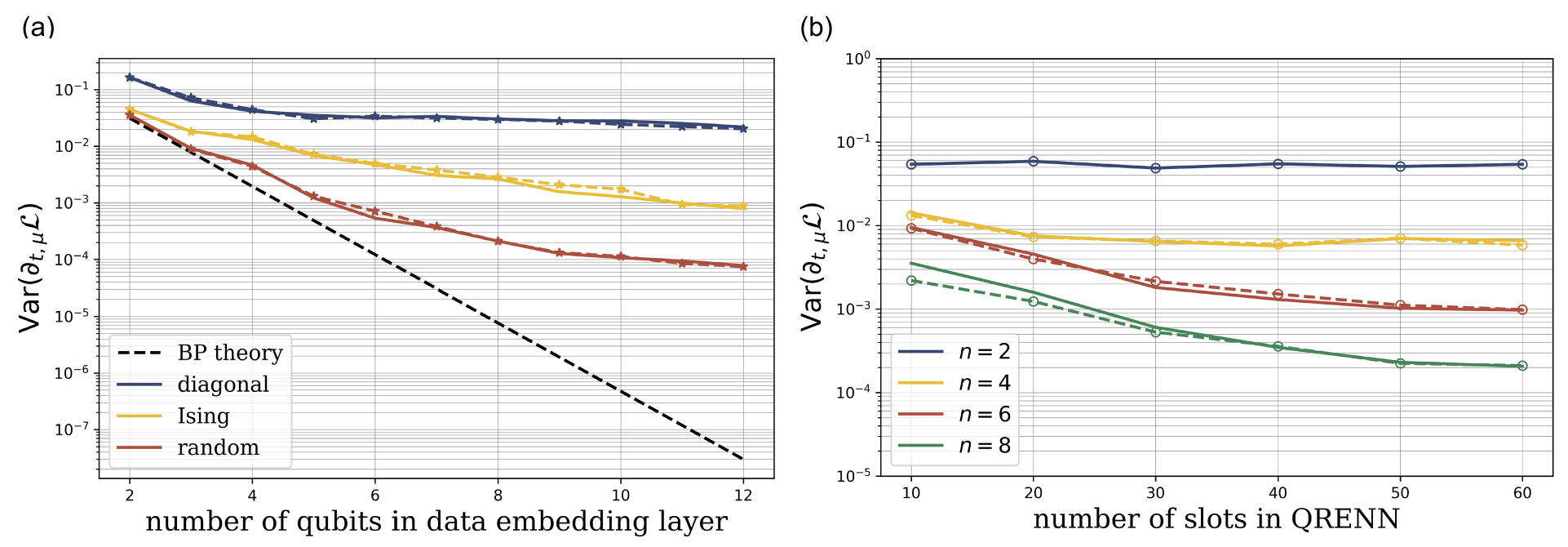}
    \caption{\textbf{Gradient statistics for QRENNs on various types of data embedded with respect to (a) the number of qubits in data embedding register and (b) the number of slots in the model regarding $\theta_{1,1}$ and $\varphi_1$.} In subfigure (a), the blue lines are derived from embedding random diagonal Hermitian data at each slot of the model, while the yellow and red lines represent repeatedly embedding one fixed random Ising model Hamiltonian and random Hermitian matrix, respectively, from each sampling iteration. The plot showcases our QRENN model experiences a polynomial decay in the gradient variance, diverging from BP; {In subfigure (b), our QRENN model is randomly embedded with one fixed Hermitian matrix for each \(n\)-qubit configuration, with each configuration distinguished by a unique color. This setup illustrates the rapid mixing towards \(e^{\liedqnn}\) within the polynomial slots of the network.}}
    \label{fig:gradient_test_dqnn_appendix}
\end{figure}

The gradient sampling experiments are performed based on  \textit{tensorcircuit} python quantum machine learning toolkits. For each experiment, both parameters $\bm{\theta}$ and $\bm{\varphi}$ are independently addressed, and the results are plotted in solid and dashed lines, respectively. We do sampling experiments with randomly initialized QRENNs where each tunable parameters are uniformly distributed in $[0,2\pi)$. The gradient statistics are illustrated in Figure~\ref{fig:gradient_test_dqnn_appendix}. Each point in the figure is derived by sampling $500$ randomized networks. From Figure~\ref{fig:gradient_test_dqnn_appendix} (a), the blue lines indicate the sampling results from embedding random diagonal unitary matrices $\Op{diag}(e^{i\phi_1}, e^{i\phi_2}, \cdots, e^{i\phi_{2^n}})$ at each slot of the QRENN model where $(\phi_j)_j$ is uniformly sampled from $[0,2\pi)$. For the other two colored curves, the random data are fixed for each number of qubits and repeatedly embedded into the networks throughout the slots. The yellow lines are derived from embedding random 1D Ising Hamiltonians defined as,
\begin{equation*}
    H_{\Op{Ising}} := \sum_{k=0}^{n-2} J_k Z_k Z_{k+1} + \sum_{k=0}^{n-1} h_k X_k,
\end{equation*}
where each $J_k, h_k$ are distributed uniformly in $[-1,1]$. The red lines represent the result by embedding the Haar random unitary. We can observe that for either cases of embedding commutative data or fixed data Hermitian matrices, the gradient statistics have showcased a divergence in the variance of gradient values from the exponential decay represented by the black dashed line. On the other hand, in Figure~\ref{fig:gradient_test_dqnn_appendix} (b), by fixing a random Hermitian data for each fixed number of qubits, we can observe a rapid convergence towards the concentration of variance values with respect to the number of slots in QRENN. This demonstrates the convergence to the second moment of the variance as a function of the network depth, and hence, clarifies the uses of the abstracted gradient.

\subsection{Proof of the main results on the gradient statistics of the total loss}\label{appendix_subsec:total_loss_gradient}
In this subsection, we demonstrate the lower bound on the gradient variance with respect to the total loss defined in Eq.~\eqref{eq:total_loss}. Here, we consider the QRENN model defined in Figure~\ref{fig:general_frame_dqnn} where the data $\bm{x}$ is carried by the Hamiltonian $H_t(\bm{x})$ at each slot of the network. In particular, let $\cT = \{(\bm{x}_q, y_q)\}_{q=1}^Q$ be the batch training set, and $\{M_{y_q}\}_q$ be the set of POVM. We consider all generators in $\{H_t(\bm{x}_q)\}_t$ for each distinct $q\in [Q]$ are commutative with each other. Focusing on solving quantum supervised learning tasks using the QRENN, we now recall the total loss defined in Eq.~\eqref{eq:total_loss},
\begin{equation}\label{eq:state_discrimination_loss}
    \cL(M_{y_q},\bm{x}_q;\bm{\theta}) = 1- \frac{1}{Q}\sum_{q = 1}^Q \tr(U(\bm{x}_q;\bm{\theta}) \rho_0 U(\bm{x}_q;\bm{\theta})^{\dagger} M_{y_q}).
\end{equation}
Notice that the subgroup $e^\liedqnn$ is data-dependent. Since the circuit is assumed to be sufficiently deep to form $\epsilon$-approximate $2$-design on $e^{\liedqnn}$ regarding any $\bm{x}_q$, we can replace the parameter derivative with respect to $\theta_{t,\mu}$ with the abstract gradient with respect to the Hermitian generator $\Omega = \Omega_{\mu} \ox I_{2^n}$, and derive,
\begin{equation}\label{eq:abs_gradient_total_loss}
    \partial_{\Omega} \cL = -\frac{1}{Q} \sum_{q=1}^Q \tr\left(U^{\dagger}_{g^-}(\bm{x}_q) i\rho_0 U_{g^-}(\bm{x}_q) [\Omega, U_{g^+}(\bm{x}_q) iM_{y_q} U^{\dagger}_{g^+}(\bm{x}_q)]\right).
\end{equation}
% Then the following lower bound on the gradient variance can be derived. 

% \begin{proposition}\label{prop:main_variance_lb_appendix}
%     Consider an ($m+n$)-qubit QRENN model with a set of Hermitian data generators $\cG = \{H_t(\bm{x}_q)\}_t$ embedded via Eq.~\eqref{eq:encoding_map} for each $q\in [Q]$, running on an input state $\rho$. %, and a tester set $\{M_{y_q}\}_q$. 
%     Assume the model circuit is sufficiently deep to form an $\epsilon$-approximate $2$-design on the corresponding dynamic subgroup. Then the total loss function defined in Eq.~\eqref{eq:total_loss} has 
%     \[
%         \EE_{\bm{\theta},\bm{\varphi}}\left[\frac{\partial \cL}{\partial \theta_{t,\mu}}\right] = 0,
%     \] 
%     for all $t\in [T]$ and $\mu\in [L]$. Moreover, denoting the derivative generator as $\Omega$, then
%     \begin{equation*}
%         \Var_{\bm{\theta}, \bm{\varphi}}\left[\frac{\partial \cL}{\partial \theta_{t,\mu}}\right] \geq \frac{1}{Q^2}\sum_{q}\sum_{\bm{\lambda}_q}   \frac{\norm{(M_{y_p})_{\frak{g}_{\bm{\lambda}_q}}}_F^2\norm{\rho_{\frak{g}_{\bm{\lambda}_q}}}_F^2 \norm{\Omega_{\frak{g}_{\bm{\lambda}_q}}}_K^2}{d^2_{\frak{g}_{\bm{\lambda}_q}}},
%     \end{equation*}
%     where for each $\bm{x}_q\in\cT$, the summation is taking over all distinct joint eigenspace $\frak{g}_{\bm{\lambda}_q}$ with respect to the corresponding DLA, and the center $\frak{c}_q$ does not contribute to the variance.
% \end{proposition}
% The above proposition indicates that the choices of the batch training sets can significantly determine the appearance of BP.
Before the formal proof of Proposition~\ref{prop:zero_mean_gradient_lb_grad_var}, we first claim and prove some useful lemmas.
\begin{lemma}\label{lem:adj_int_AoxB}
    Let $G$ be a compact and connected Lie group with Lie algebra $\frak{g}$. Suppose $V$ is a finite-dimensional inner product space $\phi:G\rightarrow \cU(V)$ is a unitary representation of $G$, and $U_g = \phi(g)$. In addition, $a,b\in \frak{g}, A = d\phi(a), B = d\phi(b)$. Then the following holds:
    \begin{equation*}
        \int_G U_g^{\ox 2} (A\ox B) (U_g^{\dagger})^{\ox 2} dg = \sum_{\alpha} \frac{\tr(A_{\frak{g}_{\alpha}}B_{\frak{g}_{\alpha}})}{d_{\frak{g}_\alpha}} K_{\frak{g}_{\alpha}} + A_{\frak{c}} \ox B_{\frak{c}},
    \end{equation*}
    where $A_{\frak{g}_{\alpha}}$ denotes the projection of $A$ onto $\frak{g}_{\alpha}$ and $K_{\frak{g}_{\alpha}}$ is the split Casimir in the subalgebra $\frak{g}_{\alpha}$.
\end{lemma}
\begin{proof}
    Since $\frak{g}$ is reductive, let $\frak{g} = \frak{c} \oplus \bigoplus_\alpha \frak{g}_{\alpha}$, the adjoint representation can break into a direct sum of irreducible representations regarding $\frak{g}_\alpha$ and the center $\frak{c}$. Since $G$ is connected, the $\frak{g}_{\alpha}$ also correspond to non-isomorphic simple $G$-modules. By the Schur lemma, we have,
    \begin{equation*}
    \begin{aligned}
        \int_G U_g^{\ox 2} (A\ox B) (U_g^{\dagger})^{\ox 2} dg &= \left(\sum_{\alpha} \int_G \Adj{U_{\frak{g}_{\alpha}}}{(A_{\frak{g}_{\alpha}})} \ox \Adj{U_{\frak{g}_{\alpha}}}{(B_{\frak{g}_{\alpha}})} dg_{\alpha}\right)+ A_{\frak{c}} \ox B_{\frak{c}}\\
        &= \sum_{\alpha} \frac{\tr(A_{\frak{g}_{\alpha}}B_{\frak{g}_{\alpha}})}{d_{\frak{g}_\alpha}} K_{\frak{g}_{\alpha}} + A_{\frak{c}} \ox B_{\frak{c}},
    \end{aligned}
    \end{equation*}
    where the last equality holds due to the follows: let $a,b\in \frak{g}$ such that $a = \sum_j a_j e_j$ and $b = \sum_j b_j e_j$ where $\{e_j\}$ forms a basis for $\frak{g}$ and $E_j = d\phi(e_j)$ is skew-Hermitian forming a basis for $d\phi(\frak{g})$. Then the integration,
    \begin{equation*}
    \begin{aligned}
        \int_G \Adj{U_g}{(A)} \ox \Adj{U_g}{(B)} dg &= \frac{1}{I_{\phi}}\sum_{ii'jj'}a_i b_{i'}\int_G[\Op{Ad}_g]_{ij}[\Op{Ad}_g]_{i'j'}dg (E_j\ox E_{j'})\\
        &= \frac{1}{I_{\phi}d_{\frak{g}}}\sum_{ii'jj'} a_i b_{i'} \delta_{ii'} \delta_{jj'} E_j \ox E_{j'} = \frac{1}{I_{\phi}d_{\frak{g}}}\sum_{i} a_i b_{i} \sum_{j} E_j \ox E_{j}.
    \end{aligned}
    \end{equation*}
    The second equality holds due to the Schur orthogonality and by definition $\sum_j E_j \ox E_j = I_{\phi} K$ for $K$ the split Casimir operator.
\end{proof}

\begin{proposition}\label{prop:exp_gradient_appendix}
    Given the total loss defined in Eq.~\eqref{eq:state_discrimination_loss}, then for any $M_{y_q}, \rho_0\in \frak{gl}(V)$, we have,
    \begin{equation*}
        \EE_{g^{\pm} \sim \mu^{\ox 2}}[\partial_{\Omega} \cL] = 0,
    \end{equation*}
    where $\EE_{g^{\pm} \sim \mu^{\ox 2}}[\; \cdot\;] \coloneqq \EE_{g^{\pm}(\bm{x}_1) \sim \mu^{\ox 2}(\bm{x}_1)}\circ \cdots \circ \EE_{g^{\pm}(\bm{x}_Q) \sim \mu^{\ox 2}(\bm{x}_Q)}[\; \cdot\;]$.
\end{proposition}
\begin{proof}
Based on the expression of Eq.~\eqref{eq:abs_gradient_total_loss} and the assumption of the circuit. We can directly compute the expectation value of the gradient by taking $g^{\pm}(\bm{x}_q) \sim \mu^{\ox 2}(\bm{x}_q)$ for any $\bm{x}_q$ in the batch training set. Due to the linearity of expectation, we have,
\begin{equation*}
\begin{aligned}
    \EE_{g^{\pm} \sim \mu^{\ox 2}}[\partial_{\Omega} \cL] &= -\frac{1}{Q} \sum_{q=1}^Q \EE_{g^{\pm} \sim \mu^{\ox 2}}\left[\tr\left(U^{\dagger}_{g^-}(\bm{x}_q) i\rho_0 U_{g^-}(\bm{x}_q) [\Omega, U_{g^+}(\bm{x}_q) iM_{y_q} U^{\dagger}_{g^+}(\bm{x}_q)]\right)\right].
\end{aligned}
\end{equation*}
We suppose $iM_{y_q}$ are some fixed operators in $\frak{gl}(V)$ that are pre-determined before training, therefore, satisfy the conditions of Lemma C.1 from~\cite{Fontana2024characterizing}; $i\rho_0 \in \frak{gl}(V)$ automatically. Besides, since each data $\bm{x}_q$ is independent and therefore, the order of integrations can be switched. As a result, each term in the summation vanishes, and we have $\EE_{g^{\pm} \sim \mu^{\ox 2}}[\partial_{\Omega} \cL] = 0$.
\end{proof}

\vspace{4pt}
\noindent\textbf{Proof of Proposition~\ref{prop:zero_mean_gradient_lb_grad_var} --.} We are now ready to prove the proposition. Firstly the expectation value of $\frac{\partial\cL}{\partial \theta_{t,\mu}}$ due to Proposition~\ref{prop:exp_gradient_appendix}. Thus, in practice, the variance of the abstract gradient values can be computed up to the non-vanishing second moment:
\begin{equation*}
    \Var[\partial_{\Omega} \cL] = \EE_{g^{\pm} \sim \mu^{\ox 2}}[(\partial_{\Omega} \cL)^2].
\end{equation*}
Let us first compute $(\partial_{\Omega} \cL)^2$ via direct calculations:
\begin{equation*}
\begin{aligned}
    Q^2(\partial_{\Omega} \cL)^2\Big|_{g^{\pm}} &= \sum_{pq} \tr(U_{g^-}(\bm{x}_p)^{\dagger} (i\rho_0) U_{g^-}(\bm{x}_p)[\Omega, U_{g^+}(\bm{x}_p)(iM_{y_p}) U_{g^+}(\bm{x}_p)^{\dagger}]) \\
    &\qquad \times \tr(U_{g^-}(\bm{x}_q)^{\dagger} (i\rho_0) U_{g^-}(\bm{x}_q)[\Omega, U_{g^+}(\bm{x}_q)(iM_{y_q}) U_{g^+}(\bm{x}_q)^{\dagger}])\\
    &= \sum_{pq} \tr\Big\{(i\rho_0)^{\ox 2} U_{g^-}(\bm{x}_p)\ox U_{g^-}(\bm{x}_q)\\
    &\qquad\quad\times 
    [\Omega, U_{g^+}(\bm{x}_p)(iM_{y_p}) U_{g^+}(\bm{x}_p)^{\dagger}] \ox [\Omega, U_{g^+}(\bm{x}_q)(iM_{y_q}) U_{g^+}(\bm{x}_q)^{\dagger}]\\
    &\qquad\quad\times U_{g^-}(\bm{x}_p)^{\dagger}\ox U_{g^-}(\bm{x}_q)^{\dagger}\Big\}.
\end{aligned}
\end{equation*}
For any fixed data vectors $\bm{x}_p$ and $\bm{x}_q$, the construction of QRENN ensures the existence of $V_{pq}$ such that, $U_{g^+}(\bm{x}_q) = U_{g^+}(\bm{x}_p)V_{pq,+}$ and $U_{g^-}(\bm{x}_q) = V_{pq,-}U_{g^-}(\bm{x}_p)$. Substitute it into the expression to derive,
\begin{equation*}
\begin{aligned}
    Q^2(\partial_{\Omega} \cL)^2\Big|_{g^{\pm}} &= \sum_{pq} \tr\Big\{(I\ox V^{\dagger}_{pq,-})(i\rho_0)^{\ox 2}(I\ox V_{pq,-}) U_{g^-}(\bm{x}_p)^{\ox 2}\\
    &\qquad\quad\times 
    [\Omega, U_{g^+}(\bm{x}_p)(iM_{y_p}) U_{g^+}(\bm{x}_p)^{\dagger}] \ox [\Omega, U_{g^+}(\bm{x}_q)(iM_{y_q}) U_{g^+}(\bm{x}_q)^{\dagger}]\\
    &\qquad\quad\times (U_{g^-}(\bm{x}_p)^{\dagger})^{\ox 2}\Big\}\\
    &= \sum_{pq} \tr\Big\{(I\ox V^{\dagger}_{pq,-})(i\rho_0)^{\ox 2}(I\ox V_{pq,-}) U_{g^-}(\bm{x}_p)^{\ox 2}\\
    &\qquad\quad\times 
    [\Omega, U_{g^+}(\bm{x}_p)(iM_{y_p}) U_{g^+}(\bm{x}_p)^{\dagger}] \ox [\Omega, U_{g^+}(\bm{x}_p)V_{pq,+}(iM_{y_q}) V^{\dagger}_{pq,+}U_{g^+}(\bm{x}_p)^{\dagger}]\\
    &\qquad\quad\times (U_{g^-}(\bm{x}_p)^{\dagger})^{\ox 2}\Big\}.
\end{aligned}
\end{equation*}
We then aim to compute the integration over the group $e^{\liedqnn}$ with respect to each $\bm{x}_p$. For the $p$-th data embedded in, let us ignore the trace, $\rho$, and expand the commutators into four terms, followed by the integration over the group. For convenience, we write $U_{g^{\pm}}(\bm{x}_p) = U_{g^{\pm}_p}$ using $g^{\pm}_p$ to distinguish different subgroup generated by $\bm{x}_p$. By denoting $X_{+,p} = \int_G U_{g^+_p}^{\ox 2} (iM_{y_p} \ox i\widetilde{M}_{y_q})(U_{g^+_p}^{\dagger})^{\ox 2} dg^+_p$ where $\widetilde{M}_{y_q} = V_{pq,+}(M_{y_q}) V^{\dagger}_{pq,+}$, we now apply Lie structure proposition where each $\liedqnn$ with respect to $\bm{x}_p$ being decomposed into simple ideals, denoted as  $\frak{g}_{\bm{\lambda}_p}$, using Lemma~\ref{lem:adj_int_AoxB}, we can first compute the inner integration for $X_{+,p}$ as,
\begin{equation*}
    X_{+,p} = \sum_{\bm{\lambda}} \frac{\tr((M_{y_p})_{\frak{g}_{\bm{\lambda}_p}}(\widetilde{M}_{y_q})_{\frak{g}_{\bm{\lambda}_p}})}{d_{\frak{g}_{\bm{\lambda}_p}}} K_{\frak{g}_{\bm{\lambda}_p}} + (M_{y_p})_{\frak{c}_p}\ox(\widetilde{M}_{y_q})_{\frak{c}_p}
\end{equation*}
Notice that the center of $\liedqnn$ would contribute nothing to the variance as $[\Omega, (M_{y_p})_{\frak{c}}] = [\Omega, (\widetilde{M}_{y_q})_{\frak{c}}] = 0$. Therefore, by ignoring the center terms, the only contribution comes from the four terms regarding the simple ideals. Let $\{E_j\}_j$ form an orthonormal basis with respect to the Hilbert-Schmidt norm of skew-Hermitian matrices for $\frak{g}_{\bm{\lambda}_p}$. We then rearrange the ordering of the terms and take the integral over $U_{g^-_p}$ to derive,
\begin{equation*}
\begin{aligned}
    &\frac{1}{I_{\phi}}\sum_{\bm{\lambda}_p} \frac{\tr((M_{y_p})_{\frak{g}_{\bm{\lambda}_p}}(\widetilde{M}_{y_q})_{\frak{g}_{\bm{\lambda}_p}})}{d_{\frak{g}_{\bm{\lambda}_p}}}\sum_{j=1}^{d_{\frak{g}_{\bm{\lambda}_p}}}\int_G U_{g_p} [\Omega, E_j]U_{g_p}^{\dagger} \ox U_{g_p} [\Omega,E_j]  U_{g_p}^{\dagger} \ dg_p \\
    &= \sum_{\bm{\lambda}_p} \frac{\tr((M_{y_p})_{\frak{g}_{\bm{\lambda}_p}}(\widetilde{M}_{y_q})_{\frak{g}_{\bm{\lambda}_p}}) \norm{\Omega_{\frak{g}_{\bm{\lambda}_p}}}_K^2}{d^2_{\frak{g}_{\bm{\lambda}_p}}} K_{\frak{g}_{\bm{\lambda}_p}}.
\end{aligned}
\end{equation*}
The results are still a summation over $\bm{\lambda}_p$ due to the closeness of the Lie bracket $[\Omega,E_j] \in \frak{g}_{\bm{\lambda}_p}$ on each simple ideal. At last, we reintroduce the $\rho$ terms in the original expression and use the property that, for any $A=d\phi(a), B = d\phi(b)$ where $a,b \in \frak{g}_{\bm{\lambda}}$ with  $d\phi: \frak{g}_{\bm{\lambda}} \rightarrow\frak{u}(V)$, we have,
\begin{equation}
\begin{aligned}
    &\tr((A\ox B)K_{\frak{g}_{\bm{\lambda}}}) = I_{\phi}^{-1}\sum_j \tr(AE_j)\tr(BE_j) = \tr(A_{\frak{g}_{\bm{\lambda}}}B_{\frak{g}_{\bm{\lambda}}}).
\end{aligned}
\end{equation}
The representation index $I_{\phi}$ is canceled out due to the fact that $\tr(E_i E_j) = I_{\phi} \delta_{ij}$. We therefore derive,
\begin{equation*}
\begin{aligned}
    \EE_{g^{\pm} \sim \mu^{\ox 2}}[(\partial_{\Omega}\cL)^2] &= \frac{1}{Q^2}\sum_{pq}\sum_{\bm{\lambda}_p} \frac{\tr((M_{y_p})_{\frak{g}_{\bm{\lambda}_p}}(\widetilde{M}_{y_q})_{\frak{g}_{\bm{\lambda}_p}}) \norm{\Omega_{\frak{g}_{\bm{\lambda}_p}}}_K^2}{d^2_{\frak{g}_{\bm{\lambda}_p}}} \tr((I\ox V_{pq,-})(i\rho)^{\ox 2}(I\ox V^{\dagger}_{pq,-})K_{\frak{g}_{\bm{\lambda}_p}})\\
    &= \frac{1}{Q^2}\sum_{pq}\sum_{\bm{\lambda}_p} \frac{\tr((M_{y_p})_{\frak{g}_{\bm{\lambda}_p}}(\widetilde{M}_{y_q})_{\frak{g}_{\bm{\lambda}_p}}) \tr(\rho_{\frak{g}_{\bm{\lambda}_p}} (V_{pq,-}\rho V^{\dagger}_{pq,-})_{\frak{g}_{\bm{\lambda}_p}}) \norm{\Omega_{\frak{g}_{\bm{\lambda}_p}}}_K^2}{d^2_{\frak{g}_{\bm{\lambda}_p}}}.
\end{aligned}
\end{equation*}
Now, since the $\tr(AB) \geq 0$ for any $A,B\geq 0$, we can ignore the cross terms of $p\neq q$ and lower bound the variance as,
\begin{equation*}
\begin{aligned}
    \EE_{g^{\pm} \sim \mu^{\ox 2}}[(\partial_{\Omega}\cL)^2] &\geq \frac{1}{Q^2}\sum_{p=q} \sum_{\bm{\lambda}_p} \frac{\tr((M_{y_p})_{\frak{g}_{\bm{\lambda}_p}}(\widetilde{M}_{y_q})_{\frak{g}_{\bm{\lambda}_p}}) \tr(\rho_{\frak{g}_{\bm{\lambda}_p}} (V_{pq,-}\rho V^{\dagger}_{pq,-})_{\frak{g}_{\bm{\lambda}_p}}) \norm{\Omega_{\frak{g}_{\bm{\lambda}_p}}}_K^2}{d^2_{\frak{g}_{\bm{\lambda}_p}}}\\
    &= \frac{1}{Q^2}\sum_{p}\sum_{\bm{\lambda}_p}   \frac{\norm{(M_{y_p})_{\frak{g}_{\bm{\lambda}_p}}}_F^2\norm{\rho_{\frak{g}_{\bm{\lambda}_p}}}_F^2 \norm{\Omega_{\frak{g}_{\bm{\lambda}_p}}}_K^2}{d^2_{\frak{g}_{\bm{\lambda}_p}}},
\end{aligned}
\end{equation*}
where the second equality holds due to the fact that when $p=q$, $V_{p=q, \pm} = I$ and we have $\widetilde{M}_{y_q} = M_{y_p}$. 

\subsection{Absence of barren plateaus with local labeling measurements}\label{appendix_subsec:proof_of_them_no_BP}
In this subsection, we apply the results from the lower bound of the gradient variance in Proposition~\ref{prop:zero_mean_gradient_lb_grad_var} to investigate the gradient statistics for the classification problems using QRENN. Here, we consider learning features of quantum Hamiltonians by setting $\cT:=\{(X_q, y_q)\}_q$
and treat $X_q$ as the data embedded into the network. In this case, we set $H_t(X_q) = X_q$ for $t\in [T]$ throughout the entire circuit, and specifically choose $M_{y_q}$ that are locally acting on the processing registers at the end. Then, the subalgebras are completely determined by the projections onto the eigenspaces of each $X_q$, and $\bm{\lambda}$ becomes a scale representing the distinct eigenvalues of $X_q$.
% \begin{theorem}
%     Under the assumptions in Proposition~\ref{prop:zero_mean_gradient_lb_grad_var}, let $m$ scale as $\cO(\log(n))$, input state $\rho = \ketbra{\psi}{\psi}\ox\rho_n$ for some $m$-qubit pure state $\ket{\psi}$, and $n$-qubit state $\rho_n$. Suppose that each processing Hermitian $\Omega_{\mu} \ox I_{2^n}$ and labeling measurement $M_{y_q} \ox I_{2^n}$ acting locally on the processing register, and $\norm{\Omega_{\mu}}^2_F, \norm{M_{y_q}}^2_F$ scales as $\Omega(1)$ for $q\in[Q]$. If the batch size $Q$ scales as $\cO(\Op{poly}(n))$ and at least one data $X_q$ offers $\Omega(1/\Op{poly}(n))$ scaled $R_{X_{q}}(\rho_n)$. Then,
%     $$\Var_{\bm{\theta}, \bm{\varphi}}[\partial_{t,\mu} \cL] \geq \Omega(1/\Op{poly}(n)).$$
% \end{theorem} 
\vspace{4pt}
\noindent\textbf{Proof of Proposition~\ref{prop:variance_dqnn_On_In} --.}
Recalling the results from the lower bound proposition, we have,
\begin{equation*}
\begin{aligned}
    \EE_{g^{\pm} \sim \mu^{\ox 2}}[(\partial_{\Omega}\cL)^2] 
    \geq \frac{1}{Q^2}\sum_{q}\sum_{\bm{\lambda}_q}\frac{\norm{(M_{y_q})_{\frak{g}_{\bm{\lambda}_q}}}_F^2\norm{\rho_{\frak{g}_{\bm{\lambda}_q}}}_F^2 \norm{\Omega_{\frak{g}_{\bm{\lambda}_q}}}_K^2}{d^2_{\frak{g}_{\bm{\lambda}_q}}}.
\end{aligned}
\end{equation*}
Since the labelling POVMs and the gradient operator are locally acting on the data processing register. Denote $\Omega = \Omega_{\mu} \ox I_{2^n}$, $M_{y_q} = M_{y_q} \ox I_{2^n}$ and $\rho = \rho_m \ox \rho_n$. 
By applying Eq.~\eqref{eq:exact_var_trace_form_loss_grad} in the proof of Proposition~\ref{prop:variance_data_processing_layer} to the inner summation with respect to the Lie algebraic structure regarding the $q$-th data $X_q$, the lower bound can be further reduced to,
\begin{equation*}
\begin{aligned}
    \EE_{g^{\pm} \sim \mu^{\ox 2}}[(\partial_{\Omega}\cL)^2] 
    &\geq \frac{1}{Q^2}\sum_{q}\sum_{\bm{\lambda}_{q}}\frac{\norm{(M_{y_{q}})_{\frak{g}_{\bm{\lambda}_{q}}}}_F^2\norm{\rho_{\frak{g}_{\bm{\lambda}_{q}}}}_F^2 \norm{\Omega_{\frak{g}_{\bm{\lambda}_{q}}}}_K^2}{d^2_{\frak{g}_{\bm{\lambda}_{q}}}}\\
    &\geq \frac{1}{Q^2}\frac{2^{m+1}\norm{\Omega_{\mu}}_F^2}{(2^{2m} - 1)^2}\sum_{q} \norm{M_{y_{q}}}_F^2 R^2_{X_{q}}(\rho_n).
\end{aligned}
\end{equation*}
Now, suppose $\norm{\Omega_{\mu}}^2_F$ and $\norm{M_{y_q}}^2_F$ scales constantly. Since at least one data point has a sufficiently large eigenspace overlap. Let $q'\in [Q]$ denote those data in $\cT$ contributing polynomially lower-bounded $R_{X_{q'}}(\rho_n)$. With polynomially growth of batch training size $Q$ and taking $m\sim \cO(\log(n))$, we derive,
\begin{equation*}
\begin{aligned}
    \EE_{g^{\pm} \sim \mu^{\ox 2}}[(\partial_{\Omega}\cL)^2] 
    \geq \Omega\left(\frac{2n}{(n^2-1)^2}\sum_{q'}R^2_{X_{q'}}(\rho_n)\right) \geq \Omega\left(1/\Op{poly}(n)\right),
\end{aligned}
\end{equation*}
and hence the proof.

\section{Joint eigenspace overlap for supervised learning}\label{app:joint_eig_overlap}

In the application of supervised learning on quantum Hamiltonian classification, we have fixed the embedded data the same throughout the entire QRENN model. In order to satisfy the preconditions of our trainability theorem for QRENN. An appropriate probe state $\rho_n$ should be selected. Start with the Pauli and involutory feature sets. Notice that for any $n$-qubit involutory Hermitian matrices $P$, there are only two projections $\{\Pi_+, \Pi_-\}$ which characterizes its joint eigenspaces,
\begin{equation*}
    \Pi_+ = \frac{I_{2^n} + P}{2}; \quad \Pi_- = \frac{I_{2^n} - P}{2}.
\end{equation*}
From the experimental setup, we have fixed the input probe state $\rho_n = \ketbra{+}{+}^{\ox n}$. The corresponding (joint) eigenspace overlap can be determined,
\begin{equation*}
\begin{aligned}
    R^2_P(\rho_n) = \tr(\Pi_+ \rho_n)^2 + \tr(\Pi_- \rho_n)^2 = \frac{1}{2}(1 + \tr(P\rho_n)^2)
\end{aligned}
\end{equation*}
The average eigenspace overlap among all involutory matrices is derived as,
\begin{equation*}
\begin{aligned}
    \EE_{P}[R_P^2(\rho_n)] &= \frac{1}{2} + \EE_{P}[\tr(P\rho_n)^2] \geq \frac{1}{2}.
\end{aligned}
\end{equation*}
In particular, any $n$-qubit Pauli matrix is automatically an involutory matrix. The above construction then leads to a constantly lower bounded joint eigenspace overlap. On the other hand, for the diagonal feature set, we have constructed a depolarizing state $\rho_n = \frac{1}{2}\ketbra{+}{+}^{\otimes n} + \frac{1}{2}I_{2^n} /2^n$. Thus, the overlap,
\begin{equation*}
    R^2_{D}(\rho_n) = \sum_j \tr(\Pi_j (\frac{1}{2}\ketbra{+}{+}^{\otimes n} + \frac{1}{2}I_{2^n} / 2^n))^2 = \frac{1}{4}\sum_j \tr(\Pi_j\ketbra{+}{+}^{\otimes n}) + \chi_j / 2^n)^2.
\end{equation*}
Notice that we consider the diagonal Hermitian matrices $D$'s with each diagonal entry uniformly distributed from $[0,\pi)$, the probability of getting $\tr(\Pi_j\ketbra{+}{+}^{\otimes n}) = 0$ for all $j$ is almost zero. This can then lead to an average overlap lower bounded,
\begin{equation*}
    \EE_{D}[R^2_{D}(\rho_n)] = \frac{1}{4}(1 + \chi_j /2^n)^2 \geq \frac{1}{4}.
\end{equation*}

\begin{figure}[t!]
    \centering
    \includegraphics[width=1\linewidth]{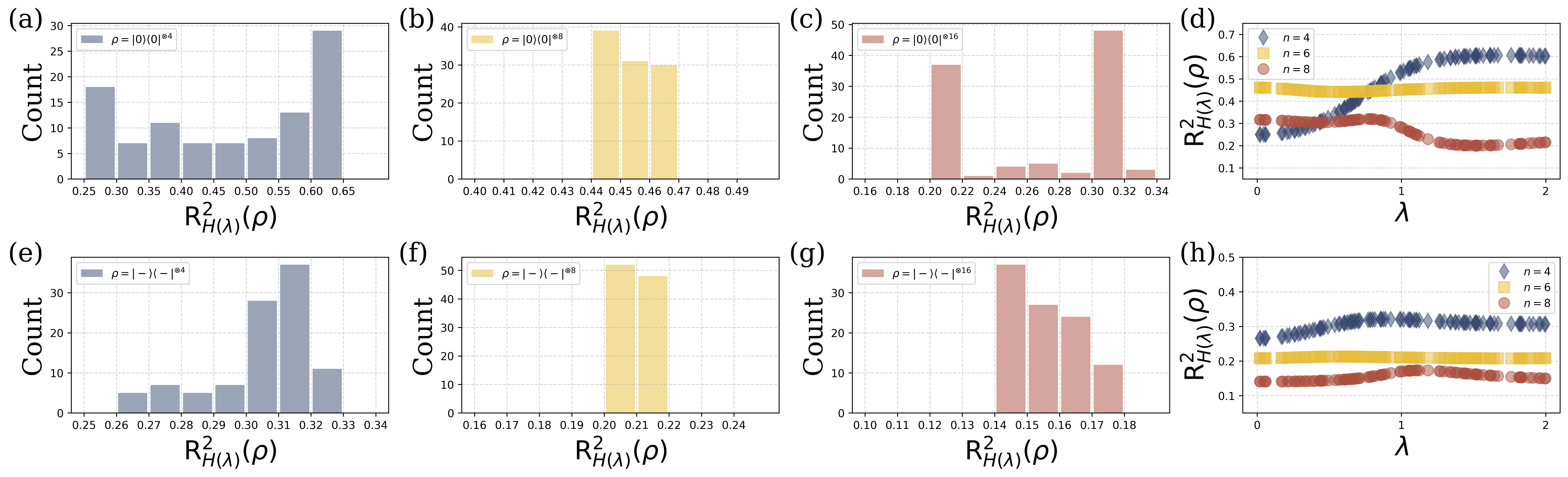}
    \caption{\textbf{Numerical demonstration of the trainability of the QRENN for detecting the SPT phase in a one-dimensional cluster-Ising model with periodic boundary conditions.} (a)-(d) correspond to the initial state $\rho_n = \ketbra{0}{0}^{\otimes n}$ and (e)-(h) show results for $\rho_n = \ketbra{-}{-}^{\otimes n}$. 
    In (a)-(c) and (e)-(g), we present histograms of $R_{H(\lambda)}^2(\rho_n)$ for $n=2,4,$ and $8$. 
    In(d) and (h), we  display the dependence of $R_{H(\lambda)}^2(\rho_n)$ on both the Hamiltonian parameter $\lambda$ and the system size ($n=4,6,8$). and observe that $R_{H(\lambda)}^2(\rho_n)$ decreases polynomially with $n$, confirming that the QRENN model remains trainable for detecting the SPT phase as $n$ grows.} \label{fig:trainability_of_detecting_spt_phase_for_minus}

\end{figure}

To demonstrate that our QRENN model can be effectively trained to detect the SPT phase, we perform a series of numerical experiments, as shown in Figure~\ref{fig:trainability_of_detecting_spt_phase_for_minus}, on a one-dimensional cluster-Ising Hamiltonian with periodic boundary conditions defined in Eq.~\eqref{cluster-Ising model with periodic boundary}. We observe that as the system size $n$ increases, $R_{H(\lambda)}^2(\rho)$ decreases polynomially, indicating that while the overlap of $\rho$ with the joint eigenspaces is large with $n$ increases.  We show the QRENN model remains trainable for detecting the SPT phase throughout the range $[4,8]$.

Specifically, for $\rho = \ketbra{0}{0}^{\otimes n}$  in (a)-(d), subfigures (a)-(c) show histograms of $R_{H(\lambda)}^2(\rho)$ for different system sizes $n = 2, 4,$ and $8$, while (d) displays how $R_{H(\lambda)}^2(\rho)$ varies with 
parameter $\lambda$ uniformly distributed from $0$ to $2$. Similarly, for $\rho = \ketbra{-}{-}^{\otimes n}$ in (e)-(h), subfigures (e)-(g) illustrate the corresponding histograms, and (h) captures the dependence on $\lambda$ for $n = 4, 6,$ and $8$. Notably, Figure~\ref{fig:trainability_of_detecting_spt_phase_for_minus}(h) reveals that $R_{H(\lambda)}^2(\ketbra{-}{-}^{\otimes n})$ is nearly independent of $\lambda$, which shows the robustness of this initial state $\rho = \ketbra{-}{-}^{\otimes n}$ in maintaining distinguishable overlaps across the entire parameter range. Overall, these results confirm that our QRENN approach remains trainable for detecting SPT phases for this many-body system in larger system sizes.

\end{document}